\documentclass{article}
\usepackage{amsmath}
\usepackage{wright}
\usepackage{tikz}
\usetikzlibrary{matrix}

\newtheorem*{theorem*}{Theorem}

\newcommand{\mixed}{\textup{\textbf{Mix}}}
\newcommand{\gps}{\calA_{\mathrm{GPS}}}
\newcommand{\gkkt}{\calA_{\mathrm{GKKT}}}
\newcommand{\hayashi}{\calA_{\mathrm{Hayashi}}}

\newcommand{\symsep}{\mathrm{SoS}}
\newcommand{\purifychan}{\Phi_{\mathrm{Purify}}}
\newcommand{\swap}{\mathrm{SWAP}}

\newcommand{\reg}[1]{\mathsf{#1}}
\newcommand{\schur}{U_{\mathrm{Schur}}}

\title{Mixed state tomography reduces to pure state tomography}
\author{Angelos Pelecanos\thanks{UC Berkeley. \texttt{\{apelecan,jspilecki,ewin,jswright\}@berkeley.edu}}  \and Jack Spilecki\footnotemark[1] \and Ewin Tang\footnotemark[1]\and John Wright\footnotemark[1]}
\date{}

\begin{document}

\maketitle

\begin{abstract}
    A longstanding belief in quantum tomography is that estimating a mixed state is far harder than estimating a pure state.
    This is borne out in the mathematics, where mixed state algorithms have always required more sophisticated techniques to design and analyze than pure state algorithms.
    We present a new approach to tomography demonstrating that, contrary to this belief, state-of-the-art mixed state tomography follows easily and naturally from pure state algorithms.
    
    We analyze the following strategy: given $n$ copies of an unknown state $\rho$, convert them into copies of a purification $\ket{\rho}$; run a pure state tomography algorithm to produce an estimate of $\ket{\rho}$; and output the resulting estimate of $\rho$.
    The purification subroutine was recently discovered via the ``acorn trick'' of Tang, Wright, and Zhandry~\cite{TWZ25}.
    With this strategy, we obtain the first tomography algorithm which is sample-optimal in all parameters.
    For a rank-$r$ $d$-dimensional state, it uses $n = O((rd + \log(1/\delta))/\eps)$ samples to output an estimate which is $\eps$-close in fidelity with probability at least $1-\delta$.
    This algorithm also uses $\poly(n)$ gates, making it the first gate-efficient tomography algorithm which is sample-optimal even in terms of the dimension $d$ alone.
    Moreover, with this method we recover essentially all results on mixed state tomography, including its applications to tomography with limited entanglement, classical shadows, and quantum metrology.
    Our proofs are simple, closing the gap in conceptual difficulty between mixed and pure tomography.
    Our results also clarify the role of entangled measurement in mixed state tomography: the only step of the algorithm which requires entanglement across copies is the purification step, suggesting that, for tomography, the reason entanglement is useful is for \emph{consistent purification}.
\end{abstract}

\newpage

\hypersetup{linktocpage}
\tableofcontents
\thispagestyle{empty}

\newpage

\section{Introduction}
Given $n$ copies of a mixed state $\rho \in \C^{d \times d}$, \emph{mixed state tomography} refers to the task of producing an estimate of $\rho$.
\emph{Pure state tomography} refers to the special case when $\rho$ is promised to be a pure state.
It has long been believed that pure state tomography is significantly easier than mixed state tomography.
As Holevo~\cite[Section 4.11]{Hol11} puts it:
\begin{center}
\textit{The full model clearly displays another feature of quantum estimation problem:\\the complexity sharply increases with passage from pure to mixed states.}
\end{center}
This sharp increase in complexity appears not just in concrete resource requirements, but also in abstract conceptual difficulty.
Pure state tomography can be solved using $n = \Theta(d)$ copies~\cite{Hay98,SSW25} with measurements which are  unentangled across the $n$ copies of $\rho$,  computationally efficient, and conceptually simple~\cite{KRT14,GKKT20}.
On the other hand, optimal mixed state tomography requires $n = \Theta(d^2)$ copies~\cite{HHJ+16,OW16}, and the measurements which achieve this \emph{must} be entangled~\cite{CHL+23,CLL24a}, are not known to be computationally efficient, and are conceptually difficult, making use of heavy representation theory~\cite{HHJ+16,OW16,PSW25}.
Put together, this state of affairs suggests that pure and mixed state tomography could not be more hopelessly dissimilar.

We overturn this intuition by showing that mixed state tomography actually \emph{reduces} to pure state tomography.
To do so, we make use of the following recent result of Tang, Wright, and Zhandry~\cite{TWZ25}.
\begin{theorem}[Efficiently generating random purifications]\label{thm:acorn}
    Let $n\geq 1$, $d \geq 1$, and $r \leq d$ be integers. 
    There is a quantum channel $\purifychan^{d, r}(\cdot)$ which acts as follows.
    Given $n$ copies of a rank-$r$ mixed state $\rho \in \C^{d \times d}$,
    \begin{equation*}
        \purifychan^{d, r}(\rho^{\otimes n}) = \E_{\ket{\brho}} \ketbra{\brho}^{\otimes n},
    \end{equation*}
    where the expectation is over a uniformly random purification\footnote{
        When we refer to a uniformly random purification, we are referring to the distribution over purifications which is unitarily invariant in the purification register.
        For example, one can generate a sample from this distribution by taking a fixed purification $\ket{\rho} \in \C^d \otimes \C^r$ and applying a Haar-random unitary $\bU$ to the second register.
    } $\ket{\brho} \in \C^d \otimes \C^r$ of $\rho$.
    In addition, $\purifychan^{d,r}(\cdot)$ can be implemented to $\delta$ error in diamond distance in time $\poly(n, \log(d), \log(1/\delta))$.
\end{theorem}
\noindent
This result is an example of their ``acorn trick'',
in which copies of a resource (the mixed state $\rho$)
are randomly ``lifted'' to copies of a stronger resource (the purification $\ket{\brho}$)
in a manner which is consistent across all $n$ copies.
Similar results have previously appeared in the quantum property testing literature in the works of \cite[Theorem 35]{SW22} and especially \cite{CWZ24}. These works showed that having access to purifications does not help solve property testing tasks, and used this to convert lower bounds for mixed state property testing problems into lower bounds for purified property testing problems. 
Here, we do the opposite: we will use \Cref{thm:acorn} to convert pure state learning \emph{upper} bounds into mixed state learning upper bounds.

In particular, we advocate for constructing mixed state tomography algorithms in the following manner.

{
\floatstyle{boxed} 
\restylefloat{figure}
\begin{figure}[htbp]
Given $n$ copies of $\rho$:
\begin{enumerate}
    \item Apply $\purifychan^{d, r}$ to produce $n$ copies of a random purification $\ket{\brho} \in \C^d \otimes \C^r$.
    \item Run an off-the-shelf pure state tomography algorithm $\calA$ to learn an estimate $\widehat{\bsigma}$ of $\ketbra{\brho}$.
    \item Convert $\widehat{\bsigma}$ to an estimate $\widehat{\brho} = \tr_{\reg{2}}(\widehat{\bsigma})$ of $\rho$. Output $\widehat{\brho}$.
\end{enumerate}
\caption{A generic reduction from mixed state tomography to pure state tomography. We refer to the resulting mixed state tomography algorithm as $\mixed(\calA)$.}
\label{fig:reduction}
\end{figure}
}
\noindent
Both the mixed state $\rho$ and its random purification $\ket{\brho}$ have $\Theta(rd)$ real parameters, so this reduction should incur little loss in terms of sample complexity.
In exchange, it makes mixed state tomography significantly more intuitive to understand.
For example, while most entangled mixed state tomography algorithms  involve complicated representation theory,
here the representation theory is fairly elementary and confined to the description and analysis of the random purification channel.
And if one takes the random purification channel as a black box, then one only needs to analyze  $\ket{\brho}^{\otimes n}$, which can be done entirely with symmetric subspace computations.

Through this reduction, we give the first mixed state tomography algorithm which is sample-optimal in all parameters and gate-efficient.
Moreover, we reproduce essentially all mixed state tomography bounds which exist in the literature with conceptually simpler algorithms and proofs.
Occasionally, this will require us to make a minor modification to this generic reduction; we will discuss what this modification is and why it seems to be necessary when it arises below.

There are two main types of pure state tomography algorithms, corresponding to the unentangled and entangled measurement settings.
This leads to two types of mixed state tomography algorithms, both of which we investigate.

\subsection{Mixed state tomography: simpler, faster, and with high probability}
\label{sec:simpler}

We begin by instantiating the reduction with an unentangled pure state tomography algorithm.
The ``standard'' such algorithm is given in \Cref{fig:unentangled-tomography}.

{
\floatstyle{boxed} 
\restylefloat{figure}
\begin{figure}[H]
Given $n$ copies of $\sigma \in \C^{d \times d}$:
\begin{enumerate}
    \item For each $1 \leq i \leq n$:
        \begin{enumerate}
            \item Measure the $i$-th copy of $\sigma$ with the uniform POVM $\{d \cdot \ketbra{u} \cdot du\}$. Let $\ket{\bv_i}$ be the outcome.
            \item Set $\widehat{\bsigma}_i = (d+1) \cdot \ketbra{\bv_i} - I_d$.
        \end{enumerate}
    \item Output $\widehat{\bsigma}_{\mathrm{avg}} = \frac{1}{n} \cdot (\widehat{\bsigma}_1 + \cdots + \widehat{\bsigma}_n)$.
\end{enumerate}
\caption{The standard unentangled measurement tomography algorithm.}
\label{fig:unentangled-tomography}
\end{figure}
}

\noindent
When $\sigma$ is promised to be a pure state $\ketbra{\psi}$,
one might hope for the estimator produced by the algorithm to also be a pure state.
In this case, it is natural to post-process the output of this algorithm by computing the top eigenvector $\ket{\bv}$ of $\widehat{\bsigma}_{\mathrm{avg}}$
and using this as the estimator for $\ket{\psi}$ instead.
This pure state tomography algorithm was introduced in the work of Guta, Kahn, Kueng, and Tropp,
who showed that $\abs{\braket{\bv}{\psi}}^2 \geq 1 - \epsilon$ with probability at least $1 - \delta$ when using $n = O((d+ \log(1/\delta))/\epsilon)$ copies~\cite[Theorem 5]{GKKT20}.
One can even make this algorithm computationally efficient
via a straightforward application of unitary $t$-designs, as we show in the next theorem.
We prove this result in \Cref{sec:unentangled} below.

\begin{theorem}[Efficient pure state tomography]
    \label{thm:efficient-pure}
    There is an algorithm $\gkkt$ which, given
    \begin{equation*}
        n = O\Big(\frac{d + \log(1/\delta)}{\epsilon}\Big)
    \end{equation*}
    copies of a pure state $\ket{\psi} \in \C^d$,
    outputs a pure state $\ket{\bv} \in \C^d$ such that $\abs{\braket{\bv}{\psi}}^2 \geq 1 - \epsilon$ with probability at least $1 - \delta$.
    Furthermore, this algorithm can be implemented in $\poly(n)$ time\footnote{
        When we state running time of algorithms, we imagine that the Hilbert space $\C^d$ is being represented on a system of $\lceil \log_2(d) \rceil$ qubits on a quantum computer.
        So, running time refers to the number of one- and two-qubit gates used by the algorithm, including classical post-processing.
    } and performs independent measurements across the copies of $\ket{\psi}$.
\end{theorem}

Combining our reduction from mixed state tomography to pure state tomography with this pure state tomography algorithm results in the strongest mixed state tomography algorithm currently known.
We attain the correct dependence on failure probability $\delta$, making this algorithm sample-optimal in all parameters.
Further, we attain this optimal dependence with a gate-efficient algorithm.

\begin{theorem}[Efficient mixed state tomography]
    \label{thm:efficient-mixed}
    Let $\widehat{\brho}$ be the output of the algorithm $\mixed(\gkkt)$ when run on
    \begin{equation*}
        n = O\Big(\frac{rd + \log(1/\delta)}{\epsilon}\Big)
    \end{equation*}
    copies of a rank-$r$ mixed state $\rho \in \C^{d\times d}$.
    Then $\fidelity(\rho, \widehat{\brho}) \geq 1 - \epsilon$ with probability at least $1 - \delta$.
    Furthermore, $\mixed(\gkkt)$ can be implemented in $\poly(n)$ time.
\end{theorem}

\begin{proof}
    The first step of $\mixed(\gkkt)$ is to apply $\purifychan^{d,r}$ to $\rho^{\otimes n}$ to produce $n$ copies of a random purification $\ket{\brho} \in \C^d \otimes \C^r \cong \C^{dr}$.
    Next, we apply $\gkkt$ to learn an estimate $\ket{\bv}$ of $\ket{\brho}$.
    By \Cref{thm:efficient-pure}, with probability at least $1-\delta/2$, this estimate will satisfy $\abs{\braket{\bv}{\brho}}^2 \geq 1 - \epsilon$.
    Now, set $\widehat{\brho} = \tr_{\reg{2}}(\ketbra{\bv})$.
    By Uhlmann's theorem,
    \begin{equation*}
        \fidelity(\rho, \widehat{\brho})
        = \max_{\ket*{\psi_{\rho}}, \ket*{\psi_{\widehat{\brho}}}} \abs{\braket*{\psi_{\widehat{\brho}}}{\psi_{\rho}}}^2
        \geq \abs{\braket{\bv}{\brho}}^2,
    \end{equation*}
    where the maximization is over all purifications $\ket*{\psi_{\rho}}$ and $\ket*{\psi_{\widehat{\brho}}}$ of $\rho$ and $\widehat{\brho}$, respectively.
    Therefore, $\fidelity(\rho, \widehat{\brho}) \geq 1-\epsilon$ with probability at least $1 - \delta/2$.

    To make this efficient, note that the pure state tomography algorithm $\gkkt$ can be implemented in time $\poly(n) = \poly(d, 1/\epsilon, \log(1/\delta))$ by \Cref{thm:efficient-pure} .
    Similarly, we can implement the purifying channel $\purifychan^{d, r}$ to error $\delta/2$ in diamond distance in time $\poly(n, \log(d), \log(1/\delta)) = \poly(d, 1/\epsilon, \log(1/\delta))$ by \Cref{thm:acorn}.
    This will introduce an additional $\delta/2$ probability of failure, meaning that the algorithm succeeds with probability at least $1-\delta$.
    This completes the proof.
\end{proof}

This is the first tomography algorithm which is sample-optimal in all four parameters $d, r, \epsilon$, and $\delta$;
optimality of the $O(rd/\epsilon)$ term follows from the lower bound of~\cite{Yue23} (see also the lower bound of~\cite{SSW25}),
and optimality of the $O(\log(1/\delta)/\epsilon)$ term follows from estimating the bias of a coin with high probability.
Prior to this work, the best known bounds were
\begin{equation*}
    n = O\Big(\frac{rd}{\epsilon} \cdot \log\Big(\frac{rd + \log(1/\delta)}{\epsilon}\Big) + \frac{\log(1/\delta)}{\epsilon}\Big)
    \quad \text{and} \quad
    n = O\Big(\frac{rd}{\epsilon} \cdot \log(1/\delta)\Big),
\end{equation*}
which follow from \cite[Equation (14)]{HHJ+16} and from combining \cite[Theorem 1.6]{PSW25} with standard amplification results (for example, \cite[Proposition 2.4]{HKOT23}), respectively.
Our bound gives a strict improvement if $\delta$ is smaller than constant and larger than $(rd)^{-rd}$.

This also marks a major improvement in the sample-complexity of time-efficient tomography; we are not aware of an explicit theorem statement in the existing literature, but the works which could be made time-efficient with standard techniques~\cite{GKKT20,KRT14} achieve a far-from-optimal scaling of $n = O(r^2 d)$ when $\eps$ and $\delta$ are constant.

Even including proofs of intermediate results, the proof of \Cref{thm:efficient-mixed} is simple: ignoring time-efficiency, all it needs are basic representation theory (used in the proof of \Cref{thm:acorn}) and a scalar concentration inequality (used in the proof of \Cref{thm:efficient-pure}).
By comparison, all prior work requires more advanced representation theory, especially of the unitary group.
In addition, most prior works achieve suboptimal sample complexities for learning in fidelity even for constant settings of the error probability~$\delta$ \cite{HHJ+16,OW16,OW17a}, and so their techniques seem unable to establish optimal sample complexity bounds in terms of all four parameters.
The one exception is the recent work of~\cite{PSW25},
which shows a tight sample complexity bound of $n = O(rd/\epsilon)$ for learning with constant error probability $\delta$.
They do so by analyzing the second moment of their estimator, and it seems plausible that one could extend their result to arbitrary error probabilities $\delta$ by analyzing higher moments of their estimator.
However, their second moment proof is already quite involved, spanning dozens of pages of complicated representation-theoretic calculations.
A proof for higher moments, without further ideas, seems like it would be at best a tremendous chore and at worst completely intractable.\footnote{
    One could also consider the easier task of learning $\rho$ to error $\eps$ in trace distance.
    There, the situation is similar.
    \Cref{thm:efficient-mixed} implies an algorithm with an optimal sample complexity of
    $
        n = O((rd + \log(1/\delta))/\epsilon^2)
    $.
    This bound was not known in the literature; the previous best bounds come from the fidelity algorithms discussed here, translated to trace distance.
    Note that, for this easier regime, the upper bound of $O(rd\log(1/\delta) / \eps^2)$ was proved earlier~\cite{OW15} and the lower bound of $\Omega(rd/\eps^2)$ was proved later~\cite{SSW25}.
    Even for trace distance, it does not seem like any prior works had a clear pathway towards achieving a fully sample-optimal algorithm.
}

An interesting feature of the algorithm in \Cref{thm:efficient-mixed} is that it only uses entanglement across the copies of $\rho$ in order to produce consistent purifications of $\rho$.
After doing this, as our algorithm shows,
it suffices to only use unentangled measurements on the purified copies.
This suggests that entangled measurements are helpful in mixed state tomography \emph{because} they allow us to produce consistent purifications across all the copies of $\rho$.
This is consistent with our understanding of tomography in these settings: for mixed states, it is known that sample-optimal mixed state tomography requires entangled measurements~\cite{CHL+23,CLL24a}, but for pure states, as \Cref{thm:efficient-pure} shows, sample-optimal tomography can be achieved with algorithms which perform independent measurements across the input copies.
We will explore this theme in further detail in \Cref{subsec:discussion} below.

\subsection{Simpler unbiased estimators for mixed state tomography}

The estimator described above is sample-optimal for the task of generic full-state tomography.
However, there are other, more fine-grained tomographic tasks which have gained importance in the literature,
and it turns out that this estimator performs sub-optimally at these tasks.
We will discuss exactly why this is the case in \Cref{subsec:discussion},
but at a high level, these tasks require having an estimator which has low variance about its mean, and the above estimator has prohibitively high variance.
To tackle these other applications, we will consider applying our reduction to another pure state tomography algorithm, one due to Hayashi.

\subsubsection{Hayashi's algorithm}
Next, we instantiate the reduction with the ``standard'' \emph{entangled} pure state tomography algorithm.
When performing tomography on $n$ copies of a pure state $\ket{\psi} \in \C^d$,
the input $\ket{\psi}^{\otimes n}$ is an element of the symmetric subspace $\lor^n \C^d$.
This means that a pure state tomography algorithm's measurement operators need only be specified on the symmetric subspace.
Motivated by this, Hayashi~\cite{Hay98} introduced the following natural pure state tomography algorithm $\hayashi$:
simply perform the POVM
\begin{equation}
    \{d[n] \cdot \ketbra{u}^{\otimes n} \cdot du\}
\end{equation}
and output the pure state $\ket*{\bv}$ that it returns. (Here, $d[n] = \dim(\lor^n \C^d)$.)
This is indeed a valid POVM on the symmetric subspace, as
\begin{equation*}
    \int_{\ket{u}} d[n] \cdot \ketbra{u}^{\otimes n} \cdot du
    = d[n] \cdot \E_{\ket{\bu} \sim \mathrm{Haar}} \ketbra{\bu}^{\otimes n}
    = \Pi_{\mathrm{sym}},
\end{equation*}
where $\Pi_{\mathrm{sym}}$ is the projector onto $\lor^n \C^d$,
as we discuss in the Preliminaries below
(cf.\ \Cref{eq:pi-sym}).
Hayashi showed that this algorithm is in fact the \emph{optimal} pure state tomography algorithm in a certain precise technical sense,
and so it is perhaps the most natural algorithm to apply our reduction to.

It is well-known that the output of $\hayashi$ satisfies $\abs{\braket*{\bv}{\psi}}^2 \geq 1- \epsilon$ with probability 99\% when $n = O(d/\epsilon)$.
We show that it also achieves optimal dependence on the failure probability $\delta$.

\begin{proposition}[Hayashi's algorithm with high probability] \label{prop:hayashi}
Given $n$ copies of a pure state $\ket{\psi} \in \C^d$,
suppose $\hayashi$ outputs the state $\ket*{\bv}$. Then $\abs{\braket*{\bv}{\psi}}^2 \geq 1 - \epsilon$ with probability at least $1 - \delta$ when
\begin{equation*}
    n = O\Big(\frac{d + \log(1/\delta)}{\epsilon}\Big).
\end{equation*}
\end{proposition}
\noindent
To our knowledge, this bound has not been previously observed in the literature for Hayashi's algorithm.
It actually follows immediately from combining Hayashi's optimality result for his algorithm with the fact that the Guta--Kahn--Kueng--Tropp pure state tomography algorithm also achieves this bound~\cite[Theorem 5]{GKKT20}.
We give an alternative proof of this fact which analyzes the output of Hayashi's algorithm directly.
Combined with our reduction,
this gives a second mixed state algorithm $\mixed(\hayashi)$ which achieves the sample complexity bound in \Cref{thm:efficient-mixed}.
However, it is not computationally efficient,
as Hayashi's algorithm is not known to be computationally efficient.

\subsubsection{The Grier--Pashayan--Schaeffer algorithm}

Grier, Pashayan, and Schaeffer~\cite{GPS24} considered the following modification to Hayashi's algorithm.

{
\floatstyle{boxed} 
\restylefloat{figure}
\begin{figure}[H]
Given $n$ copies of $\sigma \in \C^{d \times d}$:
\begin{enumerate}
    \item Measure the copies with the POVM $\{d[n] \cdot \ketbra{u}^{\otimes n} \cdot du\}$. Let $\ket{\bv}$ be the outcome.
    \item Output the estimator
    \begin{equation*}\widehat{\sigma}_{\bv} \coloneqq \frac{d+n}{n} \cdot \ketbra{\bv} - \frac{1}{n} \cdot I_d.\end{equation*}
\end{enumerate}
\caption{The Grier--Pashayan--Schaeffer tomography algorithm $\gps$.}
\label{fig:entangled-tomography}
\end{figure}
}

\noindent
They observed that this modification results in an \emph{unbiased estimator} for pure state tomography,
meaning that if this algorithm is performed on $n$ copies of a pure state $\sigma \in \C^{d\times d}$,
then its output satisfies $\E [\widehat{\sigma}_{\bv}] = \sigma$.
The quality of an unbiased estimator is governed by how much it deviates from its mean,
which we can quantify using its variance $\E[(\widehat{\sigma}_{\bv} - \sigma)^{\otimes 2}] = \E [\widehat{\sigma}_{\bv}^{\otimes 2}] - \sigma^{\otimes 2}$.
This entails calculating its second moment $\E [\widehat{\sigma}_{\bv}^{\otimes 2}]$,
which we do as follows.

\begin{theorem}[Moments of Grier--Pashayan--Schaeffer]
\label{thm:gps-moments}
Let $\widehat{\sigma}_{\bv}$ be the output of $\gps$ when run on $n$ copies of a pure state $\sigma \in \C^{d \times d}$.
Then $\widehat{\sigma}_{\bv}$ is an unbiased estimator for $\sigma$, i.e.\ $\E \widehat{\sigma}_{\bv} = \sigma$.
In addition,
\begin{equation*}
    \E[\widehat{\sigma}_{\bv}\otimes \widehat{\sigma}_{\bv}] = \frac{n-1}{n} \cdot \sigma \otimes \sigma + \frac{1}{n} \cdot \big(\sigma \otimes I_d + I_d \otimes \sigma\big) \cdot \swap + \frac{1}{n^2} \cdot \swap-\mathrm{Lower}_{\sigma},
    \end{equation*}
    where $\mathrm{Lower}_{\sigma} \in \symsep(d)$.
\end{theorem}

This theorem uses the following definition for the $\mathrm{Lower}_{\sigma}$ term.
\begin{definition}[Sum of Hermitian squares] \label{def:symsep}
Given an integer $d$, we define
\begin{equation*}
    \symsep(d) \coloneqq \mathrm{cone}(\{X \otimes X \mid \text{$X \in \C^{d \times d}$ is Hermitian}\}),
\end{equation*}
where $\mathrm{cone}(\cdot)$ is the conical hull of its input,
i.e.\ the set of all nonnegative linear combinations of matrices in its input set.
\end{definition}

\noindent
The $\mathrm{Lower}_{\sigma}$ term is so-named because it is lower order in the parameter $d$ and tends towards $0$ as $d \rightarrow \infty$.
        What is important to us is that for the applications we care about, its contribution always turns out to be negative and hence can be discarded, as we will discuss in more detail below.

This suggests a natural unbiased estimator for mixed state tomography:
simply apply our reduction to the Grier--Pashayan--Schaeffer algorithm.
The result is the following algorithm.

{
\floatstyle{boxed} 
\restylefloat{figure}
\begin{figure}[H]
Given $n$ copies of $\rho$:
\begin{enumerate}
    \item\label{item:purify} First apply $\purifychan^{d, r}$ to produce $n$ copies of a random purification $\ket{\brho} \in \C^d \otimes \C^r$.
    \item Apply the Grier--Pashayan--Schaeffer algorithm to learn an estimate $\widehat{\sigma}_{\bv}$ of $\ketbra{\brho}$.
    \item Set $\widehat{\rho}_{\bv} = \tr_{\reg{2}}(\widehat{\sigma}_{\bv})$ of $\rho$. Output $\widehat{\rho}_{\bv}$.
\end{enumerate}
\caption{The mixed state tomography algorithm $\mixed(\gps)$.}
\label{fig:gps-reduction-basic}
\end{figure}
}

\noindent
By construction, this produces an unbiased estimator, as for any purification $\ket{\brho}$ of $\rho$,
\begin{equation*}
    \E[\widehat{\rho}_{\bv}]
    = \E[\tr_{\reg{2}}(\widehat{\sigma}_{\bv})]
    = \tr_{\reg{2}}(\E[\widehat{\sigma}_{\bv}])
    = \tr_{\reg{2}}(\ketbra{\brho})
    = \rho.
\end{equation*}
Using \Cref{thm:gps-moments},
we will show the following expression for the second moment of this estimator.

\begin{theorem}[Moments of the Grier--Pashayan--Schaeffer mixed state tomography algorithm] \label{thm:gps-moments-loose}
    Let $\widehat{\rho}_{\bv}$ be the output of $\mixed(\gps)$ when run on $n$ copies of a rank-$r$ state $\rho \in \C^{d \times d}$.
    Then $\widehat{\rho}_{\bv}$ is an unbiased estimator for $\rho$ with second moment
    \begin{equation*}
    \E[\widehat{\rho}_{\bv}\otimes \widehat{\rho}_{\bv}] = \frac{n-1}{n} \cdot \rho^{\otimes 2} + \frac{1}{n} \cdot \big(\rho \otimes I_d + I_d \otimes \rho\big) \cdot \swap + \frac{r}{n^2} \cdot \swap-\mathrm{Lower}_{\rho},
    \end{equation*}
    where $\mathrm{Lower}_{\rho} \in \symsep(d)$.
\end{theorem}

Historically, designing good unbiased estimators for mixed state tomography has been a challenging task.
Until recently,
our only known unbiased estimators,
such as the estimator from \Cref{fig:unentangled-tomography}, 
were not sample optimal,
and our only known sample-optimal estimators
were not unbiased~\cite{OW16,HHJ+16}.
This changed with the work of Pelecanos, Spilecki, and Wright~\cite{PSW25}, who introduced the \emph{debiased Keyl's algorithm},
the first estimator for mixed state tomography which is both sample-optimal and unbiased.
Using it, they proved a number of new and optimal sample complexity upper bounds on a number of interesting tomographic tasks,
which we will describe in \Cref{sec:our-estimator} below.

Let us compare the performance of their estimator to the estimator produced by $\mixed(\gps)$.
Writing $\widehat{\brho}$ for the output of their debiased Keyl's algorithm on $n$ copies of a rank-$r$ mixed state $\rho \in \C^{d \times d}$, they showed that
\begin{equation}\label{eq:debiased-keyl-2nd-moment}
    \E[\widehat{\brho} \otimes \widehat{\brho}]
    = \frac{n-1}{n} \cdot \rho^{\otimes 2} + \frac{1}{n} \cdot \big(\rho \otimes I_d + I_d \otimes \rho\big) \cdot \swap + \frac{\E[\ell(\blambda)]}{n^2} \cdot \swap-\mathrm{Lower}_{\rho},
    \end{equation}
where $\mathrm{Lower}_{\rho} \in \symsep(d)$~\cite[Theorem 1.4]{PSW25}.
This matches the bound in \Cref{thm:gps-moments-loose} on all terms except the factor of $r$ on the third term is replaced with the term $\E[\ell(\blambda)]$.
We will explain what exactly this notation means below;
for now, it suffices to know that it can be upper-bounded by $\E[\ell(\blambda)] \leq \min\{r, 2\sqrt{n}\}$,
which is smaller than $r$ whenever $n < r^2/4$.
This means that the debiased Keyl's algorithm will actually outperform $\mixed(\gps)$ when the number of copies $n = o(r^2)$.
As we describe below, some of our applications operate in the regime of $n = o(r^2)$ copies, in which case this difference is significant,
and some of our applications operate in the regime of $n = \omega(r^2)$ copies, in which case this difference is insignificant and the two algorithms behave similarly.

Our main goal, however, will be to improve our algorithm so that its second moment matches that of the debiased Keyl's algorithm.
Doing so will require us to design a modified version of the purification channel, which we describe in the next section.

\begin{remark}[On the lower term]\label{rem:lower}
    Let us point out one subtle difference between the way we have stated the second moment formula for the debiased Keyl's algorithm and the way it appears in \cite[Theorem 1.4]{PSW25}.
    In \cite{PSW25}, the $\mathrm{Lower}_{\rho}$ term is characterized as being a positive linear combination of matrices of the form $(P \otimes P) \cdot \swap$,
    where $P$ is Hermitian and positive semidefinite.
    As it turns out, any matrix of this form can be shown to be an element of $\symsep(d)$, and so their $\mathrm{Lower}_{\rho} \in \symsep(d)$ as well.
    This means that \Cref{eq:debiased-keyl-2nd-moment} is actually a slightly weaker characterization of the second moment than the one given in \cite{PSW25},
    but as we show below, this is still sufficient to derive all of their applications.
\end{remark}

\subsubsection{Quasi-purification}

\paragraph{The single copy case.}
To understand why $\mixed(\gps)$ performs sub-optimally in the regime of small~$n$,
it will help to first gain some understanding for how the purification channel operates in the simplest case of $n = 1$ copy.
Let $\rho \in \C^{d \times d}$ be the rank-$r$ state to be purified,
and write $\rho = \sum_{i=1}^r \alpha_i \cdot \ketbra{v_i}$ for its eigendecomposition.
Consider the purification of $\rho$ given by
\begin{equation*}
    \ket{\rho_0}_{\reg{A}\reg{B}} \coloneqq \sum_{i=1}^r \sqrt{\alpha_i} \cdot \ket{v_i}_{\reg{A}} \otimes \ket{i}_{\reg{B}},
\end{equation*}
where  $\reg{B}$ is an $r$-dimensional register.
One can generate a random purification of $\rho$ by sampling a Haar random unitary $\bU \in U(r)$
and outputting $\ket{\brho} \coloneqq \bU_{\reg{B}} \cdot \ket{\rho_0}_{\reg{A} \reg{B}}$.
The mixture over these random purifications is given by
\begin{equation}\label{eq:append-a-reg}
    \E \ketbra{\brho} = \E[\bU_{\reg{B}} \cdot \ketbra{\rho_0}_{\reg{A} \reg{B}} \cdot \bU_{\reg{B}}^{\dagger}]
    = \tr_{\reg{B}}(\ketbra{\rho_0}) \otimes (I_r/r)_{\reg{B}}
    = \rho_{\reg{A}} \otimes (I_r/r)_{\reg{B}}.
\end{equation}
Implementing the random purification channel $\purifychan^{d,r}(\rho)$ is therefore easy: simply take $\rho$ and append to it a maximally mixed state in a second register.

After performing this purification,
$\mixed(\gps)$ will perform pure state tomography on both registers of the purified state in \Cref{eq:append-a-reg}.
By \Cref{thm:gps-moments-loose}, it will output an estimator $\widehat{\rho}_{\bv}$ satisfying
\begin{equation}\label{eq:bad-second-moment}
\E[\widehat{\rho}_{\bv}\otimes \widehat{\rho}_{\bv}] = \big(\rho \otimes I_d + I_d \otimes \rho\big) \cdot \swap + r \cdot \swap-\mathrm{Lower}_{\rho}.
\end{equation}
However, the purification register in \Cref{eq:append-a-reg} contains no information whatsoever,
and so it seems wasteful to include it when performing the pure state tomography step.
Instead, what we can do is omit the purification register altogether and simply perform pure state tomography on $\rho$ itself.
This at least ``type checks'' because although $\rho$ is itself not necessarily a pure state,
it is still a mixture over the pure states corresponding to its eigenvectors.
So perform the $\gps$ algorithm on $\rho$ directly and let $\widehat{\sigma}_{\bv}$ be its output.
Then by \Cref{thm:gps-moments},
\begin{equation*}
\E[\widehat{\sigma}_{\bv}\otimes \widehat{\sigma}_{\bv}] = \big(\rho \otimes I_d + I_d \otimes \rho\big) \cdot \swap + \swap-\mathrm{Lower}_{\rho}.
\end{equation*}
This improves on \Cref{eq:bad-second-moment} by a factor of $r$ on the $\swap$ term.

\paragraph{The general case.}
Generalizing this to larger values of $n$
is conceptually more interesting than simply discarding the purification register
and requires the use of some representation theory.
In particular, let us recall \emph{Schur--Weyl duality},
which states that there is a unitary change of basis $\schur$ known as the \emph{Schur transform} under which the $n$ copy state $\rho^{\otimes n}$ becomes block diagonal,
with a block for every partition $\lambda \vdash n$ of height $\ell(\lambda) \leq d$.
We illustrate this in \Cref{fig:schur-weyl}.
It is common for entangled tomography algorithms to begin with a step known as \emph{weak Schur sampling}, in which one measures $\rho^{\otimes n}$ with the projective measurement $\{\Pi_{\lambda}\}$,
where $\Pi_{\lambda}$ is the projector onto the subspace corresponding to $\lambda$ in the Schur basis.
This produces a random partition $\blambda = (\blambda_1, \ldots, \blambda_d)$ as a measurement outcome and collapses $\rho^{\otimes n}$ to the state $\rho|_{\blambda}$.
Typically, one then performs a further measurement within the $\blambda$-subspace on the state $\rho|_{\blambda}$ in order to learn $\rho$.

{
\floatstyle{plain}
\restylefloat{figure}
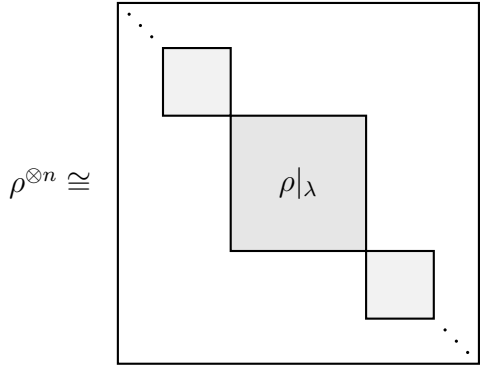
\begin{figure}
\centering
\usetikzlibrary{arrows.meta}
\begin{tikzpicture}[scale=1.2, font=\large]

\node at (-0.75, 2) {$\rho^{\otimes n} \cong$};

\draw[thick] (0, 0) rectangle (4, 4);

\draw[thick, fill=gray!10] (0.5, 3.5) rectangle (1.25, 2.75);

\draw[thick, fill=gray!20] (1.25, 1.25) rectangle (2.75, 2.75);
\node at (2, 2) {$\rho|_{\lambda}$};

\draw[thick, fill=gray!10] (2.75, 0.5) rectangle (3.5, 1.25);

\fill (0.125, 3.875) circle (0.5pt);
\fill (0.25, 3.75) circle (0.5pt);
\fill (0.375, 3.625) circle (0.5pt);
\fill (3.625, 0.375) circle (0.5pt);
\fill (3.75, 0.25) circle (0.5pt);
\fill (3.875, 0.125) circle (0.5pt);

\end{tikzpicture}
\caption{Schur--Weyl duality applied to $\rho^{\otimes n}$. Here, $\rho|_{\lambda}$ is the normalized restriction of $\rho$ to the $\lambda$-block.}
\label{fig:schur-weyl}
\end{figure}
}

The length $\ell(\blambda)$ of the measurement outcome $\blambda$, defined to be the number of nonzero coordinates $\blambda_i$ of the partition, can be viewed as a loose estimate for the rank $r$ of $\rho$.
It is in fact an \emph{under}estimate,
as it always satisfies $\ell(\blambda) \leq r$,
and for small values of $n$ it can underestimate $r$ by a significant amount.
For example, when $n = 1$, the length $\ell(\blambda) = 1$ always.
More generally, for larger $n$, its expectation satisfies $\E[\ell(\blambda)] \leq 2 \sqrt{n}$, which is significantly smaller than $r$ so long as $n = o(r^2)$.

What is so nice about this is that for the sake of purification, it turns out that we can treat $\rho|_{\ell(\blambda)}$ as if it came from a rank-$\ell(\blambda)$ state rather than a rank-$r$ state.
In particular, we will see that the rank-$k$ purification channel applied to this state, which outputs the state
\begin{equation*}\purifychan^{d, k}(\rho|_{\blambda}),
\end{equation*}
is at the very least well-defined so long as $k \geq \ell(\blambda)$.
In addition, its output, while no longer necessarily a mixture over states of the form $\ket*{\brho}^{\otimes n}$, is still guaranteed to be an element of the symmetric subspace $\lor^n (\C^d \otimes \C^{\ell(\blambda)})$, and so it at least ``type checks'' to run pure state tomography on it.

Our basic purification algorithm  can be viewed as the algorithm which always performs rank $k = r$ purification on this state, no matter what the length $\ell(\blambda)$ of the partition is.
We will now consider 
a more fine-grained purification algorithm which performs rank $k = \ell(\blambda)$ purification, the smallest value of $k$ possible.
We refer to this operation as \emph{quasi-purification},
as although it is applying the purification channel,
it can no longer be viewed as outputting purified copies of $\rho$.
In addition, it is not even a channel anymore,
as its output space $\lor^n (\C^d \otimes \C^{\ell(\blambda)})$ 
is no longer fixed, but depends on the measured $\blambda$.
However, we can still integrate into our generic reduction,
which yields the following ``upgraded'' reduction from mixed to pure state tomography.

{
\floatstyle{boxed} 
\restylefloat{figure}
\begin{figure}[H]
Given $n$ copies of $\rho$:
\begin{enumerate}
    \item Run weak Schur sampling to produce a Young diagram $\blambda \vdash n$; the $n$ states collapse to $\rho|_{\blambda}$. 
    \item Apply $\purifychan^{d, \ell(\blambda)}$ to the resulting state to produce $\purifychan^{d, \ell(\blambda)}(\rho|_{\blambda}) \in (\C^d \otimes \C^{\ell(\blambda)})^{\otimes n}$.
    \item Run an off-the-shelf pure state tomography algorithm $\calA$ on $\purifychan^{d, \ell(\blambda)}(\rho|_{\blambda})$; call the output $\widehat{\bsigma}^{\blambda}$.
    \item Convert $\widehat{\bsigma}^{\blambda}$ to an estimate $\widehat{\brho}^{\blambda} = \tr_2(\widehat{\bsigma}^{\blambda})$ of $\rho$. Output $\widehat{\brho}^{\blambda}$.
\end{enumerate}
\caption{A tighter version of our main reduction. We refer to the resulting mixed state tomography algorithm as $\mixed^+(\calA)$.}
\label{fig:reduction-llambda}
\end{figure}
}

\noindent
Note that this generalizes our $n = 1$ example from above: when $n = 1$, we have that $\ell(\blambda) = 1$ as well.
In this case, we want to apply the purification channel $\purifychan^{d, 1}(\cdot)$, but it is easy to see that this channel should not modify its input,
as a rank-1 state is already pure.

\subsubsection{Our new unbiased estimator}
\label{sec:our-estimator}

We now combine quasi-purification with the Grier--Pashayan--Schaeffer algorithm.
The resulting algorithm is still unbiased and has the following second moment guarantee.

\begin{theorem}[Moments of our unbiased estimator] \label{thm:gps-moments-tight}
    Let $\widehat{\rho}_{\bv}^{\blambda}$ be the output of $\mixed^+(\gps)$ when run on $n$ copies of a rank-$r$ state $\rho \in \C^{d \times d}$.
    Then $\widehat{\rho}_{\bv}^{\blambda}$ is an unbiased estimator for $\rho$ with second moment
    \begin{equation*}
    \E[\widehat{\rho}_{\bv}^{\blambda}\otimes \widehat{\rho}_{\bv}^{\blambda}] = \frac{n-1}{n} \cdot \rho^{\otimes 2} + \frac{1}{n} \cdot \big(\rho \otimes I_d + I_d \otimes \rho\big) \cdot \swap + \frac{\E[\ell(\blambda)]}{n^2} \cdot \swap-\mathrm{Lower}_{\rho},
    \end{equation*}
    where $\mathrm{Lower}_{\rho} \in \symsep(d)$.
\end{theorem}

This matches the second moment formula that Pelecanos, Spilecki, and Wright proved for the debiased Keyl's algorithm, which we saw previously in \Cref{eq:debiased-keyl-2nd-moment}.
They showed how to use this second moment formula to derive a number of applications of the debiased Keyl's algorithm.
We will briefly survey these applications below;
for more background on these applications, see \cite[Section 1]{PSW25},
and for proofs that they can be derived from the debiased Keyl's algorithm, see \cite[Part II]{PSW25}.
Because our estimator $\mixed^+(\gps)$ has the same second moment formula as the debiased Keyl's algorithm, these applications hold for it as well with essentially identical proofs.
The one minor modification needed to adapt these proofs to our algorithm comes from the fact that the $\mathrm{Lower}_{\rho}$ term takes a slightly different form in our algorithm than it does in the debiased Keyl's algorithm, as we saw in \Cref{rem:lower}.
To address this, we will explain below what properties of the $\mathrm{Lower}_{\rho}$ term each of these applications need, and we will show the $\mathrm{Lower}_{\rho}$ term from our algorithm does indeed satisfy these properties.

\paragraph{Application 1: tomography with limited entanglement.}
In the \emph{$k$-entangled tomography} problem, the goal is to estimate an unknown mixed state $\rho$ while performing entangled measurements on at most $k$ copies of $\rho$ at a time. 
The natural algorithm for doing so is the following.
\begin{enumerate}
    \item Divide the $n$ copies of $\rho$ into $n' \coloneqq n/k$ batches of size $k$.
    \item For each $1 \leq i \leq n'$, run a mixed state tomography algorithm on the $i$-th batch of copies and let $\widehat{\brho}_i$ be its output.
    \item Output the estimator $\widehat{\brho} = \frac{1}{n'} \cdot (\widehat{\brho}_1 + \cdots + \widehat{\brho}_{n'})$.
\end{enumerate}
When the $\widehat{\brho}_i$'s are produced by an unbiased estimator,
then averaging them together to produce $\widehat{\brho}$ results in an estimator which remains unbiased, but has significantly decreased variance.
Pelecanos, Spilecki, and Wright showed that when the debiased Keyl's algorithm is used, the $\widehat{\brho}$ this algorithm produces is $\epsilon$-close to $\rho$ in trace distance with probability 99\% when
    \begin{equation}\label{eq:limited-entanglement-bound}
            n =O\left(\max \Big(\frac{d^3}{\sqrt{k}\epsilon^2}, \frac{d^2}{\epsilon^2} \Big) \right)
    \end{equation}
copies of $\rho$ are used~\cite[Theorem 1.8]{PSW25}.
This improves on prior work of Chen, Li, and Liu~\cite{CLL24a} and matches their lower bound for this task of $n = \Omega(d^3/(\sqrt{k} \epsilon^2))$ copies, which they showed for the case when $k \leq 1/\epsilon^c$, for $c$ a small constant. 

The proof of \cite{PSW25} uses the second moment formula for the debiased Keyl's algorithm,
and as a result their sample complexity bound also applies if we use our estimator $\mixed^+(\gps)$ instead.
The one property of $\mathrm{Lower}_{\rho}$ that this proof needs is that $\tr(\swap \cdot \mathrm{Lower}_{\rho}) \geq 0$ (cf.\ the proof of \cite[Lemma 4.2]{PSW25}).
This holds in our case too, as our $\mathrm{Lower}_{\rho}$ is a nonnegative linear combination of terms of the form $X \otimes X$, where $X$ is a Hermitian matrix,
and
\begin{equation*}
    \tr(\swap \cdot X \otimes X)
    = \tr(X^2)
    \geq 0,
\end{equation*}
which holds because $X$ is Hermitian and therefore $X^2$ has all nonnegative eigenvalues.
We note that the interesting regime of the sample complexity bound in \Cref{eq:limited-entanglement-bound} is when $k = o(d^2)$;
in this case, we truly do need the quasi-purification-based algorithm $\mixed^+(\gps)$ rather than $\mixed(\gps)$ to achieve optimal sample complexity.

\paragraph{Application 2: shadow tomography.}
In the shadow tomography problem, one is given $m$ bounded observables $O_1, \ldots, O_m \in \C^{d \times d}$ which satisfy $\Vert O_i \Vert_{\infty} \leq 1$, for all $1 \leq i \leq n$, and asked to estimate the observable values $\tr(O_1 \cdot \rho), \ldots, \tr(O_m \cdot \rho)$ up to $\epsilon$ accuracy each.
The natural strategy for doing so is the following ``plug-in'' approach.
\begin{enumerate}
    \item Run a mixed state tomography algorithm on $n'$ copies of $\rho$.
    Let $\widehat{\brho}$ be the estimator it produces.
    \item Output $\widehat{\bo}_1 \coloneqq \tr(O_1 \cdot \widehat{\brho})$, \ldots, $\widehat{\bo}_m \coloneqq \tr(O_m \cdot \widehat{\brho})$. 
\end{enumerate}
When $\widehat{\brho}$ is an unbiased estimator for $\rho$,
the $\widehat{\bo}_i$'s are unbiased estimators for the true observable values $\tr(O_i \cdot \rho)$.
If $\widehat{\brho}$ has small variance about its mean,
then these $\widehat{\bo}_i$'s will also have small variance about their means; in particular, one wants to take $n'$ large enough so that each $\widehat{\bo}_i$ is within $\epsilon$ of its mean with 99\% probability.
To ensure that all  $\widehat{\bo}_i$'s are within $\epsilon$ if their mean at once, one can then perform the following algorithm.

\begin{enumerate}
    \item Repeat the ``plug-in'' approach $k$ times, producing the estimators $\widehat{\bo}_1^1, \ldots, \widehat{\bo}_m^1$ through $\widehat{\bo}_1^k, \ldots, \widehat{\bo}_m^k$.
    \item For each $1 \leq i \leq m$, output the estimator $\widehat{\bo}_i = \mathrm{median}\{\widehat{\bo}_i^1, \ldots, \widehat{\bo}_i^k\}$.
\end{enumerate}
In total, the whole process takes $n = k \cdot n'$ copies of $\rho$.
Pelecanos, Spilecki, and Wright showed that when the debiased Keyl's algorithm is used for the plug-in estimator, then this algorithm succeeds with probability 99\% when using
    \begin{equation}\label{eq:shadow-tomograph-bound}
            n =O\Big(\log(m) \cdot \Big(\min\Big\{\frac{\sqrt{rF}}{\epsilon}, \frac{F^{2/3}}{\epsilon^{4/3}}\Big\} + \frac{1}{\epsilon^2}\Big)\Big)
    \end{equation}
copies of a rank $r$ state $\rho$~\cite[Theorem 1.9]{PSW25}.
Here, each observable $O_i$ is assumed to satisfy the bound $\tr(O_i^2) \leq F$.
Since the measurements this algorithm performs are independent of the observables $O_i$,
it also solves the related ``classical shadows'' problem,
in which one is provided the observables only after measuring, with this sample complexity,
and in doing so it improves on the prior works of~\cite{HKP20,GLM24}.
In addition, note that because the observables $O_i$ satisfy $\Vert O_i \Vert_{\infty} \leq 1$,
we have $F \leq d$.
Thus, in the ``high accuracy regime'' of $\epsilon = O(1/d)$, this shows that $n = O(\log(m)/\epsilon^2)$ copies suffice,
improving on a bound of~\cite{CLL24b}.

 The proof of \cite{PSW25} uses the second moment formula for the debiased Keyl's algorithm,
and as a result their sample complexity bound also applies if we use our estimator $\mixed^+(\gps)$ instead.
The one property of $\mathrm{Lower}_{\rho}$ that this proof needs is that $\tr(O \otimes O \cdot \mathrm{Lower}_{\rho}) \geq 0$ for any observable $O$ (cf.\ the proof of \cite[Lemma 7.3]{PSW25}). 
This holds in our case too, as our $\mathrm{Lower}_{\rho}$ is a nonnegative linear combination of terms of the form $X \otimes X$, where $X$ is a Hermitian matrix,
and
\begin{equation*}
    \tr(O \otimes O \cdot X \otimes X)
    = \tr(O X)^2
    \geq 0,
\end{equation*}
which holds because $X$ is Hermitian and therefore $\tr(O X)$ is a real number whose square is therefore nonnegative.
Finally, we note that to achieve the full sample complexity bound of \Cref{eq:shadow-tomograph-bound}, we do require using the quasi-purified algorithm $\mixed^+(\gps)$.
However, in the high accuracy regime when $\epsilon = O(1/d)$,
if we assume bounds of $r, F \leq d$,
then our algorithm requires using $n = \Omega(d^2)$ copies of $\rho$,
in which case it suffices to use the simpler algorithm $\mixed(\gps)$.

\paragraph{Application 3: quantum metrology.}

In multiparameter quantum metrology, one is given copies of a quantum state $\rho_{\theta}$ parameterized by a vector $\theta \in \R^m$, and the goal is to output an estimator $\widehat{\btheta}$ of $\theta$.
This estimator is \emph{locally unbiased} at a point $\theta^*$ if it satisfies
\begin{equation*}
    (i)~\E[\widehat{\btheta} \mid \rho_{\theta^*}] = 0
    \quad\text{and}\quad
    (ii)~\frac{\partial}{\partial \theta_i} \E[\widehat{\btheta} \mid \rho_{\theta}]\Big|_{\theta = \theta^*} = 0.
\end{equation*}
Given a locally unbiased estimator, we can evaluate its performance using the \emph{mean squared error matrix (MSEM)} $V \in \R^{m \times m}$ defined as
\begin{equation*}
    V_{ij} \coloneqq \E[(\widehat{\btheta}_i - \theta_i^*)(\widehat{\btheta}_j - \theta_j^*) \mid \rho_{\theta^*}], \quad \text{for all $i, j \in [m]$}.
\end{equation*}
The \emph{quantum Cramér–Rao bound (QCRB)}~\cite{Hel67b} states that the performance of the MSEM can always be lower-bounded via $V \succeq \calF^{-1}$, where $\calF$ is a particular matrix known as the \emph{Quantum Fisher Information} (QFI) matrix.
A definition of this matrix can be found, for example, in \cite[Section 1.6]{PSW25}.

Recently, Zhou and Chen gave a generic method for converting unbiased estimators for mixed state tomography into locally unbiased estimators for multiparameter quantum metrology~\cite{ZC25}.
Plugging the debiased Keyl's algorithm into this transformation,
Pelecanos, Spilecki, and Wright gave locally unbiased unbiased estimator with the following guarantee.
Writing $V_n$ for the MSEM of their matrix when given $n$ copies of $\rho_{\theta}$, it satisfies $n \cdot V_n \rightarrow 2\cdot \calF^{-1}$ in the limit as $n \rightarrow \infty$~\cite[Theorem 1.10]{PSW25}.
In other words, their estimator achieves twice the QCRB asymptotically.
This is optimal, as there are examples of parameterized quantum states in which the factor of 2 is necessary~\cite[Section 3.1.1]{DGG20}.

 The proof of \cite{PSW25} uses the second moment formula for the debiased Keyl's algorithm,
and as a result, we can also achieve twice the QCRB asymptotically by plugging our estimator $\mixed^+(\gps)$ into the Zhou and Chen transformation.
Their proof requires two properties of $\mathrm{Lower}_{\rho_{\theta}}$.
First, they need that for any observables $O_1$ and $O_2$, $\tr(O_1 \otimes O_2 \cdot \mathrm{Lower}_{\rho_{\theta}}) = \tr(O_2 \otimes O_1 \cdot \mathrm{Lower}_{\rho_{\theta}})$.
This holds in our case too, as our $\mathrm{Lower}_{\rho_{\theta}}$ is a nonnegative linear combination of terms of the form $X \otimes X$, where $X$ is a Hermitian matrix,
and
\begin{equation*}
    \tr(O_1 \otimes O_2  \cdot X \otimes X) = \tr(O_1 \cdot X)\cdot  \tr(O_2  \cdot X)
    = \tr(O_2 \otimes O_1 \cdot X \otimes X).
\end{equation*}
In addition, they need that for any matrix $Q$, $\tr(Q^\dagger \otimes Q \cdot \mathrm{Lower}_{\rho_{\theta}}) \geq 0$. This too holds in our case, as
\begin{equation*}
    \tr(Q^\dagger \otimes Q \cdot X \otimes X)
    = \tr(Q^\dagger \cdot X) \cdot \tr(Q \cdot X)
    = \overline{\tr(Q \cdot X)} \cdot \tr(Q \cdot X)
    = \abs{\tr(Q \cdot X)}^2 \geq 0,
\end{equation*}
where the third equality uses the fact that $X$ is Hermitian.
(See \cite[Proof of Theorem 8.1]{PSW25} for both of these required properties.)
Note that for this result, we care about the asymptotic regime of $n \rightarrow \infty$, in which case it actually suffices to use the  simpler algorithm $\mixed(\gps)$.

\subsection{Discussion} \label{subsec:discussion}

Our results demonstrate that the purification channel provides a powerful tool for designing mixed state tomography algorithms.
We believe that it will have more applications in the future, and we conclude with some open directions along these lines.

\paragraph{The purification channel.}
We still feel like we lack a deep understanding of the purification channel.
For example, currently we only know how to implement it using Schur transforms; could this channel be implemented in a more elementary way, without using representation theory?
Similarly, for some of our applications, we must apply the purification channel to a state which has larger rank than the purification channel ``supports''.
Is there a natural interpretation of the behavior of the channel in this case?

\paragraph{Other applications of the purification channel.}

Are there more applications of the mixed state to pure state reduction beyond the ones we considered in this work?
One possibility is the shadow tomography problem: an immediate consequence of our reduction is that mixed state shadow tomography generically reduces to pure state shadow tomography.
Could this help design better shadow tomography algorithms in the future?
More generally, are there additional applications of the purification channel, even beyond quantum learning?

\paragraph{Efficient algorithms for Hayashi's algorithm.}
Although we now have an efficient quantum algorithm for vanilla mixed state tomography due to \Cref{thm:efficient-mixed}, we still do not have efficient algorithms for any of our other applications, as these all require performing Hayashi's algorithm.
This motivates the following question: can Hayashi's algorithm be made efficient?
We believe the answer is yes,
and we leave this question to future work.

\paragraph{The debiased Keyl's algorithm.}
In the course of this work, we have resolved a number of open problems stated in the debiased Keyl's algorithm paper (admittedly much faster than we were expecting), namely open problems 1 (learning with high probability), 4 (efficient algorithms), and 6 (simpler proofs of the variance formula), except we resolved these questions for our algorithm $\mixed^+(\gps)$ rather than for the debiased Keyl's algorithm.
Can any of our techniques help us give simpler proofs for the debiased Keyl's algorithm?
More concretely, does purification give us a simpler perspective to understand the measurements that the debiased Keyl's algorithm performs?
Also, are there any tomographic tasks which separate the performance of these two algorithms, or do they essentially behave identically for all tasks?

\paragraph{Unentangled measurements for pure state problems.}
For the task of vanilla mixed state tomography, our \Cref{thm:efficient-mixed} showed that after purifying $\rho$, it suffices to run a pure state tomography algorithm which uses unentangled measurements.
On the other hand, our other applications involve entangled measurements across the purified copies.
So, is this just a fluke of vanilla tomography?
Or could it be that once $\rho$ has been purified, one never requires entangled measurements?
In other words, could it be that entanglement is only useful for producing consistent purifications, for all tomographic tasks?

To investigate these questions, let us first explain why we did not use unentangled measurements for our mixed state unbiased estimators.
The natural first thing to try would be to plug in the standard unentangled measurement tomography algorithm into our mixed state to pure state reduction, giving the algorithm $\mixed(\calA_{\mathrm{standard}})$.
To understand how well this performs, let us first look at $\calA_{\mathrm{standard}}$.
Given $n$ copies of a pure state $\sigma \in \C^{d \times d}$,
its output $\widehat{\bsigma}_{\mathrm{avg}}$ has second moment
\begin{equation*}
    \E[\widehat{\bsigma}_{\mathrm{avg}}\otimes \widehat{\bsigma}_{\mathrm{avg}}] = \frac{n-1}{n} \cdot \sigma \otimes \sigma + \frac{1}{n} \cdot \big(\sigma \otimes I_d + I_d \otimes \sigma\big) \cdot \swap + \frac{1}{n} \cdot \swap-\mathrm{Lower}_{\sigma}.
\end{equation*}
(In fact, this expression also holds when $\sigma$ is a mixed state.)
This is easy to show by direct calculation and we omit the proof. 
Note that the third term has a $1/n$ factor on the $\swap$, whereas the Grier--Pashayan--Schaeffer algorithm has a significantly smaller factor of $1/n^2$ (see \Cref{thm:gps-moments}).
This means that this algorithm is not useful for any of our unbiased estimator applications. 
The ultimate issue is that although $\calA_{\mathrm{standard}}$ has strong $\ell_{\infty}$ guarantees, which are sufficient for achieving optimal pure state tomography bounds,
it has relatively weak $\ell_2$ guarantees, and these are what many of our other applications rely on.

This could be a fundamental limitation of all algorithms which use unentangled measurements, but it could also be a sign that we are simply using the wrong unentangled pure state tomography algorithm.
Perhaps a better algorithm would be to first compute the top eigenvector $\ket{\bv}$ of $\widehat{\bsigma}_{\mathrm{avg}}$
and consider the density matrix $\ketbra{\bv}$.
This will certainly be a biased estimator of $\sigma$, but one can certainly correct for its bias, much as Grier, Pashayan, and Schaeffer corrected for the bias of Hayashi's algorithm~\cite{GPS24}. 
How well does the resulting algorithm perform?
We don't have a clear answer to this question, but there is limited evidence suggesting that it may indeed be a better estimator than $\widehat{\bsigma}_{\mathrm{avg}}$ on its own.
In particular, Grier, Pashayan, and Schaeffer showed that when $\sigma$ is a pure state, there is an unbiased estimator for $\sigma$ which is closely related to $\widehat{\bsigma}_{\mathrm{avg}}^2$, and this unbiased estimator outperforms the unsquared $\widehat{\bsigma}_{\mathrm{avg}}$ at the classical shadows task, at least in some regimes of parameters.
Note that squaring $\widehat{\bsigma}_{\mathrm{avg}}$ has the effect of putting more weight on its larger eigenvectors and less weight on its smaller eigenvectors, which is a step in the direction of $\ketbra{\bv}$.

We note that there are some quantum learning tasks where entangled measurements help, even when the input states are pure.
These include testing if a bipartite pure state is product or entangled~\cite{CCHL22,Har23}
and testing if a multipartite pure state has a hidden cut~\cite{BCS+25}.
As for tasks with a more tomographic flavor,
the only separation we are aware of is learning stabilizer states, which can be solved with $O(n)$ copies using entangled measurements but requires $\Omega(n^2)$ copies using unentangled measurements~\cite{ABDY23}, though we note that this lower bound proof only seems to apply to algorithms which make nonadaptive measurements.
(We thank Sitan Chen for pointing out these applications to us.)
However, it seems possible to us that the unentangled pure state tomography algorithm we suggested above could be competitive with Hayashi's algorithm for many tomographic tasks.
This would of course not yield improved sample complexity bounds, as it seems clear that Hayashi's algorithm is performing the ``right'' measurement for pure state tomography, but it would help us understand the philosophical question of when and why entangled measurements help.

\subsection{Organization}
The rest of this document closely follows the introduction in organization.
We begin with relevant preliminaries in \Cref{sec:prelims} below.
Then, in \Cref{sec:the-channel}, we describe the random purification channel (\Cref{thm:acorn}) and provide the relevant tools for analyzing quasi-purification.

Next, we investigate pure state tomography algorithms in the unentangled and entangled measurement settings in \Cref{sec:unentangled,sec:entangled}, respectively.
In the unentangled section, we give a sample-optimal gate-efficient pure state tomography algorithm (\Cref{thm:efficient-pure}), which through our reduction produces a sample-optimal gate-efficient mixed state tomography algorithm (\Cref{thm:efficient-mixed}).
In the entangled tomography section, we study the Grier--Pashayan--Schaeffer unbiased estimator for pure states and compute its second moment (\Cref{thm:gps-moments}).
Through our reduction, this produces a state-of-the-art unbiased estimator for mixed states (\Cref{thm:gps-moments-tight}); the proof that this reduction works is the focus of \Cref{sec:unbiased-estimator}.

On the way, we prove the intermediate results about Hayashi's algorithm and the Grier--Pashayan--Schaeffer algorithm discussed in the introduction.
We include these, as they give simple proofs of slightly weaker statements: sample-optimal tomography without time-efficiency (\Cref{prop:hayashi}) and unbiased estimators which do not use quasi-purification (\Cref{thm:gps-moments-loose}).

Finally, in \Cref{sec:pgm}, we analyze the behavior of $\mixed(\hayashi)$, $\mixed(\gps)$, and $\mixed^+(\gps)$, and find that the measurements they perform are Pretty Good Measurements (PGMs) over the Hilbert--Schmidt measure.
This gives a way to describe these algorithms without using the purification channel.

\section{Preliminaries} \label{sec:prelims}

We use \textbf{boldface} to denote random variables and \textsf{sans serif} to denote registers, i.e.\ subsystems of a larger system. 
For example, $\rho = \rho_{\reg{A}_1\dots\reg{A}_n}$ denotes a state on $n$ registers, where we drop the subscript if the registers are clear; we can then denote the partial trace on a subsystem as, for example, $\rho_{\reg{A}_1} = \tr_{\reg{A}_2\dots\reg{A}_n}(\rho)$.
If not otherwise stated, $\reg{1}$ denotes the first subsystem, $\reg{2}$ the second, and so on.

We use $\norm{X}_1$ and $\norm{X}_\infty$ to denote the trace norm and operator norm (i.e.\ Schatten $1$-norm and Schatten $\infty$-norm, respectively).
For two quantum states $\rho$ and $\sigma$, the fidelity between them is $\fidelity(\rho, \sigma) = \norm*{\rho^{1/2} \sigma^{1/2}}_1^2$, so that, in particular, $\fidelity(\ketbra{u}, \ketbra{v}) = \abs{\braket{u}{v}}^2$.
We use $A \preceq B$ to denote PSD ordering, i.e.\ $B - A$ is positive semidefinite.

Recall the definition of $\symsep(d)$ from \cref{def:symsep}.
We observe that this cone is closed under partial trace.

\begin{proposition}[Partial trace preserves the cone of Hermitian squares]\label{prop:symset-partial-trace}
Given integers $d$ and $r$,
consider the Hilbert space $\C^d \otimes \C^r \cong \C^D$,
where $D = d \cdot r$.
Call $d$-dimensional registers $\reg{A}$ and $r$-dimensional registers $\reg{B}\textsl{}$.
Let $M_{\reg{A}_1\reg{B}_1\reg{A}_2\reg{B}_2} \in \symsep(D)$. Then
$\tr_{\reg{B}_1\reg{B}_2}(M) \in \symsep(d)$.
\end{proposition}
\begin{proof}
    Since $M_{\reg{A}_1\reg{B}_1\reg{A}_2\reg{B}_2} \in \symsep(D)$,
    it can be written as a positive linear combination of matrices of the form $X_{\reg{A}_1\reg{B}_1} \otimes X_{\reg{A}_2\reg{B}_2}$, where $X$ is a Hermitian matrix acting on $\C^D$.
    Then $\tr_{\reg{B}_1\reg{B}_2}(M)$
    can be written as a positive linear combination of matrices of the form 
    \begin{equation*}
        \tr_{\reg{B}_1\reg{B}_2}(X_{\reg{A}_1\reg{B}_1} \otimes X_{\reg{A}_2\reg{B}_2})
        = \tr_{\reg{B}_1}(X_{\reg{A}_1\reg{B}_1}) \otimes \tr_{\reg{B}_2}(X_{\reg{A}_2\reg{B}_2}).
    \end{equation*}
    Since $\tr_{\reg{B}}(X_{\reg{AB}})$ is Hermitian, $\tr_{\reg{B}_1\reg{B}_2}(M)$ is in $\symsep(d)$.
\end{proof}

\subsection{The symmetric group}
We write $S_n$ for the symmetric group on $n$ elements, and $e$ for the identity element of $S_n$.
We sometimes use cycle notation to represent permutations; for example, $(i, j)$ is the transposition swapping elements $i, j \in [n]$, and $(1, 2) (2, 3) =  (1, 2, 3)$.

For $m < n$, the subgroup $S_{m}$ naturally embeds into $S_n$ by associating the permutation $\sigma \in S_{m}$ with the permutation $\sigma' \in S_n$ defined by taking $\sigma'(i) = \sigma(i)$ when $i \leq m$ and taking $\sigma'(i) = i$ otherwise.
The corresponding \emph{group algebra} of $S_n$ consists of all linear combinations of symmetric group elements $\sum_{\pi \in S_n} \alpha_\pi \cdot \pi$ with coefficients $\alpha_\pi \in \C$.
An especially important subset of the symmetric group algebra are the \emph{Jucys--Murphy elements}, given by 
\begin{equation*}
    X_1 = 0, \quad \text{and} \quad X_i \coloneqq (1, i) + \cdots + (i-1, i), \quad \text{for $2 \leq i \leq n$}.
\end{equation*}
For us, their significance arises from the following formula,
which relates them to the uniform sum over all permutations.

\begin{proposition}[Product of Jucys--Murphy elements]\label{prop:jm-prod}
    Let $n \geq 1$. Then
    \begin{equation*}
        (e + X_n) \cdots (e + X_1) = \sum_{\pi \in S_n} \pi.
    \end{equation*}
\end{proposition}
\begin{proof}
    We prove this by induction on $n$.
    The $n = 1$ base case is trivial.
    For the inductive step, let us assume that this statement is true for $n-1$, i.e.\
    \begin{equation*}
        (e + X_{n-1}) \cdots (e + X_1) = \sum_{\sigma \in S_{n-1}} \sigma.
    \end{equation*}
    Then our goal is to show that
    \begin{align}
        \sum_{\pi \in S_n} \pi
        = (e + X_n) \cdots (e + X_1)
        &= (e + X_n) \cdot \sum_{\sigma \in S_{n-1}} \sigma\nonumber\\
        &= \sum_{\sigma \in S_{n-1}} \sigma
        + \sum_{i=1}^{n-1} \sum_{\sigma \in S_{n-1}} (i, n) \cdot \sigma.\label{eq:big-goal}
    \end{align}
    To prove this, we will show that every $\pi \in S_n$ occurs as a summand in \Cref{eq:big-goal}.
    First, suppose that $\pi(n) = n$. Then $\pi \in S_{n-1}$ as well, and so it can be found in the first sum in \Cref{eq:big-goal}.
    Otherwise, suppose $\pi(n) = i$ for some $1 \leq i \leq n-1$. Then $(i, n) \cdot \pi$, fixes $n$, and so we can write $(i, n) \cdot \pi = \sigma$, for some $\sigma \in S_{n-1}$. But then $\pi = (i, n) \cdot \sigma$, which can be found in the second sum in \Cref{eq:big-goal}.
    Thus, every $\pi \in S_n$ occurs as a summand in \Cref{eq:big-goal}.
    As there are exactly $(n-1)! + (n-1) \cdot (n-1)! = n!$ terms in this equation, each must correspond to a unique permutation in $S_n$, completing the proof.
\end{proof}

Given $\pi \in S_n$, $P_d(\pi)$ is defined to be the unitary operation which acts on $(\C^d)^{\otimes n}$ by permuting the $n$ registers via the equation
\begin{equation*}
    P_d(\pi) \cdot \ket{i_1, \ldots, i_n}
    = \ket{i_{\pi^{-1}(1)}, \ldots, i_{\pi^{-1}(n)}},
    \quad \text{for all $i_1, \ldots, i_n \in [d]$}.
\end{equation*}
When $d$ is clear from context, we will often write this as $P(\pi)$ for simplicity.
We may sometimes even drop the notation $P(\cdot)$ altogether and write this simply as $\pi$.
When $n = 2$, we will often write $\swap \coloneqq P((1, 2))$.

We will commonly encounter having two sets of $n$ registers: $n$ registers of Hilbert space $\C^d$, named $\reg{A}_1, \ldots, \reg{A}_n$ and $n$ registers of Hilbert space $\C^r$ named $\reg{B}_1, \ldots, \reg{B}_n$.
Pairing these up gives $n$ registers of Hilbert space $(\C^d \otimes \C^r) \cong \C^D$, for $D = d \cdot r$, named $\reg{A}_1 \reg{B}_1, \ldots, \reg{A}_n \reg{B}_n$.
Given a permutation $\pi \in S_n$,
we will write:
\begin{itemize}
\item[$\circ$] either $P_d(\pi)$ or $P_{\reg{A}}(\pi)$ for the permutation matrix acting on $\reg{A}_1, \ldots, \reg{A}_n$,
\item[$\circ$]  either $P_r(\pi)$ or $P_{\reg{B}}(\pi)$ for the permutation matrix acting on $\reg{B}_1, \ldots, \reg{B}_n$, and
\item[$\circ$]  either $P_D(\pi)$ or $P_{\reg{A}\reg{B}}(\pi)$ for the permutation matrix acting on $\reg{A}_1\reg{B}_1, \ldots, \reg{A}_n\reg{B}_n$.
\end{itemize}
When we are in the above situation with $n = 2$,
we will apply similar notational choices to the $\swap$ matrix
(so that we write $\swap_d$ or $\swap_{\reg{A}}$ for the $\swap$ matrix acting on $\reg{A}_1 \reg{A}_2$, and so on).
We will abuse notation and use e.g.\ $\swap_{\reg{A}}$ to refer to both the operator on $\reg{A}$ alone, along with the operator $\swap_{\reg{A}} \otimes I_\reg{B}$.

The operator $P_{\reg{A}\reg{B}}(\pi)$ decomposes across the $\reg{A}$ and $\reg{B}$ registers in the natural way.
\begin{proposition}\label{prop:permute-factor}
Let $\reg{A}_1, \ldots, \reg{A}_n$ be registers with Hilbert space $\C^d$
and $\reg{B}_1, \ldots, \reg{B}_n$ be registers with Hilbert space $\C^r$.
Then for all $\pi \in S_n$,
\begin{equation*}
    P_{\reg{A}\reg{B}}(\pi) = P_{\reg{A}}(\pi) \otimes P_{\reg{B}}(\pi).
\end{equation*}
\end{proposition}
\begin{proof}
    Let $\pi \in S_n$. Then for all $a_1, \ldots, a_n \in [d]$ and $b_1, \ldots, b_n \in [r]$, 
    \begin{align*}
    &P_{\reg{A}\reg{B}}(\pi)
        \cdot \Big(\ket{a_1}_{\reg{A}_1}\ket{b_1}_{\reg{B}_1} \otimes \cdots \otimes \ket{a_n}_{\reg{A}_n} \ket{b_n}_{\reg{B}_n}\Big)\\
        ={}& \Big(\ket*{a_{\pi^{-1}(1)}}_{\reg{A}_1}\ket*{b_{\pi^{-1}(1)}}_{\reg{B}_1}\Big) \otimes \cdots \otimes \Big(\ket*{a_{\pi^{-1}(n)}}_{\reg{A}_n} \ket*{b_{\pi^{-1}(n)}}_{\reg{B}_n}\Big)\\
        ={}& \Big(\ket*{a_{\pi^{-1}(1)}}_{\reg{A}_1} \otimes \cdots \otimes \ket*{a_{\pi^{-1}(n)}}_{\reg{A}_n}\Big) \otimes \Big(\ket*{b_{\pi^{-1}(1)}}_{\reg{B}_1} \otimes \cdots \otimes  \ket*{b_{\pi^{-1}(n)}}_{\reg{B}_n}\Big)\\
        ={}& \Big(P_{\reg{A}}(\pi) \cdot \ket{a_1}_{\reg{A}_1} \otimes \cdots \otimes \ket{a_n}_{\reg{A}_n}\Big) \otimes \Big(P_{\reg{B}}(\pi) \cdot \ket{b_1}_{\reg{B}_1} \otimes \cdots \otimes \ket{b_n}_{\reg{B}_n}\Big)\\
        ={}& P_{\reg{A}}(\pi) \otimes P_{\reg{B}}(\pi)
        \cdot \Big(\ket{a_1}_{\reg{A}_1}\ket{b_1}_{\reg{B}_1} \otimes \cdots \otimes \ket{a_n}_{\reg{A}_n} \ket{b_n}_{\reg{B}_n}\Big).
    \end{align*}
    This completes the proof.
\end{proof}

Our moment computations will feature partial traces of permuted operators very heavily.
Simplifying these expressions is straightforward, either through manipulating them algebraically or inspecting their tensor network diagrams.
We will now prove a couple of helper propositions to illustrate how this can be done.
Later, we will perform similar computations, and their proofs will follow similarly.

\begin{proposition}\label{prop:no-tensor-network-diagrams}
    Let $M = M_{\reg{A}\reg{B}}$ act on $\C^d \otimes \C^r$. Then
    \begin{equation*}
        \tr_{\reg{A}}(\swap_{\reg{A}\reg{C}} \cdot (M_{\reg{A}\reg{B}} \otimes (I_d)_\reg{C}))
        = M_{\reg{C}\reg{B}}
        = \tr_{\reg{A}}((M_{\reg{A}\reg{B}} \otimes (I_d)_\reg{C}) \cdot \swap_{\reg{A}\reg{C}}).
    \end{equation*}
\end{proposition}
\begin{proof}
    We verify that every entry of $\tr_{\reg{A}}(\swap_{\reg{A}\reg{C}} \cdot (M_{\reg{A}\reg{B}} \otimes (I_d)_\reg{C}))$ equals the corresponding entry of $M_{\reg{C}\reg{B}}$: for arbitrary $b,b' \in [r]$ and $c,c' \in [d]$, we have
    \begin{align*}
        \bra{bc}_{\reg{BC}} \tr_{\reg{A}}(\swap_{\reg{A}\reg{C}} \cdot (M_{\reg{A}\reg{B}} \otimes (I_d)_\reg{C})) \ket{b'c'}_{\reg{BC}}
        &= \sum_{a=1}^d \bra{abc}_{\reg{ABC}} \cdot \swap_{\reg{AC}} \cdot (M_{\reg{AB}} \otimes (I_d)_{\reg{C}}) \cdot \ket{ab'c'}_{\reg{ABC}}\\
        &= \sum_{a=1}^d \bra{cba}_{\reg{ABC}} \cdot (M_{\reg{AB}} \otimes (I_d)_{\reg{C}}) \cdot \ket{ab'c'}_{\reg{ABC}} \\
        &= \sum_{a=1}^d \bra{cb}M\ket{ab'} \cdot \bra{a}I_d\ket{c'}
        = \bra{cb}M\ket{c'b'}
        = \bra{bc}_{\reg{BC}} M_{\reg{CB}} \ket{b'c'}_{\reg{BC}}.
    \end{align*}
    The analogous statement for $\tr_{\reg{A}}((M_{\reg{A}\reg{B}} \otimes (I_d)_\reg{C}) \cdot \swap_{\reg{A}\reg{C}})$ follows by taking the conjugate transpose of the computation above.
\end{proof}

We will occasionally make use of \Cref{prop:no-tensor-network-diagrams} under the following guise, and so we record it here.

\begin{proposition}\label{prop:weird-prop}
    Let $M$ act on $\reg{A}_1 \reg{B}_1$.
    Then
    \begin{equation*}
        \tr_{\reg{B}_1 \reg{B}_2}( M_{\reg{A}_1\reg{B}_1} \otimes I_{\reg{A}_2 \reg{B}_2} \cdot \swap_{\reg{A} \reg{B}}) = (\tr_{\reg{B}_1}(M_{\reg{A}_1\reg{B}_1}) \otimes I_{\reg{A}_2}) \cdot \swap_{\reg{A}}.
    \end{equation*}
\end{proposition}
\begin{proof}  
    Recall that $\swap_{\reg{A}\reg{B}}$ swaps both $\reg{A}_1$ with $\reg{A}_2$ and $\reg{B}_1$ with $\reg{B}_2$.
    We calculate:
    \begin{align*}
    \tr_{\reg{B}_1\reg{B}_2}(M_{\reg{A}_1\reg{B}_1}  \otimes I_{\reg{A}_2\reg{B}_2} \cdot \swap_{\reg{A}\reg{B}})
    &= \tr_{\reg{B}_1\reg{B}_2}(M_{\reg{A}_1\reg{B}_1}  \otimes I_{\reg{A}_2\reg{B}_2} \cdot \swap_{\reg{A}} \cdot \swap_{\reg{B}})\tag{by \Cref{prop:permute-factor}}\\
    &= \tr_{\reg{B}_1\reg{B}_2}(M_{\reg{A}_1\reg{B}_1}  \otimes I_{\reg{A}_2\reg{B}_2} \cdot \swap_{\reg{B}}) \cdot \swap_{\reg{A}}\\
    &= \tr_{\reg{B}_2}\Big(\tr_{\reg{B}_1}(M_{\reg{A}_1\reg{B}_1}  \otimes I_{\reg{A}_2\reg{B}_2} \cdot \swap_{\reg{B}})\Big) \cdot \swap_{\reg{A}} \\
    &= \tr_{\reg{B}_2}\Big(M_{\reg{A}_1 \reg{B}_2} \otimes I_{\reg{A}_2}\Big) \cdot \swap_{\reg{A}} \tag{by \Cref{prop:no-tensor-network-diagrams}}\\
    &= (\tr_{\reg{B}_2}(M_{\reg{A}_1\reg{B}_2}) \otimes I_{\reg{A}_2}) \cdot \swap_{\reg{A}}.
\end{align*}
This completes the proof.
\end{proof}

\subsection{The symmetric subspace}

Now, we will recall standard facts about the symmetric subspace.
These can be found, for example, in the survey by Harrow~\cite{Har13}.
The \emph{symmetric subspace} on $(\C^d)^{\otimes n}$ is given by
\begin{equation*}
    \lor^n \C^d \coloneqq \{\ket{\psi} \in (\C^d)^{\otimes n} \mid \forall \pi \in S_n, ~ P(\pi) \cdot \ket{\psi} = \ket{\psi}\}.
\end{equation*}
It can be equivalently written as
\begin{equation*}
    \lor^n \C^d = \mathrm{span}\{\ket{\psi}^{\otimes n} \mid \ket{\psi} \in \C^d\}.
\end{equation*}
We will denote by $d[n]$ the dimension of $\lor^n \C^d$, which is given by the formula
\begin{equation*}
    d[n] = \binom{n+d-1}{n}.
\end{equation*}
The ratio of successive dimensions satisfies the formula
\begin{equation}\label{eq:ratio}
    \frac{d[n]}{d[n+1]}=\frac{n+1}{n+d}.
\end{equation}
Finally, we write $\Pi_{\mathrm{sym}}^{n,d}$ for the projector onto $\lor^n \C^d$.
It can be computed by the formula
\begin{equation}\label{eq:pi-sym}
    \Pi_{\mathrm{sym}}^{n,d}
    = \E_{\bpi \sim S_n} [P(\bpi)]
    = d[n] \cdot \E_{\ket{\bu} \sim \mathrm{Haar}} \ketbra{\bu}^{\otimes n}.
\end{equation}
We will sometimes drop either $n$ or $d$ from the notation $\Pi_{\mathrm{sym}}^{n,d}$ when they are clear from context.
The projector onto the symmetric subspace obeys the following nice recurrence relation.

\begin{proposition}[Symmetric subspace projector recurrence]\label{prop:pi-sym-recurrence}
    Let $n \geq 2$ and $1 \leq m \leq n$. Then
    \begin{equation*}
         \Pi_{\mathrm{sym}}^{n,d} = \Big(\frac{e + X_n}{n}\Big) \cdots \Big(\frac{e + X_{m+1}}{m+1}\Big) \cdot\Pi_{\mathrm{sym}}^{m,d} \otimes (I_d)^{\otimes (n-m)}.
    \end{equation*}
\end{proposition}
\begin{proof}
    By \Cref{eq:pi-sym},
    \begin{align*}
        \Pi_{\mathrm{sym}}^{n,d}
        = \E_{\bpi \sim S_n} [\bpi]
        = \frac{1}{n!}\cdot \sum_{\pi \in S_n} \pi
        &= \frac{1}{n!}\cdot(e + X_n) \cdots (e + X_1)\\
        &= \frac{m!}{n!} \cdot (e + X_n) \cdots (e+ X_{m+1}) \cdot \frac{1}{m!} \sum_{\sigma \in S_{m}} \sigma\\
        &= \frac{m!}{n!}\cdot (e + X_n) \cdots (e+ X_{m+1}) \cdot \E_{\bsigma \sim S_{m}} [\bsigma]\\
        &= \frac{m!}{n!} \cdot (e + X_n) \cdots (e + X_{m+1}) \cdot \Pi_{\mathrm{sym}}^{m,d} \otimes (I_d)^{\otimes (n-m)},
    \end{align*}
    where the third and the fourth equalities used \Cref{prop:jm-prod}.
    This completes the proof.
\end{proof}

\subsection{Representation theory}

To analyze the random purification channel (\Cref{sec:the-channel}) and its subsequent use in analyzing the unbiased estimator (\Cref{sec:gps-mixed-llambda} and \Cref{sec:pgm}), we will need some standard facts from representation theory.
These facts are not necessary for the rest of the document.
For a more thorough treatment of these topics, see~\cite{Wri16}.

\paragraph{Partitions and Young diagrams.}
A \emph{partition of $n$}, denoted $\lambda \vdash n$, is a tuple of integers $\lambda  = (\lambda_1, \ldots, \lambda_k)$ such that $\lambda_1 \geq \cdots \geq \lambda_k \geq 0$ and $\lambda_1 + \cdots + \lambda_k = n$. The \emph{length of $\lambda$}, denoted $\ell(\lambda)$, is equal to the number of nonzero components $\lambda_i$.
Partitions are typically represented pictorially using \emph{Young diagrams}, which consist of boxes arranged into rows of length $\lambda_1$ through $\lambda_k$.
A \emph{standard Young tableau (SYT) $S$ of shape $\lambda$} is a Young diagram of shape $\lambda$ in which each box has been filled in with a number from $[n]$,
with the restriction that the numbers in each row must be strictly increasing from left-to-right
and the numbers in each column must be strictly increasing from top-to-bottom.
A \emph{semistandard Young tableau (SSYT) $T$ of shape $\lambda$ and alphabet $[d]$} is a Young diagram of shape $\lambda$ in which each box has been filled in with a number from $[d]$,
with the restriction that the numbers in each row must be weakly increasing from left-to-right
and the numbers in each column must be strictly increasing from top-to-bottom.
We illustrate these concepts in \Cref{fig:example_Young_tableaux}.

\begin{figure}[h!]
    \centering
    \begin{minipage}{0.35\textwidth}
        \centering
        \begin{ytableau}
              1 & 2 & 5 & 7\\
              3 & 6  \\
              4 
        \end{ytableau}
        \vspace{0.7em} 
    \end{minipage}
    \hspace{1.6em}
    \begin{minipage}{0.35\textwidth}
        \centering
        \begin{ytableau}
            1 & 1 & 1 & 3 \\ 
            2 & 2 \\
            3 
        \end{ytableau}
        \vspace{0.7em}
    \end{minipage}

    \caption{Examples of tableaux of shape $\lambda = (4,2,1)$. Left: an SYT. Right: an SSYT, for $d \geq 3$.}
    \label{fig:example_Young_tableaux}
\end{figure}
\paragraph{Representation theory of the symmetric group.}

The irreducible representations of the symmetric group are indexed by partitions $\lambda \vdash n$ and are written $(\kappa_{\lambda}, \mathrm{Sp}_{\lambda})$, where $\mathrm{Sp}_{\lambda}$ is known as the \emph{Specht module}.
There is a convenient choice of basis for these irreducible representations known as \emph{Young's orthogonal basis},
which gives rise to an explicit choice of the matrices $\kappa_{\lambda}(\pi)$ known as \emph{Young's orthogonal representation}.
For the irreducible representation corresponding to the Young diagram $\lambda$, Young's orthogonal basis has a basis vector $\ket{S}$ for each standard Young tableau $S$ of shape $\lambda$.
These vectors form an orthonormal basis of $\mathrm{Sp}_{\lambda}$.
Furthermore, when written in this basis, the matrices $\kappa_{\lambda}(\pi)$ have only real-valued entries.
For shorthand, we write the dimension of the $\lambda$-irrep as $\dim(\lambda)$.

\paragraph{Representation theory of the general linear group.}
The \emph{general linear group $GL(d)$} is the group consisting of all invertible matrices $M$ in $\C^{d \times d}$.
The polynomial irreducible representations of the general linear group are indexed by partitions $\lambda$ with $\ell(\lambda) \leq d$ and are written $(\nu^d_{\lambda}, V^d_{\lambda})$, where $V^d_{\lambda}$ is known as the \emph{Weyl module}.
Here, ``polynomial'' means that the matrix entries of each representation $\nu^d_{\lambda}(M)$ can be expressed as a polynomial in the matrix entries of $M$.
There is an orthonormal basis of the Weyl module $V_{\lambda}^d$ known as the \emph{Gelfand-Tsetlin basis}, in which the basis vectors $\ket{T}$ are indexed by SSYTs of shape $\lambda$ and alphabet $[d]$.
An important subgroup of $GL(d)$ is the unitary group $U(d)$.
The irreducible representations $(\nu^d_{\lambda}, V^d_{\lambda})$ of $GL(d)$ also form the irreducible representations of $U(d)$.

We will also sometimes want to plug in density matrices $\rho \in \C^{d \times d}$ into the irreducible representation $\nu^d_{\lambda}$. When $\rho$ has positive eigenvalues, $\nu^d_{\lambda}(\rho)$ is clearly well-defined as such a $\rho$ is in $GL(d)$.
But this is also well-defined even when $\rho$ has some eigenvalues which are zero, by remembering that the matrix entries of $\nu^d_{\lambda}(\rho)$ are polynomials in the matrix entries of $\rho$ and are therefore always well-defined.
Equivalently, by continuity, we can view $\nu^d_{\lambda}(\rho)$ as the limit of $\nu^d_{\lambda}(\rho^+)$ for a sequence of matrices $\rho^+$ with positive eigenvalues which approach $\rho$ in the limit.

\paragraph{Schur--Weyl duality.}
Given $\pi \in S_n$, let us recall the representation $P(\pi)$ of $S_n$ which acts on $(\C^d)^{\otimes n}$ as follows:
\begin{equation*}
    P(\pi) \cdot \ket{i_1, \ldots, i_n}
    = \ket*{i_{\pi^{-1}(1)}, \ldots, i_{\pi^{-1}(n)}},
    \quad \text{for all $i_1, \ldots, i_n \in [d]$}.
\end{equation*}
Given $M \in GL(d)$, we can also define a representation $Q^d(M)$ of $GL(d)$ which acts on $(\C^d)^{\otimes n}$, as follows:
\begin{equation*}
    Q^d(M) \cdot \ket{i_1, \ldots, i_n}
    = (M \cdot \ket{i_1}) \otimes \cdots \otimes (M \cdot \ket{i_n}).
\end{equation*}
For all $\pi \in S_n$ and $M \in GL(d)$,
we have that $P(\pi) \cdot Q^d(M) = Q^d(M)\cdot P(\pi)$,
and so these representations commute with each other.
As a result, there is a change of basis in which these representations are simultaneously block-diagonalized into their irreducible representations.
The precise form of this simultaneous block-diagonalization is provided by \emph{Schur--Weyl duality}, which states that there is a unitary $\schur^d$ known as the \emph{Schur transform} such that, for all $\pi \in S_n$ and $M \in GL(d)$,
\begin{equation*}
    \schur^d \cdot P(\pi) Q^d(M) \cdot (\schur^d)^{\dagger}
    = \sum_{\lambda \vdash n, \ell(\lambda) \leq d} \ketbra{\lambda} \otimes \kappa_{\lambda}(\pi) \otimes \nu^d_{\lambda}(M).
\end{equation*}
An efficient algorithm for computing the Schur transform was provided by Bacon, Chuang, and Harrow~\cite{BCH05}.
By the footnote on page 160 of Harrow's Ph.D.\ thesis~\cite{Har05}, it can be computed to $\epsilon$ accuracy in diamond distance in time $\poly(n, \log(d), \log(1/\epsilon))$.
The fact that this efficient algorithm produces the Young orthogonal basis on the symmetric group register was recently shown by Pelecanos, Spilecki, and Wright~\cite[Appendix A]{PSW25}.

While in the Schur basis,
it is natural to measure which irrep $\lambda$ one is in.
Doing so is known as \emph{weak Schur sampling}
and involves performing the projective measurement $\{\Pi_{\lambda}\}_{\lambda \vdash n, \ell(\lambda) \leq d}$ in which
\begin{equation*}
\Pi_{\lambda} \coloneqq \ketbra{\lambda} \otimes I_{\dim(\lambda)} \otimes I_{\dim(V_{\lambda}^d)}.
\end{equation*}

\section{The random purification channel} \label{sec:the-channel}

\newcommand{\epr}{\mathrm{EPR}}

In this section, we formally define the random purification channel $\purifychan^{d,r}$.
Then, we reprove \Cref{thm:acorn}, closely following the treatment given in \cite[Section 2.3]{TWZ25}.
Finally, we prove some additional facts needed to understand the $\mixed^+$ reduction, which we use later in \Cref{sec:gps-mixed-llambda}. 

The purification channel $\purifychan^{d,r}$ converts the state $\rho^{\otimes n}$
into the mixture $\E_{\ket{\brho}} \ketbra{\brho}^{\otimes n}$.
To do so, it is crucial to understand what these two mixed states look like in the Schur basis. For the former, this is easy, as $\rho^{\otimes n}$ is just the state $Q^d(\rho)$. As a result, we can apply Schur--Weyl duality, which states that
\begin{equation}\label{eq:normal-sw-transform}
    \schur^d \cdot \rho^{\otimes n} \cdot (\schur^d)^{\dagger}
    = \sum_{\lambda \vdash n, \ell(\lambda) \leq d} \ketbra{\lambda} \otimes I_{\dim(\lambda)} \otimes \nu^d_{\lambda}(\rho).
\end{equation}
For the latter, this is a bit more challenging.
To begin, let us set up some notation.
Each state $\ket{\brho}$ lives inside $\C^d \otimes \C^r \cong \C^D$, where $D = d \cdot r$. The $n$ $d$-dimensional registers will be denoted $\reg{A_1}, \dots, \reg{A_n}$, and the $n$ $r$-dimensional registers will be denoted $\reg{B_1}, \dots, \reg{B_n}$, though these will often be dropped when clear from context. We will consider what $\E_{\ket{\brho}} \ketbra{\brho}^{\otimes n}$ looks like when we apply a Schur transform $\schur^d$ to the $n$ different $\C^d$ registers and a second Schur transform $\schur^r$ to the $n$ different $\C^r$ registers. We will abbreviate $\schur^d \otimes \schur^r$ as $\schur^{\otimes 2}$. The first Schur transform will produce a basis of the form $\ket{\lambda}_{\reg{Y}} \otimes \ket{S}_{\reg{P}} \otimes \ket{T}_{\reg{Q}}$, where $\reg{Y}$ is the Young diagram register, $\reg{P}$ is the symmetric group register, and $\reg{Q}$ is the unitary group register.
The second Schur transform will produce a basis of the form $\ket{\lambda'}_{\reg{Y}'} \otimes \ket{S'}_{\reg{P}'} \otimes \ket{T'}_{\reg{Q}'}$.  

Our first observation is that $\ketbra{\brho}^{\otimes n}$ is contained inside the symmetric subspace $\lor^n \C^D$, and so $\E_{\ket{\brho}} \ketbra{\brho}^{\otimes n}$ is in $\lor^{n} \C^D$ as well. Towards understanding the random purification channel, we now turn to characterizing~$\lor^n \C^D$. 

\subsection{The symmetric subspace across two registers}

We start by considering the projector onto the symmetric subspace of $(\C^d \otimes \C^r)^{\otimes n}$ in the Schur basis. \cite[Lemma 2.13]{TWZ25} gives a convenient formula for this projector, which we restate and reprove. We will need the following~definition. 

\begin{definition}[Specht module EPR state]
    Let $\lambda \vdash n$. We write $\ket{\epr_{\lambda}}$ for the pure state inside $\mathrm{Sp}_{\lambda} \otimes \mathrm{Sp}_{\lambda}$ given by
    \begin{equation*}
        \ket{\epr_{\lambda}} \coloneqq \frac{1}{\sqrt{\dim(\lambda)}} \cdot \sum_{S} \ket{S} \otimes \ket{S},
    \end{equation*}
    where the sum ranges over all SYTs of shape $\lambda$.
\end{definition}

\begin{lemma}[The projector onto the symmetric subspace in the Schur basis]\label{lem:double-schur-pi-sym}
    \begin{equation} \label{eq:double-schur-pi-sym}
        \schur^{\otimes 2} \cdot \Pi_{\mathrm{sym}}^{n, D} \cdot (\schur^{\otimes 2})^\dagger
        = \sum_{\lambda \vdash n, \ell(\lambda) \leq r} \ketbra{\lambda\lambda}{\lambda\lambda}_{\reg{Y}\reg{Y'}} \otimes \ketbra{\epr_{\lambda}}{\epr_{\lambda}}_{\reg{P}\reg{P'}} \otimes I_{\reg{Q}\reg{Q'}}.
    \end{equation}
\end{lemma}

\begin{proof}
    By \Cref{prop:permute-factor},
    \begin{equation*}
        \Pi_{\mathrm{sym}}^{n, D} = \E_{\bpi \sim S_n} \big[P_{\reg{AB}}(\bpi)\big] = \E_{\bpi \sim S_n}\big[ P_{\reg{A}}(\bpi) \otimes P_{\reg{B}}(\bpi)\big].
    \end{equation*}
    If we now Schur transform $\reg{A}$ and $\reg{B}$, we obtain:
    \begin{align}
        & (\schur^d \otimes \schur^r) \cdot \Pi_{\mathrm{sym}}^{n, D} \cdot (\schur^d \otimes \schur^r)^\dagger  \nonumber \\
        & = \E_{\bpi \sim S_n} \Big[\big(\schur^d \cdot P_{\reg{A}}(\bpi) \cdot (\schur^d)^\dagger\big) \otimes \big(\schur^r \cdot P_{\reg{B}}(\bpi) \cdot (\schur^r)^\dagger\big) \Big] \nonumber \\
        & = \E_{\bpi \sim S_n} \Big[\Big(\sum_{\lambda \vdash n, \ell(\lambda) \leq d} \ketbra{\lambda}_{\reg{Y}} \otimes \kappa_\lambda(\bpi)_{\reg{P}} \otimes (I_{\dim(V^d_\lambda)})_{\reg{Q}}\Big) \otimes \Big(\sum_{\mu \vdash n, \ell(\mu) \leq r} \ketbra{\mu}_{\reg{Y'}} \otimes \kappa_\mu(\bpi)_{\reg{P'}} \otimes (I_{\dim(V^r_\mu)})_{\reg{Q'}}\Big)\Big] \nonumber \\
        & = \sum_{\substack{\lambda \vdash n, \ell(\lambda) \leq d \\ \mu \vdash n, \ell(\mu) \leq r}} \ketbra{\lambda}_{\reg{Y}} \otimes \ketbra{\mu}_{\reg{Y'}} \otimes \E_{\bpi \sim S_n} \big[\kappa_\lambda(\bpi)_{\reg{P}} \otimes \kappa_\mu(\bpi)_{\reg{P'}}\big] \otimes (I_{\dim(V^d_\lambda)})_{\reg{Q}} \otimes (I_{\dim(V^r_\mu)})_{\reg{Q'}}. \label{eq:averaging_1}
    \end{align}
    We can simplify the operator acting on $\reg{PP'}$ using the grand orthogonality relations:
    \begin{align*}
        \E_{\bpi \sim S_n} \big[\kappa_\lambda(\bpi)_{\reg{P}} \otimes \kappa_\mu(\bpi)_{\reg{P'}}\big] & = \sum_{S, S', S'', S'''} \E_{\bpi \sim S_n} \big[\kappa_\lambda(\bpi)_{S,S''} \cdot \kappa_{\mu}(\bpi)_{S',S'''}\big] \cdot \ketbra{S}{S''}_{\reg{P}} \otimes \ketbra{S'}{S'''}_{\reg{P'}} \\
        & = \sum_{S, S', S'', S'''} \E_{\bpi \sim S_n} \big[\overline{\kappa_\lambda(\bpi)_{S,S''}} \cdot \kappa_{\mu}(\bpi)_{S',S'''}\big] \cdot \ketbra{S}{S''}_{\reg{P}} \otimes \ketbra{S'}{S'''}_{\reg{P'}} \\
        & = \sum_{S, S', S'', S'''} \frac{1}{\dim(\lambda)} \cdot \delta_{\lambda, \mu} \cdot \delta_{S,S'} \cdot \delta_{S'',S'''} \cdot \ketbra{S}{S''}_{\reg{P}} \otimes \ketbra{S'}{S'''}_{\reg{P'}} \\
        & = \delta_{\lambda,\mu} \cdot \ketbra{\epr_\lambda}_{\reg{PP'}}.
    \end{align*}
    In the second equality, we have also used the fact that Young's orthogonal representation is real-valued. Plugging this back into \Cref{eq:averaging_1} yields \Cref{eq:double-schur-pi-sym}.
\end{proof}

For any pure state $\ket{u} \in \C^D$, the state $\ketbra{u}^{\otimes n} \in \lor^n \C^D$. It will be useful for us to understand what this state looks like in the Schur basis as well. 

\begin{lemma}[Formula for $\ketbra{u}^{\otimes n}$ in the Schur basis]
\label{lem:double_schur_transform_pure_state}
    Let $\ket{u}_{\reg{A_1}\reg{B_1}} \in \C^d \otimes \C^r \cong \C^D$. Then 
    \begin{equation} \label{eq:double_schur_n_fold_pure_state}
        \schur^{\otimes 2} \cdot \ketbra{u}^{\otimes n}_{\reg{AB}} \cdot (\schur^{\otimes 2})^\dagger = \sum_{\substack{\lambda \vdash n, \ell(\lambda)\leq r \\ \mu \vdash n, \ell(\mu) \leq r}} \ketbra{\lambda \lambda}{\mu \mu}_{\reg{YY'}} \otimes \ketbra{\epr_\lambda}{\epr_{\mu}}_{\reg{PP'}} \otimes \ketbra{u_{\lambda \lambda}}{u_{\mu \mu}}_{\reg{QQ'}},
    \end{equation}
    where $\ket{u_{\lambda \lambda}}_{\reg{QQ'}}$ is some (unnormalized) vector in $V^d_\lambda \otimes V^r_\lambda$. Moreover, for each $\lambda$, \begin{equation}\label{eq:double_schur_n_fold_pure_state_partial_trace_restriction}
        \tr_{\reg{Q'}}( \ketbra{u_{\lambda \lambda}}_{\reg{QQ'}}) = \dim(\lambda) \cdot \nu_\lambda^d(\sigma_u)_{\reg{Q}},
    \end{equation}
    where $\sigma_u = \tr_{\reg{B_1}}(\ketbra{u}_{\reg{A_1B_1}})$.
\end{lemma}

\begin{proof}
    Since $\ket{u}^{\otimes n}_{\reg{AB}} \in \lor^n \C^D$, we have $\Pi_{\mathrm{sym}}^{n,D} \cdot \ket{u}^{\otimes n}_{\reg{AB}} = \ket{u}^{\otimes n}_{\reg{AB}}$, and thus 
    \begin{align*}
        \schur^{\otimes 2} \cdot \ket{u}^{\otimes n}_{\reg{AB}} & = \Big(\schur^{\otimes 2} \cdot \Pi_{\mathrm{sym}}^{n,D} \cdot (\schur^{\otimes 2})^\dagger\Big) \cdot  \Big(\schur^{\otimes 2} \cdot \ket{u}^{\otimes n}_{\reg{AB}}\Big) \\
        & = \Big(\sum_{\lambda \vdash n, \ell(\lambda) \leq r} \ketbra{\lambda \lambda}_{\reg{YY'}} \otimes \ketbra{\epr_\lambda}_{\reg{PP'}} \otimes I_{\reg{QQ'}} \Big) \cdot \Big(\schur^{\otimes 2} \cdot \ket{u}_{\reg{AB}}^{\otimes n}\Big) \tag{\Cref{lem:double-schur-pi-sym}} \\ 
        & = \sum_{\lambda \vdash n, \ell(\lambda) \leq r} \ket{\lambda\lambda}_{\reg{YY'}} \otimes \ket{\epr_\lambda}_{\reg{PP'}} \otimes \ket{u_{\lambda \lambda}}_{\reg{QQ'}},
    \end{align*}
    where $\ket{u_{\lambda \lambda}}_{\reg{QQ'}}$ is some (unnormalized) vector in $V^d_\lambda \otimes V^r_\lambda$. Thus, 
    \begin{align*}
        \schur^{\otimes 2} \cdot \ketbra{u}_{\reg{AB}}^{\otimes n} \cdot ( \schur^{\otimes 2} )^\dagger = \sum_{\substack{\lambda \vdash n, \ell(\lambda) \leq r \\ \mu \vdash n, \ell(\mu) \leq r}} \ketbra{\lambda \lambda}{\mu \mu}_{\reg{YY'}} \otimes \ketbra{\epr_\lambda}{\epr_\mu}_{\reg{PP'}} \otimes \ketbra{u_{\lambda \lambda}}{u_{\mu \mu}}_{\reg{QQ'}}. 
    \end{align*}
    This proves \Cref{eq:double_schur_n_fold_pure_state}. To show \Cref{eq:double_schur_n_fold_pure_state_partial_trace_restriction}, we first note that 
\begin{align}
    \schur^d \cdot \tr_{\reg{B}}( \ketbra{u}^{\otimes n} ) \cdot (\schur^d)^\dagger & = \schur^d \cdot \sigma_u^{\otimes n} \cdot (\schur^d)^\dagger \nonumber \\
    & = \sum_{\lambda \vdash n, \ell(\lambda) \leq r} \ketbra{\lambda}_{\reg{Y}} \otimes (I_{\dim(\lambda)})_{\reg{P}} \otimes \nu_\lambda^d(\sigma_u)_{\reg{Q}}.\label{eq:double_schur_eq_1}
\end{align} 
However, we also have 
\begin{align}
    \schur^d \cdot \tr_{\reg{B}}( \ketbra{u}^{\otimes n} ) \cdot (\schur^d)^\dagger & = \tr_{\reg{B}} \Big(  (\schur^d \otimes \schur^r) \cdot  \ketbra{u}^{\otimes n}  \cdot (\schur^d \otimes \schur^r)^\dagger\Big) \nonumber \\
    & = \tr_{\reg{B}} \Big( \sum_{\lambda, \mu} \ketbra{\lambda \lambda}{\mu \mu}_{\reg{YY'}} \otimes \ketbra{\epr_\lambda}{\epr_{\mu}}_{\reg{PP'}} \otimes \ketbra{u_{\lambda \lambda}}{u_{\mu \mu}}_{\reg{QQ'}} \Big) \nonumber \\
    & = \sum_{\lambda \vdash n, \ell(\lambda) \leq r} \ketbra{\lambda} \otimes \Big( \frac{I_{\dim(\lambda)}}{\dim(\lambda)} \Big) \otimes \tr_{\reg{Q'}} (\ketbra{u_{\lambda \lambda}}_{\reg{QQ'}}). \label{eq:double_schur_eq_2}
\end{align}
\Cref{eq:double_schur_n_fold_pure_state_partial_trace_restriction} then follows by comparing the matrices acting on $\reg{Q}$ in \Cref{eq:double_schur_eq_1,eq:double_schur_eq_2}.
\end{proof}

\subsection{A formula for a many-copy random purification}

We now apply the results from the previous section to the random purification channel. We start by understanding the output of the purification channel, $\E_{\ket{\brho}} \ketbra{\brho}^{\otimes n}$, in the Schur basis, using the fact that this state is in the symmetric subspace. 
\cite[Lemma 2.16]{TWZ25} gives us the following formula. 

\begin{lemma}[Random purification formula]\label{lem:final-formula}
Let $\rho_{\reg{A_1}} \in \C^{d \times d}$ be a mixed state, and let $\ket{\rho_0}_{\reg{A_1B_1}} \in \C^d \otimes \C^r \cong \C^D$ be a fixed purification of $\rho$. Recall that a random purification $\ket{\brho}_{\reg{A_1B_1}}$ is obtained from $\ket{\rho_0}$ by applying a Haar random unitary $\bU \in U(r)$ to the purifying register, i.e.\ $\ket{\brho}_{\reg{A_1B_1}} = (I_d \otimes \bU) \cdot \ket{\rho_0}_{\reg{A_1B_1}}$. Then 
\begin{equation*} 
    \schur^{\otimes 2} \cdot\Big(\E_{\ket{\brho}} \ketbra{\brho}_{\reg{AB}}^{\otimes n}\Big) \cdot (\schur^\dagger)^{\otimes 2}
        = \sum_{\lambda \vdash n, \ell(\lambda)\leq r} \dim(\lambda) \cdot \ketbra{\lambda\lambda}_{\reg{Y} \reg{Y'}} \otimes \ketbra{\epr_{\lambda}}_{\reg{P}\reg{P'}} \otimes \nu^d_{\lambda}(\rho)_{\reg{Q}} \otimes \Big(\frac{I_{\dim(V_{\lambda}^r)}}{\dim(V_{\lambda}^r)}\Big)_{\reg{Q'}}.
\end{equation*}
\end{lemma}

\begin{proof}[Proof of \Cref{lem:final-formula}]
    First, we have 
    \begin{equation*}
        \E_{\ket{\brho}} \ketbra{\brho}^{\otimes n}_{\reg{AB}} = \E_{\bU \sim \mathrm{Haar}} \big[ \bU_{\reg{B}}^{\otimes n} \cdot \ketbra{\rho_0}^{\otimes n}_{\reg{AB}} \cdot (\bU_{\reg{B}}^{\otimes n})^\dagger \big].
    \end{equation*}
    We will now rewrite everything in the Schur basis. \Cref{lem:double_schur_transform_pure_state} shows us how to rewrite $\ketbra{\rho_0}^{\otimes n}$:
    \begin{equation*}
        \schur^{\otimes 2} \cdot \ketbra{\rho_0}^{\otimes n}_{\reg{AB}} \cdot (\schur^{\otimes 2})^\dagger = \sum_{\substack{\lambda \vdash n, \ell(\lambda) \leq r \\ \mu \vdash n, \ell(\mu) \leq r}} \ketbra{\lambda \lambda}{\mu \mu}_{\reg{YY'}} \otimes \ketbra{\epr_{\lambda}}{\epr_{\mu}}_{\reg{PP'}} \otimes \ketbra{\rho_{0, \lambda \lambda}}{\rho_{0, \mu \mu}}_{\reg{QQ'}},\end{equation*}
    for some unnormalized vectors $\{\ket{\rho_{0,\lambda\lambda}} \}_{\lambda}$. 
    Conjugating this state by a Haar random unitary yields, in the Schur basis, 
    \begin{align}
     & \schur^{\otimes 2} \cdot\Big(\E_{\ket{\brho}} \ketbra{\brho}_{\reg{AB}}^{\otimes n}\Big) \cdot (\schur^\dagger)^{\otimes 2} \nonumber \\
     & = \sum_{\substack{\lambda \vdash n, \ell(\lambda) \leq r \\ \mu \vdash n, \ell(\mu) \leq r}} \ketbra{\lambda \lambda}{\mu \mu}_{\reg{YY'}} \otimes \ketbra{\epr_{\lambda}}{\epr_{\mu}}_{\reg{PP'}} \otimes \E_{\bU \sim \mathrm{Haar}} \big[ \nu_{\lambda}^r(\bU)_{\reg{Q'}} \cdot \ketbra{\rho_{0, \lambda \lambda}}{\rho_{0, \mu \mu}}_{\reg{QQ'}} \cdot \nu_{\mu}^r(\bU)^\dagger_{\reg{Q'}}\big]. \label{eq:random_purification_formula_eq_1}
    \end{align}
    Let us focus on the operator in the unitary registers. This operator is a linear combination of terms of the form 
    \begin{equation*} 
        \E_{\bU \sim \mathrm{Haar}} \big[ \nu_{\lambda}^r(\bU)_{\reg{Q'}} \cdot \ketbra{T}{T''}_{\reg{Q}} \otimes \ketbra{T'}{T'''}_{\reg{Q'}} \cdot \nu_{\mu}^r(\bU)^\dagger_{\reg{Q'}}\big] = \ketbra{T}{T''}_{\reg{Q}} \otimes \E_{\bU \sim \mathrm{Haar}} \big[ \nu^r_\lambda(\bU) \cdot \ketbra{T'}{T'''} \cdot \nu^r_\mu(\bU)^\dagger]_{\reg{Q'}}.
    \end{equation*}
    Here, $T, T''$ are SSYTs of shape $\lambda$ and $\mu$, respectively, with alphabet $[d]$. Meanwhile, $T', T'''$ are SSYTs of shape $\lambda$ and $\mu$, respectively, with alphabet $[r]$. The operator
    \begin{equation*}
        \E_{\bU \sim \mathrm{Haar}} \big[ \nu^r_\lambda(\bU) \cdot \ketbra{T'}{T'''} \cdot \nu^r_\mu(\bU)^\dagger]_{\reg{Q'}}
    \end{equation*}
    is an intertwiner for $\nu_\lambda^r$ and $\nu_\mu^r$, and so by Schur's lemma is nonzero only when $\lambda = \mu$. In this case, we have 
    \begin{align*}
        \E_{\bU \sim \mathrm{Haar}} \big[ \nu^r_\lambda(\bU) \cdot \ketbra{T'}{T'''} \cdot \nu^r_\lambda(\bU)^\dagger]_{\reg{Q'}} & = \tr(\E_{\bU \sim \mathrm{Haar}} \big[ \nu^r_\lambda(\bU) \cdot \ketbra{T'}{T'''} \cdot \nu^r_\lambda(\bU)^\dagger \big]) \cdot \Big(\frac{I_{\dim(V^r_\lambda)}}{\dim(V^r_\lambda)}\Big)_{\reg{Q'}} \\
        & = \tr( \ketbra{T'}{T'''} ) \cdot \Big(\frac{I_{\dim(V^r_\lambda)}}{\dim(V^r_\lambda)}\Big)_{\reg{Q'}}.
    \end{align*}
    Thus, twirling by a Haar random unitary maps 
    $\ketbra{T}{T''}_{\reg{Q}} \otimes \ketbra{T'}{T'''}_{\reg{Q'}}$ to
    \begin{equation*}
    \delta_{\lambda,\mu} \cdot \ketbra{T}{T''}_{\reg{Q}} \otimes \tr(\ketbra{T'}{T'''}) \cdot \Big(\frac{I_{\dim(V^r_\lambda)}}{\dim(V^r_\lambda)}\Big)_{\reg{Q'}}  = \delta_{\lambda,\mu} \cdot \tr_{\reg{Q'}}( \ketbra{T}{T''}_{\reg{Q}} \otimes \ketbra{T'}{T'''}_{\reg{Q'}}) \otimes \Big(\frac{I_{\dim(V^r_\lambda)}}{\dim(V^r_\lambda)}\Big)_{\reg{Q'}}.\end{equation*}
Substituting this back into \Cref{eq:random_purification_formula_eq_1} then gives:
    \begin{equation*}
        \eqref{eq:random_purification_formula_eq_1} = \sum_{\lambda \vdash n, \ell(\lambda) \leq r } \ketbra{\lambda \lambda}{\lambda \lambda}_{\reg{YY'}} \otimes \ketbra{\epr_{\lambda}}{\epr_{\lambda}}_{\reg{PP'}} \otimes \tr_{\reg{Q'}}( \ketbra{\rho_{0,\lambda\lambda}}_{\reg{QQ'}}) \otimes\Big(\frac{I_{\dim(V^r_\lambda)}}{\dim(V^r_\lambda)}\Big)_{\reg{Q'}}.
    \end{equation*}
    Finally, by \Cref{lem:double_schur_transform_pure_state}, we have $\tr_{\reg{Q'}}(\ketbra{\rho_{0,\lambda\lambda}}_{\reg{QQ'}}) = \dim(\lambda) \cdot \nu_\lambda^d(\rho)_{\reg{Q}}$. Plugging this into the above expression yields:
    \begin{equation*} 
    \schur^{\otimes 2} \cdot\Big(\E_{\ket{\brho}} \ketbra{\brho}_{\reg{AB}}^{\otimes n}\Big) \cdot (\schur^\dagger)^{\otimes 2}
        = \sum_{\lambda \vdash n, \ell(\lambda)\leq r} \dim(\lambda) \cdot \ketbra{\lambda\lambda}_{\reg{Y} \reg{Y'}} \otimes \ketbra{\epr_{\lambda}}_{\reg{P}\reg{P'}} \otimes \nu^d_{\lambda}(\rho)_{\reg{Q}} \otimes \Big(\frac{I_{\dim(V_{\lambda}^r)}}{\dim(V_{\lambda}^r)}\Big)_{\reg{Q'}}.
\end{equation*}
This completes the proof. \qedhere

\end{proof}

\subsection{The random purification channel}

We now define the random purification channel $\purifychan^{d,r}$. 

{
\floatstyle{boxed} 
\restylefloat{figure}
\begin{figure}[H]
Given $n$ copies of $\rho$:
\begin{enumerate}
    \item Apply the Schur transform $\schur^d$ to $\rho^{\otimes n}$.
    \item \label{item:intro-2} Perform weak Schur sampling. Letting $\blambda$ be the outcome, the state collapses to
    \begin{equation*}
        \rho|_{\blambda} \coloneqq \ketbra{\blambda}_{\reg{Y}} \otimes \Big(\frac{I_{\dim(\blambda)}}{\dim(\blambda)}\Big)_{\reg{P}} \otimes \Big(\frac{\nu^d_{\lambda}(\rho)}{s^d_{\lambda}(\rho)}\Big)_{\reg{Q}}.
    \end{equation*}
    \item \label{item:intro-y} Introduce a new register $\reg{Y'}$ and copy $\blambda$ into it.
    \item \label{item:intro-p}  Introduce a new register $\reg{P}'$. Discard the contents of $\reg{P}$ and place the state $\ket{\mathrm{EPR}_{\blambda}}_{\reg{P}\reg{P}'}$ into these registers.
    \item \label{item:intro-q}  Introduce a new register $\reg{Q}'$ initialized to the maximally mixed state $I_{\dim(V^r_{\blambda})}/\dim(V^r_{\blambda})$.
    \item Apply the inverse Schur transform $(\schur^d \otimes \schur^r)^\dagger$ and output the resulting state.
\end{enumerate}
\caption{The random purification channel $\purifychan^{d,r}$.}
\label{fig:purify-channel}
\end{figure}
}

By \cite[Lemma 2.11]{TWZ25}, the random purification channel satisfies
\begin{equation*}
    \purifychan^{d,r}(\rho^{\otimes n}) = \E_{\ket{\brho}} \ketbra{\brho}^{\otimes n}.
\end{equation*}
Furthermore,
it can be computed to $\delta$ error in time $\poly(n, \log(d), \log(1/\delta))$.  We now reprove these results, for completeness. 

\begin{proof}[Proof of \Cref{thm:acorn}]
    To show $\purifychan^{d,r}(\rho^{\otimes n})$ prepares a random purification, we track how our input state changes with each applied step. After Schur transforming, we have the state
    \begin{equation*}
        \sum_{\lambda \vdash n, \ell(\lambda) \leq r} \ketbra{\lambda}_{\reg{Y}} \otimes ( I_{\dim(\lambda)} )_{\reg{P}} \otimes \nu^d_{\lambda}(\rho)_{\reg{Q}}. 
    \end{equation*}
    After weak Schur sampling, we obtain outcome $\blambda \vdash n$ with probability $\dim(\blambda) \cdot s^d_{\blambda}(\rho)$, and state $\rho|_{\blambda}$. Conditioned on $\blambda$, following steps 3--5, we obtain the state
    \begin{equation*}
        \ketbra{\blambda\blambda}_{\reg{YY'}} \otimes \ketbra{\epr_{\blambda}}_{\reg{PP'}} \otimes \Big(\frac{\nu^d_{\blambda}(\rho)}{s^d_{\blambda}(\rho)}\Big)_{\reg{Q}} \otimes \Big( \frac{I_{\dim(V^r_{\blambda})}}{\dim(V^r_{\blambda})} \Big)_{\reg{Q'}}. 
    \end{equation*}
    Averaged over the outcome of weak Schur sampling, we have prepared the state
    \begin{equation*}
        \sum_{\lambda \vdash n, \ell(\lambda) \leq r}\dim(\lambda) \cdot \ketbra{\lambda\lambda}_{\reg{YY'}} \otimes \ketbra{\epr_{\lambda}}_{\reg{PP'}} \otimes \nu^d_{\lambda}(\rho)_{\reg{Q}} \otimes \Big( \frac{I_{\dim(V^r_{\lambda})}}{\dim(V^r_{\lambda})} \Big)_{\reg{Q'}}.
    \end{equation*}
    Thus, by \Cref{lem:final-formula}, the output after the sixth step is
    \begin{equation*}
        \purifychan^{d,r}(\rho^{\otimes n}) = \E_{\ket{\brho}} \ketbra{\brho}^{\otimes n}. 
    \end{equation*}
    We now turn to efficiency. The gate complexity is dominated by the cost of the two Schur transforms, each of which can be computed to $\delta$ accuracy in diamond distance in time $\poly(n, \log(d), \log(1/\delta))$ \cite{Har05}.
\end{proof}

Having defined the random purification channel, there are a couple of properties of it that we need to establish for our mixed state tomography algorithm.
We begin by looking at its application conditioned on the Young diagram $\blambda$.

\begin{notation}
    We will slightly abuse notation by allowing the input to $\purifychan^{d,r}$ to be a state expressed in the Schur basis. In particular, when the input is $\rho|_\lambda$, we have
    \begin{equation*}
        \purifychan^{d,r}(\rho|_{\lambda})
    = \ketbra{\lambda\lambda}{\lambda\lambda}_{\reg{Y} \reg{Y'}} \otimes \ketbra{\epr_{\lambda}}{\epr_{\lambda}}_{\reg{P} \reg{P'}} \otimes \Big(\frac{\nu^d_{\lambda}(\rho)}{s^d_{\lambda}(\rho)}\Big)_{\reg{Q}}  \otimes \Big(\frac{I_{\dim(V_{\lambda}^{r})}}{\dim(V_{\lambda}^{r})}\Big)_{\reg{Q'}}.
    \end{equation*}
    This will be convenient for us in the context of quasi-purification, where we will want to apply $\purifychan^{d,\ell(\blambda)}$, conditioned on the outcome $\blambda$ observed in weak Schur sampling. 
\end{notation}


We now turn to a couple of properties of the random purification channel that we need to establish for our mixed state tomography algorithm.

\begin{lemma}[Partial trace of the purification channel]\label{lem:partial-trace-restriction}
    \begin{equation*}
        \tr_{\reg{Y}'\reg{P}'\reg{Q}'}(\purifychan^{d,r}(\rho|_{\lambda})) = \rho|_{\lambda}.
    \end{equation*}
\end{lemma}

We will also use the following formula, which shows that  weak Schur sampling, when applied to $\rho^{\otimes n}$, is non-destructive, in that it does not perturb the state. 




\begin{lemma}[Average of $\rho|_{\blambda}$]
    \label{lem:blambda-average}
    $
        \E_{\blambda}[\rho|_{\blambda}]= \schur^d \cdot \rho^{\otimes n} \cdot (\schur^d)^{\dagger}.
    $
\end{lemma}
\begin{proof}
    By \Cref{eq:normal-sw-transform},
    \begin{equation*}
    \schur^d \cdot \rho^{\otimes n} \cdot (\schur^d)^{\dagger}
    = \sum_{\lambda \vdash n, \ell(\lambda) \leq d} \dim(\lambda) \cdot s_{\lambda}(\rho) \cdot \rho|_{\lambda}.
    \end{equation*}
    For each $\lambda$, $\rho|_{\lambda}$ is a density matrix which lives inside the $\lambda$-irrep space.
    Therefore, $\Pr[\blambda = \lambda] = \dim(\lambda) \cdot s_{\lambda}(\rho)$.
    As a result,
    \begin{equation*}
        \E_{\blambda}[\rho|_{\blambda}]
        = \sum_{\lambda \vdash n, \ell(\lambda) \leq d} \Pr[\blambda = \lambda] \cdot \rho|_{\lambda}
        = \schur^d \cdot \rho^{\otimes n} \cdot (\schur^d)^{\dagger}.\qedhere
    \end{equation*}
\end{proof}

\section{Pure state tomography with unentangled measurements}
\label{sec:unentangled}

In this section, we prove \Cref{thm:efficient-pure}, establishing the existence of a gate-efficient and sample-optimal pure state tomography algorithm.
To the best of our knowledge, no result combining these guarantees has been explicitly stated in the literature before, although the special case when $\delta = 1/\poly(d)$ appears in \cite[Appendix C]{HKOT23}.

Throughout this section, we will assume that our pure state is represented by a system of qubits, which allows us to discuss the gate complexity of our algorithm.
A consequence of this is that the overall dimension of our state $d$ will always be a power of two.
This is essentially without loss of generality: any mixed state $\rho$ whose dimension $d$ is not a power of two may be embedded into a Hilbert space of $\lceil \log_2 d \rceil$ qubits. This incurs only constant-factor loss, since the resulting Hilbert space has dimension $2^{\lceil \log_2 d \rceil} \leq 2d$, and since our algorithms will scale polynomially in $d$ in both sample and gate complexity.  

Our starting point is the sample-optimal pure state tomography algorithm given by Guta, Kahn, Kueng, and Tropp~\cite{GKKT20}, which we discussed in \Cref{sec:simpler}. 
Their algorithm uses the uniform POVM $\{d \cdot \ketbra{u} \cdot du\}$, which is continuous and therefore cannot be implemented exactly.
To resolve this, we will replace the uniform POVM with a measurement in a random basis drawn from an (approximate) $t$-design.
These measurements can be implemented efficiently, and the resulting algorithm performs just as well so long as our application only requires using the first $t$ moments of the uniform POVM.


The standard tomography algorithm from \Cref{fig:unentangled-tomography} performs full state tomography when given $n = O(d^3)$ copies of a generic mixed state $\sigma \in \C^{d \times d}$.
It is well-known that this bound only requires using the first two moments of the uniform POVM (see, for example, \cite[Section 5.1]{Wri16}), and so it is common in the literature to see the uniform POVM in this algorithm replaced with a $t$-design for $t = 2$ (or $t = 3$, as is required for some applications).
In our case, however, two features of our algorithm will require us to use $t$-designs with a much higher value of $t$.
First, we want an improved sample complexity of $n = O(d)$ in the case when $\sigma$ is a pure state, and this will require us to take $t = \Omega(\log(d))$ (as was done previously in~\cite{HKOT23}).
Second, we want our sample complexity to have the correct dependence on the parameter $\delta$,
and this will turn out to require taking $t = \Theta(d)$.
We use the following definition for an approximate unitary $t$-design.

\begin{definition}[Approximate unitary $t$-design]
    \label{def:approximate-t-design}
    A distribution $\mathcal{P}$ over $U(d)$ is an \emph{$\epsilon$-approximate unitary $t$-design} if the two mixed unitary channels
    \begin{equation*}
        \mathcal{C}_t: \eta \to \E_{\bV \sim \mathcal{P}} [\bV^{\otimes t} \cdot \eta \cdot \bV^{\dagger, \otimes t}]
        \quad
        \mathrm{and}
        \quad
        \mathcal{H}_t :\eta \to \E_{\bU \sim \mathrm{Haar}} [\bU^{\otimes t} \cdot \eta \cdot \bU^{\dagger, \otimes t}]
    \end{equation*}
    satisfy
    \begin{equation*}
        (1 - \epsilon) \cdot \mathcal{H}_t(\eta) \preceq \mathcal{C}_t(\eta) \preceq (1 + \epsilon) \cdot \mathcal{H}_t(\eta),\qquad
        \mathrm{for~all~PSD}~\eta.
    \end{equation*}
\end{definition}
There is a rich literature on unitary $t$-designs, with many constructions offering different tradeoffs between efficiency, random seed length, and approximation. For the purposes of this section, we will use the construction from O'Donnell, Servedio, and Paredes~\cite{OSP23}, which provides the following guarantee.
\begin{theorem}[Existence of approximate unitary $t$-designs]
    \label{thm:t-design-theorem}
    Let $d$ be a power of two. For any $t$, there exists a distribution $\mathcal{P}$ over $U(d)$ such that:
    \begin{enumerate}
        \item\label{item:apprx-t-design-i} $\mathcal{P}$ is a $\frac{1}{2}$-approximate unitary $t$-design.
        \item\label{item:apprx-t-design-ii} Each unitary in the distribution can be constructed using a $\log_2(d)$-qubit circuit that consists of $\mathrm{poly}(\log_2(d) \cdot t)$ gates from a fixed discrete gate set.
        \item\label{item:apprx-t-design-iii} Computing the circuit that implements a unitary from $\mathcal{P}$ takes $\mathrm{poly}(\log_2(d)\cdot t)$ classical processing.
    \end{enumerate}
\end{theorem}

\begin{proof}
    The proof follows by combining~\cite[Theorem 2.11]{OSP23} with the proof of~\cite[Corollary 1.2]{HLT24}. We include an argument below for completeness.

    A first observation is that it suffices to give a distribution over the unitary subgroup $SU(d)$ that satisfies the constraints of~\Cref{thm:t-design-theorem}. This is because any unitary $U$ can be written as the product of an element $U_0$ of $SU(d)$ with a complex number $z$ of unit norm. Then
    \begin{equation*}
        U^{\otimes t} \cdot \eta \cdot U^{\dagger, \otimes t} = z^t\cdot U_0^{\otimes t} \cdot \eta \cdot U_0^{\dagger, \otimes t} \cdot \bar{z}^t = U_0^{\otimes t} \cdot \eta \cdot U_0^{\dagger, \otimes t}.
    \end{equation*}
    In particular,~\cite[Theorem 2.11]{OSP23} implies that there exists some distribution $\mathcal{P}$ over $SU(d)$ that satisfies~\Cref{item:apprx-t-design-ii,item:apprx-t-design-iii} of~\Cref{thm:t-design-theorem} and
    \begin{equation*}
        \norm{\E_{\bV \sim \mathcal{P}} \big[\bV^{\otimes t} \otimes \overline{\bV}^{\otimes t}\big] - \E_{\bU \sim \mathrm{Haar}} \big[\bU^{\otimes t} \otimes \overline{\bU}^{\otimes t}\big]}_{\infty} \leq \frac{1}{2}\cdot d^{-3t}.
    \end{equation*}
    This bound on the operator norm implies a bound on the diamond distance between the channels $\mathcal{C}_t$ and $\mathcal{H}_t$. Formally, Low in~\cite[Lemma 2.2.14, Item 4]{Low10} gives the following inequality:
    \begin{equation*}
        \norm{\mathcal{C}_t - \mathcal{H}_t}_{\diamond} \leq d^t\cdot \norm{\E_{\bV \sim \mathcal{P}} \big[\bV^{\otimes t} \otimes \overline{\bV}^{\otimes t}\big] - \E_{\bU \sim \mathrm{Haar}} \big[\bU^{\otimes t} \otimes \overline{\bU}^{\otimes t}\big]}_{\infty} \leq \frac{1}{2} \cdot d^{-2t}.
    \end{equation*}
    We conclude that $\mathcal{P}$ is a $\frac{1}{2}$-approximate unitary $t$-design using~\cite[Lemma 3]{BHH16}, which states that if the channels $\mathcal{C}_t$ and $\mathcal{H}_t$ satisfy
    \begin{equation*}
        \norm{\mathcal{C}_t - \mathcal{H}_t}_{\diamond} \leq \epsilon\cdot d^{-2t},
    \end{equation*}
    then $\mathcal{P}$ is an $\epsilon$-approximate unitary $t$-design.
\end{proof}

{
\floatstyle{boxed} 
\restylefloat{figure}
\begin{figure}[H]
Given $n$ copies of $\ketbra{v} \in \C^{d \times d}$:
\begin{enumerate}
    \item Let $t = 6d + 1$.
    \item For each $1 \leq i \leq n$:
        \begin{enumerate}
            \item Sample $\bV$ from an approximate unitary $t$-design $\mathcal{P}$ that satisfies~\Cref{thm:t-design-theorem}.
            \item Rotate the $i$-th copy of $\ketbra{v}$ to $\bV\ketbra{v} \bV^{\dagger}$ and measure in the computational basis to obtain the outcome $\ket{\bj}$ for $\bj \in [d]$.
            \item Set $\widehat{\brho}_i = (d+1) \cdot \bV^{\dagger} \ketbra{\bj} \bV - I_d$.
        \end{enumerate}
    \item Let $\widehat{\brho}_{\mathrm{avg}} = \frac{1}{n} \cdot (\widehat{\brho}_1 + \cdots + \widehat{\brho}_n)$.
    \item Output $\ketbra{\widehat{\bv}}$, where $\ket{\widehat{\bv}}$ is the top eigenvector of $\widehat{\brho}_{\mathrm{avg}}$.
\end{enumerate}
\caption{The~\cite{GKKT20} unentangled measurement tomography algorithm for pure states, made time-efficient.}
\label{fig:gkkt-using-t-design}
\end{figure}
}
\noindent

We now analyze the algorithm of Guta, Kahn, Kueng, and Tropp~\cite{GKKT20}, except with the uniform POVM measurement replaced with the approximate unitary $t$-design from \Cref{thm:t-design-theorem}.
See \Cref{fig:gkkt-using-t-design} for this algorithm, and see the discussion around \Cref{fig:unentangled-tomography} for the original algorithm.
We will follow their proof almost identically.
In particular, we first prove a bound on the $k$-th moment.
\begin{lemma}[{Moment bounds of the time-efficient version of \cite{GKKT20}}]
    \label{lem:k-th-moment-bound}
    Let $t = 6d+1$, and let $\calP$ be the approximate $t$-design guaranteed by~\Cref{thm:t-design-theorem}. For $\bV$ sampled from $\mathcal{P}$ and $\bj \in [d]$ being the outcome we obtain when we measure $\bV\ketbra{v} \bV^{\dagger}$ in the computational basis, define
    \begin{equation*}
        \widehat{\brho} = (d+1)\cdot \bV^{\dagger} \ketbra{\bj} \bV - I_d.
    \end{equation*}
    It then holds for all $\ket{z}$ and all integers $k \geq 2$ that
    \begin{equation}
        \label{eq:our-moment-bound}
        \E\big[\abs{\bra{z}(\widehat{\brho} - \ketbra{v}) \ket{z}}^k\big] \leq 41 \cdot 6^{k-2}k!.
    \end{equation}
\end{lemma}

\begin{proof}
    Our first observation is that
    \begin{equation*}
        \abs{\bra{z}(\widehat{\brho} - \ketbra{v}) \ket{z}}^k \leq \norm{\widehat{\brho} - \ketbra{v}}_{\infty}^k \leq (d+1)^k.
    \end{equation*}
    Whenever $k \geq 6d$, the Stirling approximation formula implies that $(d+1)^k \leq (k/e)^k \leq k! \leq 41 \cdot 6^{k-2}k!$, which means that the desired bound holds. Therefore, we will assume that $k \leq 6d$ for the remainder of this proof.
    
    We define $\ket{\bu} = \bV^{\dagger} \ket{\bj}$ for notational brevity.
    Our starting point is an observation of~\cite{GKKT20} that for a fixed $\ket{z}$, one can write
    \begin{align*}
        \bs_z &= \bra{z}(\widehat{\brho} - \ketbra{v}) \ket{z} \\
        &= \bra{z}\big((d+1) \ketbra{\bu} - I_d - \ketbra{v}\big) \ket{z} \\
        &= (d+1)\bra{\bu} \ketbra{z} \ket{\bu} - (1 + \abs{\braket{v}{z}}^2) \\
        &= (d+1)\bra{\bu} B \ket{\bu},
    \end{align*}
    for $B$ equal to $\ketbra{z} - \frac{1 + \abs{\braket{v}{z}}^2}{d+1} \cdot I_d$. This allows us to upper bound $\E[\abs{\bs_z}^k]$ by
    \begin{equation*}
        \E[\abs{\bs_z}^k] = \E[\abs{(d+1) \bra{\bu} B \ket{\bu}}^k] = (d+1)^k\E[\abs{\tr(\ketbra{\bu} \cdot B)}^k] \leq (d+1)^k \E[\tr(\ketbra{\bu} \cdot \abs{B})^k].
    \end{equation*}
    Here, $|B| = \sqrt{B^{\dagger}B}$ stands for the PSD matrix obtained from $B$ by taking the absolute value of its eigenvalues.
    We now expand $\ketbra{\bu}$ using $\bV$ to obtain
    \begin{align*}
        \E[\abs{\bs_z}^k] &\leq (d+1)^k \E[\tr(\ketbra{\bu} \cdot \abs{B})^k] \\
        &= (d+1)^k \E_{\bV \sim \mathcal{P}} \bigg[\sum_{j=1}^d \tr(\bV^{\dagger} \ketbra{j} \bV \cdot \abs{B})^k \cdot \tr(\ketbra{v} \cdot \bV^{\dagger} \ketbra{j} \bV)\bigg] \\
        &= (d+1)^k \E_{\bV \sim \mathcal{P}} \bigg[\sum_{j=1}^d \tr(\bV^{\dagger, \otimes k+1} \cdot \ketbra{j}^{\otimes k+1} \cdot \bV^{\otimes k+1} \cdot \abs{B}^{\otimes k} \otimes \ketbra{v})\bigg] \\
        &= (d+1)^k \tr(\sum_{j=1}^d \E_{\bV \sim \mathcal{P}} \big[\bV^{\dagger, \otimes k+1} \cdot \ketbra{j}^{\otimes k+1} \cdot \bV^{\otimes k+1}\big] \cdot \abs{B}^{\otimes k} \otimes \ketbra{v}).
    \end{align*}
    Let us use $A$ to denote the first factor in the trace:
    \begin{equation*}
        A = \sum_{j=1}^d \E_{\bV \sim \mathcal{P}} \big[\bV^{\dagger, \otimes k+1} \cdot \ketbra{j}^{\otimes k+1} \cdot \bV^{\otimes k+1}\big].
    \end{equation*}
    We bound the expression above using H\"older's inequality:
    \begin{equation*}
        \E[\abs{\bs_z}^k] \leq (d+1)^k \tr(A \cdot \abs{B}^{\otimes k} \otimes \ketbra{v}) \leq (d+1)^k \cdot \norm{A}_{\infty} \cdot \norm{\abs{B}}_1^k \cdot \norm{\ketbra{v}}_1 \leq (d+1)^k \cdot \norm{A}_{\infty} \cdot \norm{\abs{B}}_1^k.
    \end{equation*}
    The trace norm of $|B|$ is at most
    \begin{equation*}
        \norm{|B|}_1 = \norm{B}_1 \leq 1 + \frac{d}{d+1}\cdot (1 + \abs{\braket{v}{z}}^2) \leq 3.
    \end{equation*}
    We proceed to bounding the operator norm of $A$ using the $t$-design property.
    Since we are in the $k \leq 6d$ regime, then for $t = 6d+1$, the random unitary $\bV$ is sampled from an approximate $t$-design such that $t \geq k+1$. In particular:
    \begin{align*}
        \norm{A}_{\infty}
        &\leq \sum_{j=1}^d  \norm{\E_{\bV \sim \mathcal{P}} [\bV^{\dagger, \otimes k+1} \cdot \ketbra{j}^{\otimes k+1} \cdot \bV^{\otimes k+1}]}_{\infty} \\
        &\leq \sum_{j=1}^d \Big(1 + \frac{1}{2}\Big) \cdot\Big \lVert\E_{\bU \sim \mathrm{Haar}} \big[ \bU^{\dagger, \otimes k+1} \cdot \ketbra{j}^{\otimes k+1} \cdot \bU^{\otimes k+1} \big] \Big \rVert_{\infty} \tag{\Cref{thm:t-design-theorem}} \\
        &= \sum_{j=1}^d \frac{3}{2}\cdot \Big \lVert \frac{1}{d[k+1]} \cdot \Pi_{\mathrm{sym}}^{k+1, d} \Big \rVert_{\infty} \\
        &= \frac{3d}{2\cdot d[k+1]}.
    \end{align*}
    Recall that
    \begin{equation*}
        d[k+1] = \binom{d+k}{k+1} = \frac{(d+k) \dots (d+1)d}{(k+1)!} \geq (d+1)^k \cdot \frac{d}{(k+1)!}.
    \end{equation*}
    Hence, we conclude that
    \begin{equation*}
        \E[\abs{\bs_z}^k] \leq (d+1)^k \cdot \frac{3d}{2\cdot d[k+1]} \cdot 3^k \leq 
        \frac{3}{2} \cdot 3^k \cdot (k+1)! \leq \frac{3}{2} \cdot 3^k \cdot (3\cdot 2^{k-2} \cdot k!)
        \leq 41 \cdot 6^{k-2} k!.
    \end{equation*}
    In the third inequality we used the fact that $(k+1)! = (k+1) \cdot k! \leq 3\cdot 2^{k-2} \cdot k!$, whenever $k \geq 2$.
\end{proof}

Given the above moment bounds, we are ready to prove a concentration result for the operator norm of the difference of our estimator from $\ketbra{v}$. For this, we will need to use a covering net argument, which we define below.

\begin{definition}[Covering net]
    Let $T$ be a set with metric $d(\cdot, \cdot)$. A subset $S$ of $T$ is a \emph{covering net with radius $\theta$} if every point $x$ in $T$ is at most $\theta$ away from some point in the covering net, i.e., there exists $y \in S$ such that $d(x, y) \leq \theta$.
\end{definition}

\begin{lemma}[{Concentration bounds of the time-efficient version of \cite{GKKT20}}]
    \label{lem:concentration-in-op-norm-gkkt}
    It holds that
    \begin{equation*}
        \Pr\big[\norm{\widehat{\brho}_{\mathrm{avg}} - \ketbra{v}}_{\infty} \geq q\big]
        \leq 2 \cdot \exp((4\log 3)d-\frac{nq^2}{704}).
    \end{equation*}
\end{lemma}

\begin{proof}
    We first rewrite the operator norm as
    \begin{equation*}
        \norm{\widehat{\brho}_{\mathrm{avg}} - \ketbra{v}}_{\infty}
        = \max_{\ket{z}} \abs{\bra{z} (\widehat{\brho}_{\mathrm{avg}} - \ketbra{v}) \ket{z}}.
    \end{equation*}
    We will use a covering net argument to compute this maximum over all pure states. It is well-known that the set of $d$-dimensional pure states admits a covering net, since each state corresponds, via a distance-preserving mapping, to a point on the real Euclidean sphere $\mathbb{S}^{2d-1}$. In particular, for any covering net $S = \{\ket{z_i}\}_{i \in [N]}$ of radius $\theta \in [0, 1/2)$, it holds that (see, for example,~\cite[Lemma 4.4.2]{Ver18}):
    \begin{equation}
        \label{eq:net-argument}
        \norm{\widehat{\brho}_{\mathrm{avg}} - \ketbra{v}}_{\infty} \leq \frac{1}{1 - 2\theta} \cdot \max_{i \in [N]} \abs{\bra{z_i} (\widehat{\brho}_{\mathrm{avg}} - \ketbra{v}) \ket{z_i}}.
    \end{equation}
    Moreover,~\cite[Corollary 4.2.11]{Ver18} implies that such a covering net with radius $\theta$ has size at most $(1 + 2/\theta)^{2d}$. We set $\theta = 1/4$, which means that $N = \abs{S} \leq 9^{2d} = 3^{4d}$.
    
    For any fixed $\ket{z}$, the right-hand side of~\Cref{eq:net-argument} can be written as a sum of $n$ independent and identically distributed copies of a mean-zero random variable:
    \begin{equation*}
        \abs{\bra{z} (\widehat{\brho}_{\mathrm{avg}} - \ketbra{v}) \ket{z}} = 
        \Big|\frac{1}{n}\sum_{i=1}^n \bra{z} (\widehat{\brho}_i - \ketbra{v}) \ket{z}\Big|.
    \end{equation*}
    In~\Cref{lem:k-th-moment-bound} we showed that the $k$-th moment of each such variable is bounded above by $41 \cdot 6^{k-2} k!$. This bound allows us to use the Bernstein inequality, whose statement is given below.
    \begin{theorem}[{Bernstein inequality~\cite[Theorem 7.30]{FR13}}]
        Consider $n$ independent samples $\bs_1, \dots, \bs_n$ of the real random variable $\bs$, which has mean zero, and whose $k$-th moment is bounded by
        \begin{equation*}
            \E[\abs{\bs}^k]\leq k!R^{k-2}\sigma^2/2,
        \end{equation*}
        for all integers $k \geq 2$, and some positive constants $R, \sigma^2$. Then, for all $q > 0$,
        \begin{equation*}
            \Pr\bigg[\Big|\sum_{i=1}^n \bs_i\Big| \geq q\bigg] \leq 2 \exp(-\frac{q^2/2}{n\sigma^2 + Rq}).
        \end{equation*}
    \end{theorem}
    \noindent Plugging in $R = 6$ and $\sigma^2 = 82$ in the inequality above implies that for all $\ket{z}$ and $0 \leq q \leq 1$,
    \begin{equation}
        \label{eq:apply-bernstein}
        \Pr\big[\big|\bra{z} (\widehat{\brho}_{\mathrm{avg}} - \ketbra{v}) \ket{z}\big| \geq q\big] \leq 2\exp(-\frac{q^2n^2/2}{82n + 6qn}) \leq 2\exp(-\frac{nq^2}{176}).
    \end{equation}
    Applying the union bound inequality over the covering net $S$ of the pure states, we conclude that 
    \begin{align*}
        \Pr\big[\norm{\widehat{\brho}_{\mathrm{avg}} - \ketbra{v}}_{\infty} \geq q\big]
        &\leq \Pr\Big[\max_{i \in [N]} \big|\bra{z_i} (\widehat{\brho}_{\mathrm{avg}} - \ketbra{v}) \ket{z_i}\big| \geq \frac{q}{2}\Big] \tag{\Cref{eq:net-argument}} \\
        &= \Pr\Big[\bigvee_{i \in [N]} \big|\bra{z_i} (\widehat{\brho}_{\mathrm{avg}} - \ketbra{v}) \ket{z_i}\big| \geq \frac{q}{2}\Big]  \\
        &\leq 2N \cdot\exp(-\frac{nq^2}{704}) \tag{\Cref{eq:apply-bernstein}} \\
        &\leq 2 \cdot \exp((4\log 3)d-\frac{nq^2}{704}).
    \end{align*}
    This completes the proof.
\end{proof}

\begin{theorem}[\Cref{thm:efficient-pure}, restated]
    There is an algorithm which, given
    \begin{equation*}
        n = O\Big(\frac{d + \log(1/\delta)}{\epsilon}\Big)
    \end{equation*}
    copies of a pure state $\ket{v} \in \C^d$,
    outputs a pure state $\ket{\widehat{\bv}} \in \C^d$ such that $\abs{\braket{\widehat{\bv}}{v}}^2 \geq 1 - \epsilon$ with probability at least $1 - \delta$.
    Furthermore, this algorithm can be implemented in $\poly(n)$ time and performs independent measurements across the copies of $\ket{v}$.
\end{theorem}

\begin{proof}
    The algorithm is outlined in~\Cref{fig:gkkt-using-t-design}. The claimed running time follows from~\Cref{thm:t-design-theorem}, since the unitary $t$-design can be implemented using $\poly(\log(d)\cdot t) = \poly(d)$ quantum gates and classical processing. Moreover, the number of samples $n = \poly(d, 1/\epsilon, \log(1/\delta))$ is also polynomial in these parameters.

    It remains to show that the output pure state satisfies $\abs{\braket{\widehat{\bv}}{v}}^2 \geq 1 - \epsilon$ with probability at least $1 - \delta$. By setting $q = \sqrt{\epsilon}$ and the number of samples equal to $n = O((d + \log(1/\delta))/\epsilon)$ in the statement of~\Cref{lem:concentration-in-op-norm-gkkt}, we conclude that
    \begin{equation*}
        \Pr\big[\norm{\widehat{\brho}_{\mathrm{avg}} - \ketbra{v}}_{\infty} \geq \sqrt{\epsilon}\big] \leq \delta.
    \end{equation*}
    Let us now restrict our attention to the case when $\norm{\widehat{\brho}_{\mathrm{avg}} - \ketbra{v}}_{\infty} \leq \sqrt{\epsilon}$, which happens with probability at least $1 - \delta$.

    Since the operator norm is unitarily invariant, the expression $\norm{\widehat{\brho}_{\mathrm{avg}} - \ketbra{w}}_{\infty}$ is minimized when $\ket{w}$ is the top eigenvector $\ket{\widehat{\bv}}$ of $\widehat{\brho}_{\mathrm{avg}}$. This follows from a theorem of Mirsky (see~\cite[Corollary 7.4.9.3]{HJ13}), which states that for any two Hermitian matrices $A, B$:
    \begin{equation*}
        \norm{A - B}_{\infty} \geq \norm{\mathrm{spec}(A) - \mathrm{spec}(B)}_{\infty},
    \end{equation*}
    where on the right-hand side, the matrices $\mathrm{spec}(A), \mathrm{spec}(B)$ are diagonal with entries equal to the eigenvalues of $A,B$ in nonincreasing order, respectively. By substituting $\widehat{\brho}_{\mathrm{avg}}$ for $A$ and $\ketbra{v}$ for $B$, we observe that
    \begin{equation*}
        \norm{\widehat{\brho}_{\mathrm{avg}} - \ketbra{v}}_{\infty} \geq 
        \norm{\mathrm{spec}(\widehat{\brho}_{\mathrm{avg}}) - \mathrm{spec}(\ketbra{v})}_{\infty}
        =
        \Big\lVert \sum_{i=1}^d \blambda_i \ketbra{i} - \ketbra{1}\Big\rVert_{\infty}
        =
        \norm{\widehat{\brho}_{\mathrm{avg}} - \ketbra{\widehat{\bv}}}_{\infty},
    \end{equation*}
    where we use $\blambda_i$ for the $i$-th largest eigenvalue of $\widehat{\brho}_{\mathrm{avg}}$. The last equality follows from the unitary invariance of the norm.
    As a result, the triangle inequality implies that
    \begin{equation*}
        \norm{\ketbra{\widehat{\bv}} - \ketbra{v}}_{\infty}
        \leq \norm{\ketbra{\widehat{\bv}} - \widehat{\brho}_{\mathrm{avg}}}_{\infty}
        + \norm{\widehat{\brho}_{\mathrm{avg}} - \ketbra{v}}_{\infty}
        \leq 2\sqrt{\epsilon}.
    \end{equation*}
    Finally, since $\ketbra{\widehat{\bv}} - \ketbra{v}$ only has at most two nonzero eigenvalues,
    \begin{equation*}
        \norm{\ketbra{\widehat{\bv}} - \ketbra{v}}_1 \leq 2 \cdot \norm{\ketbra{\widehat{\bv}} - \ketbra{v}}_{\infty} \leq 4\sqrt{\epsilon}.
    \end{equation*}
    From the relationship between the trace distance and squared overlap for pure states~\cite[Eq.\ (1.186)]{Wat18} we conclude that
    \begin{equation*}
        \abs{\braket{\widehat{\bv}}{v}}^2 = 1 - \Big(\frac{1}{2}\norm{\ketbra{\widehat{\bv}} - \ketbra{v}}_1\Big)^2 \geq 1 - 4\epsilon.
    \end{equation*}
    The theorem statement then follows by adjusting the parameter $\epsilon$ by a constant.
\end{proof}

\section{Pure state tomography with entangled measurements} \label{sec:entangled}

In this section, we discuss and analyze two pure state tomography algorithms, both of which use the optimal choice of measurement for this task.
The first is Hayashi's algorithm $\hayashi$ \cite{Hay98}, which performs this measurement and outputs the naive corresponding estimator, which is itself also a pure state.
The second is the algorithm $\gps$  of Grier, Pashayan, and Schaeffer \cite{GPS24}, which adjusts Hayashi's estimator in order to remove its bias.

We begin with Hayashi's algorithm, where we directly show that it achieves the optimal sample complexity of $n = O((d + \log(1/\delta))/\eps)$ copies for pure state tomography.
This will not be used later on; we include it because $\mixed(\hayashi)$ gives arguably the cleanest proof of sample-optimal mixed state tomography, even among results which do not achieve the optimal $\delta$ dependence.

We then analyze the moments of the Grier--Pashayan--Schaeffer estimator.
In \Cref{sec:unbiased-estimator}, we will study $\mixed(\gps)$ and $\mixed^+(\gps)$; the moments of these mixed state estimators will follow directly from the moments of $\gps$.

\subsection{Hayashi's algorithm}

Let us first recall Hayashi's algorithm. When performing tomography on $n$ copies of a pure state $\ket{\psi} \in \C^d$,
the input $\ket{\psi}^{\otimes n}$ is an element of the symmetric subspace $\lor^n \C^d$.
This means that a pure state tomography algorithm's measurement operators need only be specified on the symmetric subspace.
Motivated by this, Hayashi~\cite{Hay98} introduced the following natural pure state tomography algorithm:
simply perform the POVM
\begin{equation}
    \{d[n] \cdot \ketbra{u}^{\otimes n} \cdot du\}
\end{equation}
and output the pure state $\ket*{\bv}$ that it returns (\Cref{fig:entangled-tomography}).
This is indeed a valid POVM on the symmetric subspace, as
\begin{equation*}
    \int_{\ket{u}} d[n] \cdot \ketbra{u}^{\otimes n} \cdot du
    = d[n] \cdot \E_{\ket{\bu} \sim \mathrm{Haar}} \ketbra{\bu}^{\otimes n}
    = \Pi_{\mathrm{sym}}^{n, d}
\end{equation*}
due to \Cref{eq:pi-sym}.
Hayashi showed that $\abs{\braket*{\bv}{\psi}}^2 \geq 1- \epsilon$ with probability 99\% when $n = O(d/\epsilon)$. To see this, note that
\begin{align}
    \E\abs{\braket*{\bv}{\psi}}^2
    &= d[n] \cdot \int_{\ket{u}} \tr\Big(\ketbra{u}^{\otimes n} \cdot \ketbra{\psi}^{\otimes n}\Big) \cdot \abs{\braket*{u}{\psi}}^2 \cdot du\nonumber\\
    &= d[n] \cdot \int_{\ket{u}} \tr\Big( \ketbra{u}^{\otimes n + 1} \cdot \ketbra{\psi}^{\otimes n+1}\Big) \cdot du\nonumber\\
    &= d[n] \cdot  \tr\Big(\Big(\int_{\ket{u}} \ketbra{u}^{\otimes n + 1} \cdot du \Big) \cdot \ketbra{\psi}^{\otimes n+1}\Big)\nonumber\\
    &= d[n] \cdot  \tr\Big(\Big(\tfrac{1}{d[n+1]} \cdot \Pi_{\mathrm{sym}}^{n+1} \Big) \cdot \ketbra{\psi}^{\otimes n+1}\Big)\nonumber\\
    &= \frac{d[n]}{d[n+1]}\nonumber\\
    &= \frac{n+1}{n+d} = 1 - \frac{d-1}{n+d}.\label{eq:standard-lb}
\end{align}
This is $1 - \epsilon/100$ when $n = O(d/\epsilon)$. The claim now follows from an application of Markov's inequality.

Next,
we upgrade this to a high probability bound for Hayashi's estimator.
\begin{proposition}[{Hayashi's algorithm with high probability; \Cref{prop:hayashi}, restated}]
Given $n$ copies of a pure state $\ket{\psi} \in \C^d$,
suppose Hayashi's algorithm returns the state $\ket*{\bv}$. Then $\abs{\braket*{\bv}{\psi}}^2 \geq 1 - \epsilon$ with probability at least $1 - \delta$ when
\begin{equation*}
    n = O\Big(\frac{d + \log(1/\delta)}{\epsilon}\Big).
\end{equation*}
\end{proposition}
\begin{proof}
It suffices to show that the random variable $\bx \coloneqq \abs{\braket*{\bv}{\psi}}^2$ exhibits strong concentration about its mean.
Let us write $p(x)$ for the probability density function of this random variable.
Noting that $\ket*{\bv} = \ket{u}$ with measure
\begin{equation*}
    d[n] \cdot \tr(\ketbra{u}^{\otimes n} \cdot \ketbra{\psi}^{\otimes n}) \cdot du
    =
    d[n] \cdot \abs{\braket*{u}{\psi}}^{2n} \cdot du,
\end{equation*}
it follows that $p(x) = d[n] \cdot x^n \cdot \mu(x)$, where $\mu$ is the probability density function of $\by \coloneqq \abs{\braket{\bu}{\psi}}^2$ when $\ket{\bu} \sim \C^d$ is a Haar random vector.
It is known that $\by$ is distributed as a $\mathrm{Beta}(1, d-1)$ random variable~\cite[Eq.\ (7)]{ZS00},
and so $\mu(x) = (d-1)(1-x)^{d-2}$.
As a result,
\begin{equation*}
    p(x) = (d-1)\cdot d[n] \cdot x^n (1-x)^{d-2},
\end{equation*}
and so $\bx$ obeys the $\mathrm{Beta}(n+1, d-1)$ distribution.
Now we can apply \cite[Theorem 1]{Sko23}\footnote{We are applying \cite[Theorem 1]{Sko23} with $\alpha = n+1$ and $\beta = d-1$, and we are interested in the $\alpha \geq \beta$ case of the lower-tail bound. In this case, his result states that
\begin{equation}\label{eq:can-i-cite-in-a-footnote}
    \Pr[\bx \leq \E[\bx] - \epsilon]
    \leq \mathrm{exp}\Big(- \frac{\epsilon^2}{2(v + \frac{c\epsilon}{3})}\Big),
\end{equation}
where
\begin{equation*}
    v = \frac{\alpha \beta}{(\alpha+\beta)^2 (\alpha+\beta+1)} \quad \text{and} \quad c = \frac{2(\beta  - \alpha)}{(\alpha+\beta) (\alpha+\beta+2)}.
\end{equation*}
There is actually a slight bug in this statement, however:
since $\alpha \geq \beta$, the variable $c$ is actually nonpositive, which means that the $c \epsilon/3$ in the denominator of \Cref{eq:can-i-cite-in-a-footnote} is (erroneously) nonpositive.
Skorski derived this statement by reducing it to his upper-tail bound,
and we can derive the correct lower-tail bound by carrying out this reduction more carefully.
In particular, if we set $\bx' \coloneqq 1 - \bx$, then $\bx'$ is a $\mathrm{Beta}(\beta, \alpha)$ random variable,
and so using his upper-tail bound, we obtain
\begin{equation*}
    \Pr[\bx \leq \E[\bx] - \epsilon] = \Pr[\bx' \geq \E[\bx'] + \epsilon] 
    \leq \mathrm{exp}\Big(- \frac{\epsilon^2}{2(v - \frac{c\epsilon}{3})}\Big) = \mathrm{exp}\Big(- \frac{\epsilon^2}{2(v + \frac{|c|\epsilon}{3})}\Big).
\end{equation*}
This is the corrected bound that we use above.
}, which states that so long as $n+1 \geq d-1$,
\begin{equation*}
    \Pr[\bx \leq \E[\bx] - \epsilon]
    \leq \mathrm{exp}\Big(- \frac{\epsilon^2}{2(v + \frac{c\epsilon}{3})}\Big),
\end{equation*}
where
\begin{align*}
    v = \frac{(n+1)(d-1)}{(n+d)^2(n+d+1)} \leq \frac{d}{n^2}
    \qquad
    \text{and}
    \qquad
    c = \frac{2(n-d+2)}{(n+d)(n+d+2)} \leq \frac{2}{n}.
\end{align*}
Hence, we have that 
\begin{equation*}
    \Pr[\bx \leq \E[\bx] - \epsilon]
    \leq \mathrm{exp}\Big(- \frac{\epsilon^2}{2(d/n^2 + (2/3) \cdot\epsilon/n)}\Big),
\end{equation*}
We will apply this bound to the case when $n = O((d + \log(1/\delta))/\epsilon)$.
For this setting of parameters, we have that $d/n^2 \leq (1/3) \cdot \epsilon/n$; furthermore, we have that $\E[\bx] \geq 1 - \epsilon$ from \Cref{eq:standard-lb} above. Putting these together, for this range of $n$,
\begin{equation*}
    \Pr[\bx \leq (1-\epsilon) - \epsilon]
    \leq \mathrm{exp}\Big(- \frac{\epsilon^2}{2 \cdot\epsilon/n}\Big)
    = \exp(-\epsilon n/2) \leq \delta.
\end{equation*} 
This completes the proof.
\end{proof}

This matches the sample complexity bound for the unentangled pure state tomography algorithm $\gkkt$ from \Cref{thm:efficient-pure},
and so it can be used to give a mixed state tomography algorithm whose sample complexity matches \Cref{thm:efficient-mixed}.
However, unlike $\gkkt$,
which can be implemented efficiently,
we do not know of any efficient implementations of Hayashi's measurement,
and so this only produces a mixed state learning algorithm which is sample-optimal but not necessarily efficient.

\subsection{Grier, Pashayan, and Schaeffer's algorithm}\label{sec:gps}

One downside of Hayashi's algorithm is that its output $\ketbra*{\bv}$ is not an unbiased estimator for the true state $\ketbra{\psi}$.
This is because $\ketbra{\psi}$ is a pure state but the expectation $\E\ketbra*{\bv}$ will necessarily be mixed.
In addition, as \Cref{eq:standard-lb} demonstrates, this expectation will have poor overlap with the true state unless $n \gg d$.
This limits the usefulness of Hayashi's algorithm to problems such as shadow tomography where one might hope to use significantly fewer than $d$ copies of the state.

To address this, Grier, Pashayan, and Schaeffer~\cite{GPS24} noted that one can correct for the bias in Hayashi's algorithm via the simple modification
\begin{equation*}
\widehat{\sigma}_{\bv} \coloneqq \frac{d+n}{n} \cdot \ketbra*{\bv} - \frac{1}{n} \cdot I.
\end{equation*}
In other words, the expectation of this estimator is $\E \widehat{\sigma}_{\bv} = \ketbra{\psi}$, and so it is an unbiased estimator for pure state tomography.
They used this to derive bounds on the sample complexity of classical shadows when the input is a pure state.

For our applications, we will need to compute the first and second moments of $\widehat{\sigma}_{\bv}$.
To do so, let us first define the following moment operators.

\begin{definition}[Moment operator]
    Given integers $n \geq 1$, $d \geq 1$, and $k \geq 1$, we define the \emph{$k$-th moment operator} to be the following operator in $(\C^{d \times d})^{\otimes (n+k)}$:
    \begin{equation*}
        M_{\mathrm{mom}}^{(k)} \coloneqq d[n] \cdot \int_{\ket{u}} \ketbra{u}^{\otimes n} \otimes \widehat{\sigma}_{u}^{\otimes k} \cdot du.
    \end{equation*}
\end{definition}

Our next proposition shows that the $k$-th moment operator can be used to calculate the $k$-th moment of the Grier--Pashayan--Schaeffer algorithm.
For a later application, we will state this proposition for the more general case when their algorithm is performed on a generic mixed state $\psi_{\mathrm{sym}}$ inside the symmetric subspace.
However, for pure state tomography, it suffices to consider the case that $\psi_{\mathrm{sym}} = \ketbra{\psi}^{\otimes n}$ for a pure state $\ket{\psi} \in \C^d$.
\begin{proposition}[Computing moments with the moment operator]\label{prop:moment}
    Let $\psi_{\mathrm{sym}}$ be a mixed state on $\lor^n \C^d$.
    If $\ket{\bv}$ is the output of Hayashi's measurement when performed on $\psi_{\mathrm{sym}}$, then
    \begin{equation*}
        \E \widehat{\sigma}_{\bv}^{\otimes k} = \tr_{\reg{1}\dots\reg{n}}(M_{\mathrm{mom}}^{(k)} \cdot \psi_{\mathrm{sym}} \otimes I_d^{\otimes k}).
    \end{equation*}
\end{proposition}
\begin{proof}
    This follows by direct computation:
    \begin{align*}
        \tr_{\reg{1}\dots\reg{n}}(M_{\mathrm{mom}}^{(k)} \cdot \psi_{\mathrm{sym}} \otimes I_d^{\otimes k})
        & = \tr_{\reg{1}\dots\reg{n}}\Big(\Big(d[n] \cdot \int_{\ket{u}} \ketbra{u}^{\otimes n} \otimes \widehat{\sigma}_{u}^{\otimes k} \cdot du\Big) \cdot \psi_{\mathrm{sym}} \otimes I_d^{\otimes k}\Big)\\
        &= d[n] \cdot \int_{\ket{u}} \tr(\ketbra{u}^{\otimes n} \cdot \psi_{\mathrm{sym}}) \cdot \widehat{\sigma}_{u}^{\otimes k} \cdot du\\
        &= \E_{\ket{\bv}} \widehat{\sigma}_{\bv}^{\otimes k}.\qedhere
    \end{align*}
\end{proof}

Thus, to compute the $k$-th moment of this unbiased estimator, it suffices to compute the corresponding moment operator.
We will do so for the first and second moments.

\begin{lemma}[Helper lemma] \label{lem:Helper_lemma}
Let $n \geq 1$, $d \geq 2$, and $k \geq 1$. Then
\begin{equation*}
d[n] \cdot (d+n)^{\uparrow k} \cdot \int_{ \ket{u}} \ketbra{u}^{\otimes n+k} \cdot du
= (e + X_{n+k}) \cdots (e + X_{n+1}) \cdot (\Pi_{\mathrm{sym}}^{n} \otimes I^{\otimes k}),
\end{equation*}
where $a^{\uparrow b} = a (a+1) \cdots (a + b-1)$ is the rising factorial.
\end{lemma}
\begin{proof}
    Using $\int_{\ket{u}} \ketbra{u}^{n +k} \cdot du = d[n+k]^{-1} \cdot \Pi_{\mathrm{sym}}^{n+k}$, this is equivalent to the statement
    \begin{equation*}
        \frac{d[n]}{d[n+k]} \cdot (d+n)^{\uparrow k} \cdot \Pi_{\mathrm{sym}}^{n+k} = (e + X_{n+k}) \cdots (e + X_{n+1}) \cdot (\Pi_{\mathrm{sym}}^{n} \otimes I^{\otimes k}).
    \end{equation*}
    This is true because
    \begin{align*}
        &\frac{d[n]}{d[n+k]} \cdot (d+n)^{\uparrow k} \cdot \Pi_{\mathrm{sym}}^{n+k}\\
        ={}& \frac{d[n]}{d[n+k]} \cdot (d+n)^{\uparrow k} \cdot \Big(\frac{e+X_{n+k}}{n+k}\Big) \cdots \Big(\frac{e+X_{n+1}}{n+1}\Big) \cdot \Pi_{\mathrm{sym}}^{n} \otimes I^{\otimes k} \tag{by \Cref{prop:pi-sym-recurrence}}\\
        ={}& \frac{d[n]}{d[n+k]} \cdot \frac{(d+n)^{\uparrow k}}{(n+1)^{\uparrow k}} \cdot (e+X_{n+k}) \cdots (e+X_{n+1}) \cdot \Pi_{\mathrm{sym}}^{n} \otimes I^{\otimes k}\\
        ={}& (e+X_{n+k}) \cdots (e+X_{n+1}) \cdot \Pi_{\mathrm{sym}}^{n} \otimes I^{\otimes k}. \tag{by iterating \Cref{eq:ratio}}
    \end{align*}
    This completes the proof.
\end{proof}

\begin{lemma}[The first moment operator]\label{lem:first-moment-operator}
    Given $n \geq 1$ and $d \geq 2$,
    \begin{equation*}
        M_{\mathrm{mom}}^{(1)} = \frac{1}{n} \cdot X_{n+1} \cdot (\Pi_{\mathrm{sym}}^n \otimes I).
    \end{equation*}
\end{lemma}
\begin{proof}
    We calculate using the helper lemma (\Cref{lem:Helper_lemma}):
    \begin{align*}
        M_{\mathrm{mom}}^{(1)}
        & = d[n] \cdot \int_{\ket{u}} \ketbra{u}^{\otimes n} \otimes \widehat{\sigma}_{u} \cdot du\\
        & = d[n] \cdot \int_{\ket{u}} \ketbra{u}^{\otimes n} \otimes \Big(\frac{d+n}{n} \cdot \ketbra{u} - \frac{1}{n} \cdot I\Big) \cdot du\\
        & =d[n] \cdot\frac{d+n}{n} \cdot  \int_{\ket{u}} \ketbra{u}^{\otimes n + 1} \cdot du - d[n] \cdot\frac{1}{n} \cdot \int_{\ket{u}} \ketbra{u}^{\otimes n} \otimes I \cdot du\\
        & = \frac{1}{n}  \cdot (e + X_{n+1})\cdot \Pi_{\mathrm{sym}}^n \otimes I - \frac{1}{n} \cdot \Pi_{\mathrm{sym}}^n \otimes I\\
        & = \frac{1}{n} \cdot X_{n+1} \cdot (\Pi_{\mathrm{sym}}^n \otimes I).
    \end{align*}
    This completes the proof.
\end{proof}

\begin{lemma}[The second moment operator]\label{lem:second-moment-operator}
    Given $n \geq 1$ and $d \geq 2$,
    \begin{equation*}
        M_{\mathrm{mom}}^{(2)} = \frac{1}{n^2} \cdot (X_{n+2}X_{n+1} + (n+1, n+2)) \cdot (\Pi_{\mathrm{sym}}^n \otimes I^{\otimes 2}) - \mathrm{Lower}_{\mathrm{mom}}.
    \end{equation*}
    where
    \begin{equation*}
        \mathrm{Lower}_{\mathrm{mom}} = \frac{d+n}{n^2} \cdot d[n] \cdot \int_{\ket{u}} \ketbra{u}^{\otimes n + 2} \cdot du.
    \end{equation*}
\end{lemma}
\begin{proof}
    We calculate:
    \begin{align*}
        M_{\mathrm{mom}}^{(2)}
        & = d[n] \cdot \int_{\ket{u}} \ketbra{u}^{\otimes n} \otimes \widehat{\sigma}_{u}^{\otimes 2} \cdot du\\
        & = d[n] \cdot \int_{\ket{u}} \ketbra{u}^{\otimes n} \otimes \Big(\frac{d+n}{n} \cdot \ketbra{u} - \frac{1}{n} \cdot I\Big)^{\otimes 2} \cdot du\\
        & = d[n] \cdot \int_{\ket{u}} \ketbra{u}^{\otimes n} \otimes \Big(\frac{(d+n)^2}{n^2} \cdot \ketbra{u}^{\otimes 2} - \frac{d+n}{n^2}\cdot \ketbra{u} \otimes I - \frac{d+n}{n^2}\cdot I \otimes \ketbra{u} + \frac{1}{n^2} \cdot I \otimes I\Big) \cdot du.
    \end{align*}
    This now splits into four terms. The first term we divide further into a main term and a lower-order term using $(d+n)^2 = (d+n)(d+n+1) - (d+n) = (d+n)^{\uparrow 2}- (d+n)$:
    \begin{align*}
        &\frac{(d+n)^2}{n^2} \cdot d[n] \cdot \int_{\ket{u}} \ketbra{u}^{\otimes n + 2} \cdot du
        \\={}& \frac{(d+n)^{\uparrow 2}}{n^2} \cdot d[n] \cdot \int_{\ket{u}} \ketbra{u}^{\otimes n + 2} \cdot du - \frac{d+n}{n^2} \cdot d[n] \cdot \int_{\ket{u}} \ketbra{u}^{\otimes n + 2} \cdot du
        \\={}& \frac{1}{n^2} \cdot (e+X_{n+2})(e+X_{n+1}) \cdot (\Pi_{\mathrm{sym}}^n \otimes I^{\otimes 2}) - \frac{d+n}{n^2} \cdot d[n] \cdot \int_{\ket{u}} \ketbra{u}^{\otimes n + 2} \cdot du\\={}&  \frac{1}{n^2} \cdot(e+X_{n+2})(e+X_{n+1}) \cdot (\Pi_{\mathrm{sym}}^n \otimes I^{\otimes 2}) - \mathrm{Lower}_{\mathrm{mom}},
    \end{align*}
    using the helper lemma (\Cref{lem:Helper_lemma}). The second term is (negative)
    \begin{align*}
        \frac{d+n}{n^2} \cdot d[n] \cdot \int_{\ket{u}}\ketbra{u}^{\otimes n+1} \otimes I \cdot du
        =\frac{1}{n^2} \cdot (e + X_{n+1}) \cdot (\Pi_{\mathrm{sym}}^{n} \otimes I^{\otimes 2}).
    \end{align*}
    The third term is identical to the second term, except with its $(n+1)$-st and $(n+2)$-nd registers exchanged. Hence, we can write it as (negative)
    \begin{align*}
        &(n+1, n+2) \cdot \Big(\frac{1}{n^2} \cdot (e + X_{n+1}) \cdot (\Pi_{\mathrm{sym}}^{n} \otimes I^{\otimes 2})\Big) \cdot (n+1, n+2)\\
        ={}&  \frac{1}{n^2} \cdot (e + (n+1, n+2)\cdot X_{n+1}\cdot(n+1, n+2)) \cdot (\Pi_{\mathrm{sym}}^{n} \otimes I^{\otimes 2}),
    \end{align*}
   where we have used here that $(n+1, n+2)$ commutes with $(\Pi_{\mathrm{sym}}^n \otimes I^{\otimes 2})$.
    Finally, the fourth term is
    \begin{equation*}
        \frac{1}{n^2} \cdot d[n] \cdot \int_{\ket{u}} \ketbra{u}^{\otimes n} \otimes I^{\otimes 2} \cdot du = \frac{1}{n^2} \cdot (\Pi_{\mathrm{sym}}^n \otimes I^{\otimes 2}).
    \end{equation*}
    Comparing all four terms (and excluding the lower-order term from the first term), we see that they are all an element of the symmetric group algebra times $(1/n^2) \cdot (\Pi_{\mathrm{sym}}^n \otimes I^{\otimes 2})$. 
    Combining these prefactors, we get
    \begin{align*}
        &(e+ X_{n+2})(e+X_{n+1}) - (e+X_{n+1}) - (e + (n+1, n+2) \cdot X_{n+1} \cdot (n+1, n+2)) + e\\
        ={}&X_{n+2} + X_{n+2}X_{n+1} - (n+1, n+2) \cdot X_{n+1} \cdot (n+1, n+2)\\
        ={}&(n+1, n+2) + X_{n+2}X_{n+1},
    \end{align*}
    where the final step uses the fact that $(n+1, n+2) \cdot X_{n+1} \cdot (n+1, n+2))$ contains all the swaps in $X_{n+2}$ except $(n+1, n+2)$.
    Multiplying this by $(1/n^2) \cdot (\Pi_{\mathrm{sym}}^n \otimes I^{\otimes 2})$ and subtracting off the lower-order term completes the proof.
\end{proof}

Now we use these operators to compute the first and second moments of the Grier--Pashayan--Schaeffer algorithm.
As before, we will compute these in the general case when their algorithm is performed on a mixed state $\psi_{\mathrm{sym}}$ inside the symmetric subspace.

\begin{lemma}[First moment]\label{lem:first-moment-gps}
    Let $\psi_{\mathrm{sym}}$ be a mixed state on $\lor^n \C^d$.
    If $\ket{\bv}$ is the output of Hayashi's measurement when performed on $\psi_{\mathrm{sym}}$, then
    \begin{equation*}
        \E[\widehat{\sigma}_{\bv}] = (\psi_{\mathrm{sym}})_{\reg{1}}.
    \end{equation*}
\end{lemma}
\noindent Recall our notation for partial trace: $(\psi_{\mathrm{sym}})_{\reg{1}} = \tr_{\reg{2}\dots\reg{n}}(\psi_{\mathrm{sym}})$.
\begin{proof}
By \Cref{prop:moment},
\begin{align*}
    \E[\widehat{\sigma}_{\bv}]
    &= \tr_{\reg{1}\dots\reg{n}}(M_{\mathrm{mom}}^{(1)} \cdot \psi_{\mathrm{sym}} \otimes I_d)\\
    &= \tr_{\reg{1}\dots\reg{n}}\Big(\Big(\frac{1}{n} \cdot X_{n+1} \cdot (\Pi_{\mathrm{sym}}^n \otimes I)\Big) \cdot \psi_{\mathrm{sym}} \otimes I_d\Big) \tag{by \Cref{lem:first-moment-operator}}\\
    &= \frac{1}{n} \cdot \tr_{\reg{1}\dots\reg{n}}(X_{n+1} \cdot \psi_{\mathrm{sym}} \otimes I_d)\\
    &= \frac{1}{n} \cdot \sum_{i=1}^n \tr_{\reg{1}\dots\reg{n}}((i, n+1) \cdot \psi_{\mathrm{sym}} \otimes I_d)\\
    &= \frac{1}{n} \cdot \sum_{i=1}^n \tr_{\reg{1}\dots\reg{n}}((1, n+1) \cdot \psi_{\mathrm{sym}} \otimes I_d) \tag{because $\psi_{\mathrm{sym}} \in \lor^n \C^d$} \\
    &= (\psi_{\mathrm{sym}})_{\reg{1}}. \tag{by \Cref{prop:no-tensor-network-diagrams}}
\end{align*}
This completes the proof.
\end{proof}

\begin{lemma}[Second moment]\label{lem:second-moment-gps}
    Let $\psi_{\mathrm{sym}}$ be a mixed state on $\lor^n \C^d$.
    If $\ket{\bv}$ is the output of Hayashi's measurement when performed on $\psi_{\mathrm{sym}}$, then
    \begin{equation*}
    \E[\widehat{\sigma}_{\bv} \otimes \widehat{\sigma}_{\bv}]=
         \frac{n-1}{n} \cdot (\psi_{\mathrm{sym}})_{\reg{1},\reg{2}} + \frac{1}{n} \cdot \big((\psi_{\mathrm{sym}})_{\reg{1}} \otimes I + I \otimes (\psi_{\mathrm{sym}})_{\reg{1}}\big) \cdot \swap + \frac{1}{n^2} \cdot \swap - \mathrm{Lower}_{\psi_{\mathrm{sym}}},
    \end{equation*}
    where $\mathrm{Lower}_{\psi_{\mathrm{sym}}} \in \symsep(d)$.
\end{lemma}
\begin{proof}
By \Cref{prop:moment},
\begin{align}\label{eq:2nd-mom}
    \E[\widehat{\sigma}_{\bv}\otimes \widehat{\sigma}_{\bv}]
    &= \tr_{\reg{1}\dots\reg{n}}(M_{\mathrm{mom}}^{(2)} \cdot \psi_{\mathrm{sym}} \otimes I_d^{\otimes 2}).
\end{align}
From \Cref{lem:second-moment-operator}, we have that
\begin{equation*}
    M_{\mathrm{mom}}^{(2)} = \frac{1}{n^2} \cdot (X_{n+2}X_{n+1} + (n+1, n+2)) \cdot (\Pi_{\mathrm{sym}}^n \otimes I^{\otimes 2}) - \mathrm{Lower}_{\mathrm{mom}},
    \end{equation*}
    where
    \begin{equation*}
        \mathrm{Lower}_{\mathrm{mom}} = \frac{d+n}{n^2} \cdot d[n] \cdot \int_{\ket{u}} \ketbra{u}^{\otimes n + 2} \cdot du.
    \end{equation*}
    Let us begin by calculating the main terms. Plugging these in to \Cref{eq:2nd-mom} above, we get
    \begin{align}
        &\tr_{\reg{1}\dots\reg{n}}\Big(\Big(\frac{1}{n^2} \cdot (X_{n+2}X_{n+1} + (n+1, n+2)) \cdot (\Pi_{\mathrm{sym}}^n \otimes I_d^{\otimes 2})\Big) \cdot \psi_{\mathrm{sym}} \otimes I_d^{\otimes 2}\Big)\nonumber\\
        ={}& \tr_{\reg{1}\dots\reg{n}}\Big(\Big(\frac{1}{n^2} \cdot (X_{n+2}X_{n+1} + (n+1, n+2))\Big) \cdot \psi_{\mathrm{sym}} \otimes I_d^{\otimes 2}\Big),\label{eq:main-term}
    \end{align}
    because $\psi_{\mathrm{sym}} \in \lor^n \C^d$.
    Now we calculate the product of the Jucys--Murphy elements:
    \begin{align*}
        X_{n+2} X_{n+1}
        &= \sum_{i=1}^{n+1} (i, n+2) \sum_{j = 1}^n (j, n+1)\\
        &= \sum_{i=1}^n (i, n+2)(i, n+1) + \sum_{j=1}^n (n+1, n+2) (j, n+1) + \sum_{i \neq j=1}^n (i, n+2)(j, n+1)\\
        &= \sum_{i=1}^n (i, n+1, n+2) + \sum_{j=1}^n (j, n+2, n+1) +\sum_{i \neq j=1}^n (i, n+2)(j, n+1)\\
        &= \sum_{i=1}^n (i, n+1)(n+1, n+2) + \sum_{j=1}^n (j, n+2)(n+1, n+2) +\sum_{i \neq j=1}^n (i, n+2)(j, n+1).
    \end{align*}
    We will plug these sums back into \Cref{eq:main-term} above one at a time. First,
    \begin{align*}
        &\frac{1}{n^2}\cdot\sum_{i=1}^n \tr_{\reg{1}\dots\reg{n}}\Big((i,n+1)(n+1, n+2) \cdot \psi_{\mathrm{sym}} \otimes I_d^{\otimes 2}\Big)\\
        ={}&\frac{1}{n^2}\cdot\sum_{i=1}^n \tr_{\reg{1}\dots\reg{n}}\Big((i,n+1) \cdot \psi_{\mathrm{sym}} \otimes I_d^{\otimes 2}\Big) \cdot \swap\\
        ={}&\frac{1}{n^2}\cdot\sum_{i=1}^n \tr_{\reg{1}\dots\reg{n}}\Big((1,n+1) \cdot \psi_{\mathrm{sym}} \otimes I_d^{\otimes 2}\Big) \cdot \swap \tag{because $\psi_{\mathrm{sym}} \in \lor^n \C^d$}\\
        ={}&\frac{1}{n}\cdot (\psi_{\mathrm{sym}})_{\reg{1}} \otimes I \cdot \swap. \tag{by \Cref{prop:no-tensor-network-diagrams}}
    \end{align*}
    A similar calculation shows that
    \begin{equation*}
        \frac{1}{n^2}\cdot\sum_{j=1}^n \tr_{\reg{1}\dots\reg{n}}\Big((j, n+2)(n+1, n+2) \cdot \psi_{\mathrm{sym}} \otimes I_d^{\otimes 2}\Big)
        = \frac{1}{n}\cdot I \otimes (\psi_{\mathrm{sym}})_{\reg{1}} \cdot \swap.
    \end{equation*}
    Next,
    \begin{align*}
    &\frac{1}{n^2}\cdot\sum_{i\neq j=1}^n \tr_{\reg{1}\dots\reg{n}}\Big((i,n+2)(j,n+1) \cdot \psi_{\mathrm{sym}} \otimes I_d^{\otimes 2}\Big)\\
    ={}&\frac{1}{n^2}\cdot\sum_{i\neq j=1}^n \tr_{\reg{1}\dots\reg{n}}\Big((1,n+2)(2,n+1) \cdot \psi_{\mathrm{sym}} \otimes I_d^{\otimes 2}\Big) \tag{because $\psi_{\mathrm{sym}} \in \lor^n \C^d$}\\
    ={}& \frac{n-1}{n}\cdot (\psi_{\mathrm{sym}})_{\reg{1},\reg{2}}. \tag{by \Cref{prop:no-tensor-network-diagrams}}
    \end{align*}
    The final term which occurs in \Cref{eq:main-term} is
    \begin{equation*}
        \frac{1}{n^2} \cdot \tr_{\reg{1}\dots\reg{n}}\Big((n+1, n+2) \cdot \psi_{\mathrm{sym}} \otimes I_d^{\otimes 2}\Big) = \frac{1}{n^2} \cdot \swap.
    \end{equation*}
    This accounts for all the main terms in the lemma statement.
    The only remaining term we have to account for is the lower order term, which comes from
    \begin{align*}
        \tr_{\reg{1}\dots\reg{n}}\Big(\mathrm{Lower}_{\mathrm{mom}} \cdot \psi_{\mathrm{sym}} \otimes I_d^{\otimes 2}\Big)
        &= \frac{d+n}{n^2} \cdot d[n] \cdot \int_{\ket{u}} \tr_{\reg{1}\dots\reg{n}}\Big(\ketbra{u}^{\otimes n+2} \cdot \psi_{\mathrm{sym}} \otimes I_d^{\otimes 2}\Big)\cdot du\\
        &= \frac{d+n}{n^2} \cdot d[n] \cdot \int_{\ket{u}} \tr\Big(\ketbra{u}^{\otimes n} \cdot \psi_{\mathrm{sym}}\Big) \cdot \ketbra{u}^{\otimes 2}\cdot du,
    \end{align*}
    which is a nonnegative linear combination of terms of the form $\ketbra{u}^{\otimes 2}$ and is therefore in~$\symsep(d)$.
    This completes the proof.
\end{proof}

Now we specialize these results to the case of $\psi_{\mathrm{sym}} = \ketbra{\psi}^{\otimes n}$ which occurs in pure state tomography.
The following is an immediate corollary of \Cref{lem:first-moment-gps,lem:second-moment-gps}.

\begin{corollary}[Moments of Grier--Pashayan--Schaeffer; \Cref{thm:gps-moments}, restated]\label{cor:gps-moments}
If $\ket{\bv}$ is the output of Hayashi's measurement on $n$ copies of $\ket{\psi} \in \C^d$, then $\widehat{\sigma}_{\bv}$ is an unbiased estimator for $\ketbra{\psi}$, i.e.\ $\E \widehat{\sigma}_{\bv} = \ketbra{\psi}$.
In addition,
\begin{equation*}
    \E[\widehat{\sigma}_{\bv}\otimes \widehat{\sigma}_{\bv}] = \frac{n-1}{n} \cdot \ketbra{\psi}^{\otimes 2} + \frac{1}{n} \cdot \big(\ketbra{\psi} \otimes I_d + I_d \otimes \ketbra{\psi}\big) \cdot \swap + \frac{1}{n^2} \cdot \swap-\mathrm{Lower}_{\psi},
    \end{equation*}
    where $\mathrm{Lower}_{\psi} \in \symsep(d)$.
\end{corollary}
\noindent
This matches the second moment of our algorithm from \Cref{thm:gps-moments-tight} in the case of pure states, as $\ell(\blambda) = 1$ always for pure states.

\section{A new unbiased estimator for mixed state tomography}\label{sec:unbiased-estimator}

In this section, we apply our reduction to the Grier--Pashayan--Schaeffer algorithm to obtain a new unbiased estimator for mixed state tomography. We will use the moments of $\calA_{\mathrm{GPS}}$ from~\Cref{sec:gps} (\Cref{thm:gps-moments}) to derive the moments of the resulting mixed state tomography algorithm.

We begin by applying the generic reduction to get $\mixed(\calA_{\mathrm{GPS}})$ and prove~\Cref{thm:gps-moments-loose}. In~\Cref{sec:gps-mixed-llambda}, we obtain the improved result of~\Cref{thm:gps-moments-tight} that is required for some of our applications. We do this by analyzing the algorithm $\mixed^+(\calA_{\mathrm{GPS}})$ that results from quasi-purification.

\subsection{Warmup: direct reduction to pure state tomography}\label{sec:gps-mixed}

First, let us apply our reduction to the Grier--Pashayan--Schaeffer algorithm.
This results in the following mixed state tomography algorithm.

{
\floatstyle{boxed} 
\restylefloat{figure}
\begin{figure}[H]
Given $n$ copies of $\rho$:
\begin{enumerate}
    \item\label{item:purify} First apply $\purifychan^{d, r}$ to produce $n$ copies of a random purification $\ket{\brho}_{\reg{AB}} \in \C^d \otimes \C^r$.
    \item Apply the Grier--Pashayan--Schaeffer algorithm to learn an estimate $\widehat{\sigma}_{\bv}$ of $\ketbra{\brho}$.
    \item Set $(\widehat{\rho}_{\bv})_{\reg{A}} = \tr_\reg{B}((\widehat{\sigma}_{\bv})_{\reg{AB}})$ of $\rho$. Output $\widehat{\rho}_{\bv}$.
\end{enumerate}
\caption{The mixed state tomography algorithm $\mixed(\gps)$; \Cref{fig:gps-reduction-basic}, restated.}
\label{fig:gps-reduction-basic-restate}
\end{figure}
}

\noindent
The main result of this section is a formula for the first and second moments of the output $\widehat{\rho}_{\bv}$.

\begin{theorem}[Moments of $\mixed(\gps)$; \Cref{thm:gps-moments-loose}, restated] \label{thm:gps-moments-loose-2}
    Let $\widehat{\rho}_{\bv}$ be the output of $\mixed(\calA_{\mathrm{GPS}})$ in \Cref{fig:gps-reduction-basic-restate} when run on $n$ copies of a rank-$r$ state $\rho \in \C^{d \times d}$.
    Then $\widehat{\rho}_{\bv}$ is an unbiased estimator for $\rho$ with second moment
    \begin{equation*}
    \E[\widehat{\rho}_{\bv}\otimes \widehat{\rho}_{\bv}] = \frac{n-1}{n} \cdot \rho^{\otimes 2} + \frac{1}{n} \cdot \big(\rho \otimes I_d + I_d \otimes \rho\big) \cdot \swap + \frac{r}{n^2} \cdot \swap-\mathrm{Lower}_{\rho},
    \end{equation*}
    where $\mathrm{Lower}_{\rho} \in \symsep(d)$.
\end{theorem}

\begin{proof}
We will compute the first and second moments, conditioned on the outcome of the random purification $\ket{\brho}$ in \hyperref[item:purify]{Step~\ref*{item:purify}} of the algorithm.
It will turn out that the moments don't depend on the purification, so then we will be done.
First, some notation: we will write $\reg{A}$ for the $\C^d$ register of $\ket{\brho}$ and $\reg{B}$ for its $\C^r$ register.
We can calculate the first moment using \Cref{cor:gps-moments}:
\begin{equation*}
    \E_{\bv}[\widehat{\rho}_{\bv}]
    = \E_{\bv}[\tr_\reg{B}((\widehat{\sigma}_{\bv})_{\reg{AB}})]
    = \tr_\reg{B}(\ketbra{\brho}_{\reg{AB}})
    = \rho.
\end{equation*}
As this holds for any $\ket{\brho}$,
it also holds on average for a random $\ket{\brho}$.
Thus, $\widehat{\rho}_{\bv}$ is an unbiased estimator for $\rho$.
Now for the second moment: here, we are looking at two copies of each register, $\reg{A}_1, \reg{B}_1$, and $\reg{A}_2, \reg{B}_2$.
Conditioned on the purification $\ket{\brho}$, we have that
\begin{equation*}
    \E_{\bv}[\widehat{\rho}_{\bv} \otimes \widehat{\rho}_{\bv}]
    = \E_{\bv}[\tr_{\reg{B}_1}((\widehat{\sigma}_{\bv})_{\reg{A}_1\reg{B}_1}) \otimes \tr_{\reg{B}_2}((\widehat{\sigma}_{\bv})_{\reg{A}_2\reg{B}_2})]
    = \E_{\bv}[\tr_{\reg{B}_1 \reg{B}_2}(\widehat{\sigma}_{\bv} \otimes \widehat{\sigma}_{\bv})]
    = \tr_{\reg{B}_1 \reg{B}_2}\big(\E_{\bv}[\widehat{\sigma}_{\bv} \otimes \widehat{\sigma}_{\bv}]\big).
\end{equation*}
To calculate the expectation in the partial trace, we can appeal to \Cref{cor:gps-moments}, which states that
\begin{equation*}
    \E_{\bv}[\widehat{\sigma}_{\bv}\otimes \widehat{\sigma}_{\bv}] = \frac{n-1}{n} \cdot \ketbra{\brho}^{\otimes 2} + \frac{1}{n} \cdot \big(\ketbra{\brho} \otimes I_{\reg{A}_2 \reg{B}_2} + I_{\reg{A}_1 \reg{B}_1} \otimes \ketbra{\brho}\big) \cdot \swap_{\reg{A} \reg{B}} + \frac{1}{n^2} \cdot \swap_{\reg{A} \reg{B}}-\mathrm{Lower}_{\ket{\brho}},
\end{equation*}
where $\mathrm{Lower}_{\ket{\brho}} \in \symsep(D)$ for $D = d\cdot r$.
Here $\swap_{\reg{A} \reg{B}}$ is the operator which swaps $\reg{A}_1 \reg{B}_1$ with $\reg{A}_2\reg{B}_2$; note that by \Cref{prop:permute-factor} it factorizes as
\begin{equation*}\swap_{\reg{A} \reg{B}} = \swap_{\reg{A}} \otimes \swap_{\reg{B}}.
\end{equation*}
Now we compute the partial trace of this expression term-by-term.
First,
\begin{equation*}
\tr_{\reg{B}_1\reg{B}_2}(\ketbra{\brho}^{\otimes 2})
= \tr_{\reg{B}_1}(\ketbra{\brho}_{\reg{A}_1\reg{B}_1}) \otimes \tr_{\reg{B}_2}(\ketbra{\brho}_{\reg{A}_2\reg{B}_2})
= \rho_{\reg{A}_1} \otimes \rho_{\reg{A}_2}.
\end{equation*}
For the second term, we have that
\begin{equation*}
    \tr_{\reg{B}_1\reg{B}_2}(\ketbra{\brho}_{\reg{A}_1\reg{B}_1} \otimes I_{\reg{A}_2\reg{B}_2} \cdot \swap_{\reg{A}\reg{B}})
    = (\rho_{\reg{A}_1} \otimes I_{\reg{A}_2}) \cdot \swap_{\reg{A}},
\end{equation*}
by \Cref{prop:weird-prop}.
For the third term, a similar calculation shows that
\begin{equation*}
    \tr_{\reg{B}_1\reg{B}_2}(I_{\reg{A}_1\reg{B}_1} \otimes \ketbra{\brho}_{\reg{A}_2\reg{B}_2}  \cdot \swap_{\reg{A}\reg{B}})
    = (I_{\reg{A}_1} \otimes \rho_{\reg{A}_2}) \cdot \swap_{\reg{A}}.
\end{equation*}
For the fourth term, we have that
\begin{equation*}
    \tr_{\reg{B}_1 \reg{B}_2}(\swap_{\reg{A} \reg{B}})
    = \tr_{\reg{B}_1 \reg{B}_2}(\swap_{\reg{A}} \otimes \swap_{\reg{B}})
    = \swap_{\reg{A}} \cdot \tr(\swap_{\reg{B}})
    = \swap_{\reg{A}} \cdot r.
\end{equation*}
Finally, for the lower-order term,
since $\mathrm{Lower}_{\ket{\brho}} \in \symsep(D)$,
 $\tr_{\reg{B}_1 \reg{B}_2}(\mathrm{Lower}_{\ket{\brho}}) \in \symsep(d)$
 due to \Cref{prop:symset-partial-trace}.

Putting everything together, we have
\begin{equation*}
    \E[\widehat{\rho}_{\bv}\otimes \widehat{\rho}_{\bv}] = \frac{n-1}{n} \cdot \rho^{\otimes 2} + \frac{1}{n} \cdot \big(\rho \otimes I_d + I_d \otimes \rho\big) \cdot \swap + \frac{r}{n^2} \cdot \swap-\tr_{\reg{B}_1 \reg{B}_2}(\mathrm{Lower}_{\ket{\brho}}).
    \end{equation*}
As this holds for any $\ket{\brho}$,
it also holds on average for a random $\ket{\brho}$.
This completes the proof.
\end{proof}

The second moment of $\mixed(\calA_{\mathrm{GPS}})$ nearly matches the bound in~\Cref{thm:gps-moments-tight}, except the $\E[\ell(\blambda)]$ in that expression is replaced by a larger factor of~$r$ here.
As mentioned before, the factor of $\E[\ell(\blambda)]$ is crucial for achieving our sample complexities for the applications of shadow tomography and tomography with limited entanglement.
We will be able to achieve the improved second moment bound from~\Cref{thm:gps-moments-tight} using quasi-purification,
which we explain in \Cref{sec:gps-mixed-llambda} below.

\subsection{Improving the reduction by only quasi-purifying}\label{sec:gps-mixed-llambda}

Now we give our unbiased estimator $\mixed^+(\calA_{\mathrm{GPS}})$ which improves upon the $\mixed(\calA_{\mathrm{GPS}})$ algorithm from \Cref{fig:gps-reduction-basic-restate} above.
To motivate our construction, let us take a new perspective on $\mixed(\calA_{\mathrm{GPS}})$ that does not treat the random purification channel as a black-box.
Previously, we viewed the random purification channel as a way to go from $n$ copies of a mixed state $\rho^{\otimes n}$ to a random purification $\ketbra{\brho}^{\otimes n}$, which is then fed into the Grier--Pashayan--Schaeffer algorithm.
However, under the perspective of \Cref{sec:the-channel}, the random purification channel performs weak Schur sampling on $\rho^{\otimes n}$, resulting in a Young diagram $\blambda \vdash n$, and the state collapses to $\rho|_{\blambda}$.
The purification channel $\purifychan^{d,r}$ is then applied, resulting in the state 
\begin{equation}\label{eq:surprise!-in-pi-sym}
    \purifychan^{d, r}(\rho|_{\blambda})
    = \ketbra{\blambda\blambda}{\blambda\blambda}_{\reg{Y} \reg{Y'}} \otimes \ketbra{\epr_{\blambda}}_{\reg{P} \reg{P'}} \otimes \Big(\frac{\nu^d_{\blambda}(\rho)}{s^d_{\blambda}(\rho)}\Big)_{\reg{Q}}  \otimes \Big(\frac{I_{\dim(V_{\blambda}^r)}}{\dim(V_{\blambda}^r)}\Big)_{\reg{Q'}}.
\end{equation}
Then the Grier--Pashayan--Schaeffer algorithm is applied to \emph{this} state. 
At first blush, this perspective looks a bit strange, because their algorithm typically takes as input a state of the form $\ketbra{\psi}^{\otimes n}$,
but the state in \Cref{eq:surprise!-in-pi-sym} is clearly not of this form.
But this state \emph{is} at least in the symmetric subspace $\lor^n (\C^d \otimes \C^r)$ due to \Cref{lem:double-schur-pi-sym},
and so applying the Grier--Pashayan--Schaeffer algorithm to this state is at the very least a well-defined operation.
In addition, we have seen in \Cref{sec:gps} that it is possible to analyze their algorithm for general states drawn from the symmetric subspace, rather than just states of the form $\ketbra{\psi}^{\otimes n}$.

From this perspective, there is no particular reason why, for each $\blambda$, we must purify our state using a purification register $\C^r$ of the same dimension $r$.
Instead, for each $\blambda$, our algorithm will purify to the smallest dimension possible for the purification register, which is $\ell(\blambda)$.
(That $\ell(\blambda)$ is the smallest possible purification dimension comes from the fact that the purifying irrep register $V_{\blambda}^r$ in \Cref{eq:surprise!-in-pi-sym} is only well-defined when $r \geq \ell(\blambda)$.)
Intuitively, reducing the size of the purification register
reduces the size of the Hilbert space that the Grier--Pashayan--Schaeffer algorithm needs to search over,
reducing the number of copies needed.

{
\floatstyle{boxed} 
\restylefloat{figure}
\begin{figure}[H]
Given $n$ copies of $\rho$:
\begin{enumerate}
    \item Apply the Schur transform $\schur^d$ to $\rho^{\otimes n}$.
    \item Perform weak Schur sampling. Letting $\blambda$ be the outcome, the state collapses to $\rho|_{\blambda}$.
    Set $\bell \coloneqq \ell(\blambda)$.
    \item Apply the purification channel to compute $\purifychan^{d, \bell}(\rho|_{\blambda})$. 
    \item Apply an inverse Schur transform to both registers, i.e.\ the operation $(\schur^d \otimes \schur^{\bell})^{\dagger}$.
    Write
    \begin{equation*}
        \tau_{\blambda}(\rho) \coloneqq (\schur^d \otimes \schur^{\bell})^{\dagger} \cdot \purifychan^{d, \bell}(\rho|_{\blambda}) \cdot (\schur^d \otimes \schur^{\bell})
    \end{equation*}
    for the resulting state.
    \item
    The state $\tau_{\blambda}(\rho)$ is an element of $\lor^n (\C^d \otimes \C^{\bell}) \cong \lor^n \C^{\bD}$, for $\bD = d\cdot \bell$.
    Apply the Grier--Pashayan--Schaeffer algorithm to learn an estimate $\widehat{\sigma}_{\bv}^{\blambda}$.
    \item Set $\widehat{\rho}^{\blambda}_{\bv} = \tr_2(\widehat{\sigma}_{\bv}^{\blambda})$. Output $\widehat{\rho}^{\blambda}_{\bv}$.
\end{enumerate}
\caption{Our improved mixed state tomography algorithm $\mixed^+(\calA_{\mathrm{GPS}})$.}
\label{fig:gps-reduction-improved}
\end{figure}
}

Before calculating any moments, let us first consider the intermediate state $\tau_{\blambda}(\rho)$.
It has $n$ registers of type $\C^d$,
    which we will refer to as the $\reg{A}_1, \ldots, \reg{A}_n$ registers,
    and $n$ registers of type $\C^{\bell}$,
    which we will refer to as the $\reg{B}_1, \ldots, \reg{B}_n$ registers.    
    We will use the following helper lemma.
\begin{lemma}[Partial trace helper lemma]\label{lem:partial-trace-helper}
For any $1 \leq k \leq n$,
    \begin{equation*}
        \E_{\blambda}\Big[\tr_{\substack{\reg{A}_{k+1} \ldots \reg{A_{n}} \\\reg{B}_1 \ldots \reg{B_{n}}}}(\tau_{\blambda}(\rho))\Big] = \rho^{\otimes k}.
    \end{equation*}
\end{lemma}
\begin{proof}
    Let us first calculate the partial trace for a fixed $\blambda$.
    \begin{align*}
        \tr_{\substack{\reg{A}_{k+1} \ldots \reg{A_{n}} \\\reg{B}_1 \ldots \reg{B_{n}}}}(\tau_{\blambda}(\rho))
        &= \tr_{\substack{\reg{A}_{k+1} \ldots \reg{A_{n}} \\\reg{B}_1 \ldots \reg{B_{n}}}}\Big((\schur^d \otimes \schur^{\bell})^{\dagger} \cdot \purifychan^{d, \bell}(\rho|_{\blambda}) \cdot (\schur^d \otimes \schur^{\bell})\Big)\\
        &= \tr_{\substack{\reg{A}_{k+1} \ldots \reg{A_{n}} \\\reg{B}_1 \ldots \reg{B_{n}}}}\Big((\schur^d \otimes I)^{\dagger} \cdot \purifychan^{d, \bell}(\rho|_{\blambda}) \cdot (\schur^d \otimes I)\Big) \tag{by the cyclic property of the partial trace}\\
        &= \tr_{\reg{A}_{k+1} \ldots \reg{A_{n}}}\Big((\schur^d)^{\dagger} \cdot \tr_{\reg{B}}(\purifychan^{d, \bell}(\rho|_{\blambda})) \cdot \schur^d \Big)\\
        &= \tr_{\reg{A}_{k+1} \ldots \reg{A_{n}}}\Big((\schur^d)^{\dagger} \cdot \rho|_{\blambda} \cdot \schur^d \Big). \tag{by \Cref{lem:partial-trace-restriction}}
    \end{align*}
    Averaging over $\blambda$ gives us
    \begin{align*}
        \E_{\blambda}\Big[\tr_{\substack{\reg{A}_{k+1} \ldots \reg{A_{n}} \\\reg{B}_1, \ldots, \reg{B_{n}}}}(\tau_{\blambda}(\rho))\Big]
        &= \E_{\blambda}\Big[\tr_{\reg{A}_{k+1} \ldots \reg{A_{n}}}\Big((\schur^d)^{\dagger} \cdot \rho|_{\blambda} \cdot \schur^d \Big)\Big]\\
        &= \tr_{\reg{A}_{k+1} \ldots \reg{A_{n}}}\Big((\schur^d)^{\dagger} \cdot \E_{\blambda}[\rho|_{\blambda}] \cdot \schur^d \Big)\\
        &= \tr_{\reg{A}_{k+1} \ldots \reg{A_{n}}}\big(\rho^{\otimes n} \big) \tag{by \Cref{lem:blambda-average}}\\
        &= \rho^{\otimes k}.
    \end{align*}
    This completes the proof.
\end{proof}

Now we calculate the first and second moments of our estimator.
We begin by showing that it does indeed give an unbiased estimator of $\rho$.

\begin{theorem}[First moment of $\mixed^+(\gps)$]
    Let $\widehat{\rho}_{\bv}^{\blambda}$ be the output of our mixed state tomography algorithm.~Then
    \begin{equation*}
        \E[\widehat{\rho}_{\bv}^{\blambda}] = \rho.
    \end{equation*}
\end{theorem}
\begin{proof}
    First, let us calculate the expectation of $\widehat{\rho}_{\bv}^{\blambda}$ conditioned on the result $\blambda$ of weak Schur sampling (which, in turn, conditions on the values $\bell$ and $\bD$).
    This estimate is the result of applying the Grier--Pashayan--Schaeffer algorithm to the state $\tau_{\blambda}(\rho)$.
    Then by \Cref{lem:first-moment-gps},
    \begin{equation*}
        \E[\widehat{\sigma}_{\bv}^{\blambda} \mid \blambda] = \tr_{\substack{\reg{A}_2 \ldots \reg{A_{n}} \\\reg{B}_2 \ldots \reg{B_{n}}}}\big(\tau_{\blambda}(\rho)\big).
    \end{equation*}
    Then because $\widehat{\rho}_{\bv}^{\blambda}$ is the result of tracing out $\widehat{\sigma}_{\bv}^{\blambda}$'s purifying register, we have
    \begin{align*}
        \E[\widehat{\rho}_{\bv}^{\blambda} \mid \blambda] &= \tr_{\substack{\reg{A}_2 \ldots \reg{A_{n}} \\\reg{B}_1 \ldots \reg{B_{n}}}}\big(\tau_{\blambda}(\rho)\big).
    \end{align*}
    Note the addition of the $\reg{B}_1$ register to the partial trace.
    Averaging over $\blambda$ gives us
    \begin{align*}
        \E[\widehat{\rho}_{\bv}^{\blambda}]
        = \E_{\blambda}\Big[\tr_{\substack{\reg{A}_2 \ldots \reg{A_{n}} \\\reg{B}_1 \ldots \reg{B_{n}}}}\big(\tau_{\blambda}(\rho)\big)\Big]
        = \rho,
    \end{align*}
    by \Cref{lem:partial-trace-helper}.
    This completes the proof.
\end{proof}

Next, we compute its second moment.

\begin{theorem}[Second moment of $\mixed^+(\gps)$]
    Let $\widehat{\rho}_{\bv}^{\blambda}$ be the output of our mixed state tomography algorithm. Then
    \begin{equation*}
        \E[\widehat{\rho}_{\bv}^{\blambda} \otimes \widehat{\rho}_{\bv}^{\blambda}]
        = \frac{n-1}{n} \cdot \rho^{\otimes 2} + \frac{1}{n} \cdot \big(\rho \otimes I_d + I_d \otimes \rho\big) \cdot \swap + \frac{\E[\ell(\blambda)]}{n^2} \cdot \swap-\mathrm{Lower}_{\rho},
    \end{equation*}
    where $\mathrm{Lower}_{\rho} \in \symsep(d)$.
\end{theorem}
\begin{proof}
    Following the proof of the first moment case, we will calculate the expectation of $\widehat{\rho}_{\bv}^{\blambda}\otimes \widehat{\rho}_{\bv}^{\blambda}$ conditioned on the result $\blambda$ of weak Schur sampling.
    By \Cref{lem:second-moment-gps},
    \begin{align*}
        \E[\widehat{\sigma}_{\bv}^{\blambda}\otimes \widehat{\sigma}_{\bv}^{\blambda} \mid \blambda]
        =~&
         \frac{n-1}{n} \cdot \tr_{\substack{\reg{A}_3 \ldots \reg{A_{n}} \\\reg{B}_3 \ldots \reg{B_{n}}}}(\tau_{\blambda}(\rho))\tag{term 1}\\
         &+ \frac{1}{n} \cdot \big(\tr_{\substack{\reg{A}_2 \ldots \reg{A_{n}} \\\reg{B}_2 \ldots \reg{B_{n}}}}(\tau_{\blambda}(\rho)) \otimes I_{\bD}\big)  \cdot \swap_{\reg{A} \reg{B}}\tag{term 2}\\
         &+ \frac{1}{n} \cdot \big( I_{\bD} \otimes (\tr_{\substack{\reg{A}_2 \ldots \reg{A_{n}} \\\reg{B}_2 \ldots \reg{B_{n}}}}(\tau_{\blambda}(\rho))\big) \cdot \swap_{\reg{A} \reg{B}}\tag{term 3}\\
         &+ \frac{1}{n^2} \cdot \swap_{\reg{A} \reg{B}}\tag{term 4}\\
         &- \mathrm{Lower}_{\tau_{\blambda}(\rho)},\tag{term 5}
    \end{align*}
    where $\mathrm{Lower}_{\tau_{\blambda}(\rho)} \in \symsep(\bD)$.
    Next, to compute $ \E[\widehat{\rho}_{\bv}^{\blambda}\otimes \widehat{\rho}_{\bv}^{\blambda} \mid \blambda]$, we trace out both purifying registers.
    For term 1, this results in 
    \begin{equation*}
        \frac{n-1}{n} \cdot \tr_{\substack{\reg{A}_3 \ldots \reg{A_{n}} \\\reg{B}_1 \ldots \reg{B_{n}}}}(\tau_{\blambda}(\rho)).
    \end{equation*}
    For term 2, this results in
    \begin{equation*}
        \tr_{\reg{B}}\Big(\frac{1}{n} \cdot \big(\tr_{\substack{\reg{A}_2 \ldots \reg{A_{n}} \\\reg{B}_2 \ldots \reg{B_{n}}}}(\tau_{\blambda}(\rho)) \otimes I_{\bD}\big)  \cdot \swap_{\reg{A} \reg{B}}\Big)
        = \frac{1}{n} \cdot \big(\tr_{\substack{\reg{A}_2 \ldots \reg{A_{n}} \\\reg{B}_1 \ldots \reg{B_{n}}}}(\tau_{\blambda}(\rho)) \otimes I_{d}\big)  \cdot \swap_{\reg{A}},
    \end{equation*}
    by \Cref{prop:weird-prop}. Term 3 results in a similar expression.
    For term 4, we have
    \begin{equation*}
        \tr_{\reg{B}}\Big(\frac{1}{n^2}\cdot \swap_{\reg{A} \reg{B}}\Big) = 
        \tr\Big(\frac{1}{n^2}\cdot \swap_{\reg{B}}\Big) \cdot \swap_{\reg{A}}
        = \frac{\ell(\blambda)}{n^2} \cdot \swap_{\reg{A}},
    \end{equation*}
    by \Cref{prop:permute-factor}. Finally, for term 5, $\mathrm{Lower}_{\blambda}' = \tr_{\reg{B}}(\mathrm{Lower}_{\tau_{\blambda}(\rho)})$ is in $\symsep(d)$, by \Cref{prop:symset-partial-trace}.
    In total, we have
    \begin{align*}
        \E[\widehat{\rho}_{\bv}^{\blambda}\otimes \widehat{\rho}_{\bv}^{\blambda} \mid \blambda]
        =~&
         \frac{n-1}{n} \cdot \tr_{\substack{\reg{A}_3 \ldots \reg{A_{n}} \\\reg{B}_1 \ldots \reg{B_{n}}}}(\tau_{\blambda}(\rho))\tag{term 1}\\
         &+ \frac{1}{n} \cdot \big(\tr_{\substack{\reg{A}_2 \ldots \reg{A_{n}} \\\reg{B}_1 \ldots \reg{B_{n}}}}(\tau_{\blambda}(\rho)) \otimes I_{d}\big)  \cdot \swap_{\reg{A}}\tag{term 2}\\
         &+ \frac{1}{n} \cdot \big( I_{d} \otimes (\tr_{\substack{\reg{A}_2 \ldots \reg{A_{n}} \\\reg{B}_2 \ldots \reg{B_{n}}}}(\tau_{\blambda}(\rho))\big) \cdot \swap_{\reg{A}}\tag{term 3}\\
         &+ \frac{\ell(\blambda)}{n^2} \cdot \swap_{\reg{A} }\tag{term 4}\\
         &- \mathrm{Lower}_{\blambda}'.\tag{term 5}
    \end{align*}
    Now we take the expectation of this with respect to $\blambda$. By \Cref{lem:partial-trace-helper}, this is 
    \begin{equation*}
        \E[\widehat{\rho}_{\bv}^{\blambda} \otimes \widehat{\rho}_{\bv}^{\blambda}]
        = \frac{n-1}{n} \cdot \rho^{\otimes 2} + \frac{1}{n} \cdot \big(\rho \otimes I_d + I_d \otimes \rho\big) \cdot \swap + \frac{\E[\ell(\blambda)]}{n^2} \cdot \swap-\E_{\blambda}[\mathrm{Lower}_{\blambda}'].
    \end{equation*}
    Note that $\E_{\blambda}[\mathrm{Lower}_{\blambda}'] \in \symsep(d)$,
    as it is a convex combination of matrices in $\symsep(d)$.
    This completes the proof.
\end{proof}

\section*{Acknowledgments}

A.P.\ is supported by DARPA under Agreement No.\ HR00112020023. 
J.S.\ and J.W.\ are supported by the NSF CAREER award CCF-233971.
E.T.\ is supported by the Miller Institute for Basic Research in Science, University of California, Berkeley.

\bibliographystyle{alpha}
\bibliography{wright}

\appendix

\section{Viewing our algorithms as pretty good measurements} \label{sec:pgm}

Any tomography algorithm---whether it performs multiple measurements, introduces ancillary systems, or applies intermediate channels---is equivalent to another tomography algorithm that performs a (possibly complicated) POVM directly on the input state $\rho^{\otimes n}$.
This raises the question: what POVMs are our algorithms performing? In this section, we show that $\mixed(\gps)$ implements a pretty good measurement (PGM) over a natural distribution known as the \emph{Hilbert--Schmidt measure} (defined below). Furthermore, we show that $\mixed^+(\gps)$ is only slightly more elaborate: it performs a PGM over one of several such distributions, conditioned on the outcome of weak Schur sampling. 

Ours is not the first tomography algorithm which can be viewed as performing a PGM on the input.
Indeed, Hayashi's pure state tomography algorithm~\cite{Hay98} can be viewed as performing a PGM over a distribution induced by a Haar random unit vector. 
In addition, one of the two mixed state tomography algorithms in~\cite{HHJ+16} performs a PGM over a particular distribution on mixed states,
and the other, like $\mixed^+(\gps)$, performs a PGM conditioned on the outcome of weak Schur sampling.
To our knowledge, however, the particular choice of PGM we study, based on the Hilbert--Schmidt measure, has not appeared in the literature prior to our work.

\subsection{PGM preliminaries} A PGM is defined in terms of a set of states $\{\rho_i\}$ and a prior probability distribution on these states $\{\alpha_i\}$. The corresponding PGM has measurement operators $\{M_i\}$ with 
\begin{equation*}
    M_i \coloneq N^{-1/2} \cdot \alpha_i \rho_i \cdot N^{-1/2},
\end{equation*}
where $N \coloneq \sum_{i} \alpha_i \rho_i$. Here, $N^{-1/2}$ is the \emph{Moore-Penrose pseudoinverse} of $N^{1/2}$, i.e.\ the inverse restricted to the support of $N^{1/2}$. 

The PGM was originally introduced for its application to the problem of \emph{quantum state discrimination} \cite{Bel75,Hol79,HW94}. In this problem, we are given one copy of a state $\brho$ drawn randomly from the set $\{\rho_i\}$, with $\rho_i$ sampled with probability $\alpha_i$, and asked to identify the state. The pretty good measurement is \emph{pretty good} at this task. In particular, let $P_\mathrm{PGM}$ be the success probability of the algorithm that first measures $\brho$ using the PGM, obtaining outcome $\bi$, and then outputs $\rho_{\bi}$. Then $P_{\mathrm{PGM}} \geq P_{\mathrm{OPT}}^2$, where $P_{\mathrm{OPT}}$ is the optimal success probability across all possible measurement schemes \cite{BK02}.

\subsection{The PGM over the Hilbert--Schmidt measure} Our PGM will have measurement outcomes indexed by mixed states $\sigma \in \C^{d \times d}$. The state indexed by $\sigma$ is the $n$-fold product $\sigma^{\otimes n}$. We will need to define a measure on mixed states to state the probability associated to this state.  

\begin{definition}[Rank-$r$ Hilbert--Schmidt measure]\label{def:HS-measure}
   The \emph{rank-$r$ Hilbert--Schmidt measure}, denoted $\mu^{d,r}_{\mathrm{HS}}$, is the measure on mixed states $\sigma \in \C^{d \times d}$ induced by the Haar measure on pure states $\ket{u}_{\reg{A_1B_1}} \in \C^d \otimes \C^r$, obtained by setting $\sigma_{\reg{A_1}} = \tr_{\reg{B_1}} ( \ketbra{u}_{\reg{A_1B_1}} )$.\footnote{This measure is also known as the \emph{induced measure} in the literature, and, in the $r = d$ special case, coincides with the \emph{Hilbert--Schmidt measure} \cite{ZS01}}
\end{definition}

Equivalently, the rank-$r$ Hilbert--Schmidt measure is the \emph{pushforward} of the Haar measure on pure states under the map $\tr_{\reg{B_1}}$. Since $\tr_{\reg{B_1}}$ is measurable, for every measurable $g$ on the space of $d\times d$ density matrices we have the change-of-variables formula
\begin{equation}\label{eq:pushfoward_change_of_variables}
\int_{\tr_{\reg{2}}(X)} g(\sigma) \cdot d\mu_{\mathrm{HS}}(\sigma) = \int_{X} g( \tr_{\reg{2}}(\ketbra{u} ) \cdot du,
\end{equation}
whenever one side is well-defined \cite[Theorem 3.6.1]{Bog07}.

For our PGM, the prior probability of $\sigma^{\otimes n}$ will be $d\mu^{d,r}_{\mathrm{HS}}(\sigma)$. The measurement operators of the PGM are then given by
\begin{equation}\label{eq:M_HS_start}
    M^{d,r}_{\mathrm{HS}} (\sigma) \coloneq N^{-1/2} \cdot \sigma^{\otimes n} \cdot d\mu^{d,r}_{\mathrm{HS}}(\sigma) \cdot N^{-1/2}.
\end{equation}
We now compute $S$. Letting $\ket{u}_{\reg{A}\reg{B}} \in \C^d \otimes \C^r$ and $(\sigma_u)_{\reg{A}} \coloneq \tr_{\reg{B}}(\ketbra{u}_{\reg{A}\reg{B}})$, we have
\begin{equation}\label{eq:S_intermediate}
    N = \int_{\sigma} \sigma^{\otimes n} \cdot d\mu^{d,r}_{\mathrm{HS}}(\sigma) = \int_{\ket{u}} \sigma_u^{\otimes n} \cdot du = \tr_{\reg{B}_1\dots\reg{B}_n} \Big( \int_{\ket{u}} \ketbra{u}^{\otimes n} \cdot du \Big) = \tr_{\reg{B}_1\dots\reg{B}_n} \Big( \frac{1}{D[n]} \cdot \Pi_{\mathrm{sym}}^{n,D} \Big).
\end{equation}
Here, $D \coloneq d \cdot r$ and $D[n]=\dim(\vee^n \C^{D})$. However, by \Cref{lem:double-schur-pi-sym}, we know
\begin{equation*}
        (\schur^d \otimes \schur^r) \cdot \Pi_{\mathrm{sym}}^{n, D} \cdot (\schur^d \otimes \schur^r)^\dagger
        = \sum_{\lambda \vdash n, \ell(\lambda) \leq r} \ketbra{\lambda\lambda}{\lambda\lambda}_{\reg{Y}\reg{Y'}} \otimes \ketbra{\epr_{\lambda}}{\epr_{\lambda}}_{\reg{P}\reg{P'}} \otimes I_{\reg{Q}\reg{Q'}}.
    \end{equation*}
Tracing out the $\reg{B}$ registers (which now correspond to $\reg{Y}'\reg{P}'\reg{Q}'$) and plugging the result back into \Cref{eq:S_intermediate} gives us an expression for $N$ in the Schur basis:
\begin{equation}\label{eq:S_in_Schur}
    \schur^d \cdot N \cdot (\schur^{d})^{\dagger} = \frac{1}{D[n]} \cdot \sum_{\lambda \vdash n, \ell(\lambda) \leq r} \ketbra{\lambda}_{\reg{Y}} \otimes \Big( \frac{I_{\dim(\lambda)}}{\dim(\lambda)}\Big)_{\reg{P}} \otimes \Big( \dim(V^r_\lambda) \cdot I_{\dim(V^d_\lambda)} \Big)_{\reg{Q}}.
\end{equation}
Moreover, $\sigma^{\otimes n}$ can also be expressed in the Schur basis as
\begin{equation}\label{eq:sigma_in_Schur}
    \schur^d \cdot \sigma^{\otimes n} \cdot (\schur^d)^{\dagger} = \sum_{\lambda \vdash n, \ell(\lambda) \leq r} \ketbra{\lambda}_{\reg{Y}} \otimes (I_{\dim(\lambda)})_{\reg{P}} \otimes (\nu^d_\lambda(\sigma))_{\reg{Q}}.
\end{equation}
Combining \Cref{eq:M_HS_start,eq:S_in_Schur,eq:sigma_in_Schur}, we get:
\begin{align}\label{eq:M_HS_in_Schur}
\schur^d \cdot M_{\mathrm{HS}}^{d,r}(\sigma) \cdot (\schur^d)^\dagger & = \sum_{\lambda \vdash n, \ell(\lambda) \leq r} \frac{D[n] \cdot \dim(\lambda)}{\dim(V^r_\lambda)} \cdot \ketbra{\lambda}_{\reg{Y}} \otimes (I_{\dim(\lambda)})_{\reg{P}} \otimes \nu^d_\lambda(\sigma)_{\reg{Q}} \cdot d\mu_{\mathrm{HS}}^{d,r}(\sigma).
\end{align}

We summarize this construction with the following definition. 

\begin{definition}
    The \emph{pretty good measurement over the rank-$r$ Hilbert--Schmidt measure} is a PGM with operators labeled by mixed states $\sigma \in \C^{d \times d}$. The state corresponding to $\sigma$ is $\sigma^{\otimes n}$, and the corresponding probability is $d\mu^{d,r}_{\mathrm{HS}}(\sigma)$. The resulting measurement operators are given by \Cref{eq:M_HS_in_Schur}. 
\end{definition}

\subsection{Viewing $\mixed(\gps)$ as a PGM} 

In this section, we show that the measurement performed by $\mixed(\gps)$ is equivalent to the PGM over the rank-$r$ Hilbert--Schmidt measure. This PGM is also the measurement performed by $\mixed(\hayashi)$, as $\mixed(\gps)$ and $\mixed(\hayashi)$ differ only in their post-processing of the measurement outcome. 

The algorithm $\mixed(\gps)$, described in \Cref{fig:gps-reduction-basic}, first applies $\purifychan^{d,r}$ to the input state, and then applies Hayashi's measurement. Recall that this is the POVM with measurement operators: 
\begin{equation*}
    \{ D[n] \cdot \ketbra{u}^{\otimes n} \cdot du \, : \, \ket{u} \in \C^{d} \otimes \C^r \}.
\end{equation*}
Finally, the algorithm proceeds by processing $\sigma_u = \tr_{\reg{2}}(\ketbra{u})$. Let $X \subseteq \C^{d\times d}$ be a subset of mixed states (measureable with respect to the Hilbert--Schmidt measure), and let $Y \subseteq \C^{D}$ be the preimage of $X$ under tracing out the second register, i.e.\ $X = \tr_{\reg{2}}(Y)$ (since $X$ is measurable, $Y$ is measurable with respect to the Haar measure). The probability that we obtain a $\bsigma \in X$ after purifying, measuring, and tracing out a given, generic, input $\psi\in \C^{d^n \times d^n}$ (i.e.\ not necessarily in the symmetric subspace) is 
\begin{equation*}
    \Pr_{\mixed(\gps)}[ \bsigma \in X | \psi ] =  \tr \Big( \purifychan^{d,r}( \psi ) \cdot \int_{\ket{u} \in Y} D[n] \cdot \ketbra{u}^{\otimes n}  \cdot du\Big).
\end{equation*}
On the other hand, the probability that our PGM obtains $\bsigma \in X$, given input $\psi$, is
\begin{equation*}
    \Pr_{\mathrm{PGM}} [ \bsigma \in X | \psi] = \int_{\sigma \in X} \tr\big( \psi \cdot M^{d,r}_{\mathrm{HS}}(\sigma) \big). 
\end{equation*}

\begin{proposition} \label{prop:base_implements_pgm}
    For any input $\psi$ and any measurable set $X \subseteq \C^{d\times d}$ (measured with $\mu_{\mathrm{HS}}^{d,r}$), 
    \begin{equation*}
        \Pr_{\mixed(\gps)}[ \bsigma \in X| \psi ] = \Pr_{\mathrm{PGM}} [ \bsigma \in X| \psi].
    \end{equation*}
    Thus, $\mixed(\gps)$ implements the pretty good measurement over the rank-$r$ Hilbert--Schmidt measure. 
\end{proposition}

Before proving this statement, it will be useful to note that we can decompose $\purifychan^{d,r}$ into Kraus operators $\{K_{\lambda S T}\}$ as $\purifychan^{d,r}(\psi) = \sum_{\lambda, S, T} K_{\lambda S T} \cdot \psi \cdot K_{\lambda S T}^\dagger$ for any generic state $\psi$ written in the Schur basis of $(\C^d)^{\otimes n}$, with
\begin{equation}
        K_{\lambda S T} \coloneq \ket{\lambda \lambda}_{\reg{YY'}}\bra{\lambda}_{\reg{Y}} \otimes \ket{\epr_\lambda}_{\reg{PP'}}\bra{S}_{\reg{P}} \otimes \frac{\ket{T}_{\reg{Q'}}}{\sqrt{\dim(V^r_\lambda)}} \otimes (I_{\dim(V^d_\lambda)})_{\reg{Q}} . \label{eq:kraus_purifychan}
\end{equation}
Here, $\lambda \vdash n$ and $\ell(\lambda) \leq r$; $S$ is an SYT of shape $\lambda$; $T$ is an SSYT of shape $\lambda$ and alphabet $[r]$. We can consider the adjoint $\purifychan^{d,r,\dagger}$, which can be expanded into Kraus operators as $\purifychan^{d,r,\dagger}(\varphi) = \sum_{\lambda,S,T} K_{\lambda S T}^\dagger \cdot \varphi \cdot K_{\lambda S T}$, where $\varphi$ is a state written in the Schur basis of $(\C^D)^{\otimes n}$.

\begin{lemma} \label{lem:purifychan_adjoint_on_n_fold_u}
    For any $\ket{u} \in \C^{d} \otimes \C^r$, we have
    \begin{equation*}\purifychan^{d,r,\dagger} \Big( (\schur^{\otimes 2}) \cdot \ketbra{u}^{\otimes n} \cdot (\schur^{\otimes 2})^\dagger \Big) = \sum_{\lambda \vdash n, \ell(\lambda)\leq r} \frac{\dim(\lambda)}{\dim(V^r_\lambda)} \cdot \ketbra{\lambda} \otimes I_{\dim(\lambda)} \otimes \nu^d_\lambda(\sigma_u).\end{equation*}
\end{lemma}

\begin{proof}
    By \Cref{lem:double_schur_transform_pure_state}, we have 
    \begin{equation*}
        (\schur^{\otimes 2}) \cdot \ketbra{u}^{\otimes n} \cdot (\schur^{\otimes 2})^{\dagger} = \sum_{\substack{\lambda \vdash n, \ell(\lambda) \leq r \\ \mu \vdash n, \ell(\mu) \leq r}} \ketbra{\lambda\lambda}{\mu\mu}_{\reg{YY'}} \otimes \ketbra{\epr_\lambda}{\epr_\mu}_{\reg{PP'}} \otimes \ketbra{u_{\lambda\lambda}}{u_{\mu \mu}}_{\reg{QQ'}},
    \end{equation*}
    for some vectors $\ket{u_{\lambda\lambda}}_{\reg{QQ'}}$ which satisfy $\tr_{\reg{Q'}}(\ketbra{u_{\lambda\lambda}}_{\reg{QQ'}}) = \dim(\lambda) \cdot \nu_\lambda^d(\sigma_u)$, with $\sigma_u = \tr_{\reg{2}}(\ketbra{u})$. 
We then have
    \begin{align*}
        & \purifychan^{d,r,\dagger}\Big( (\schur^d \otimes \schur^r) \cdot \ketbra{u}^{\otimes n} \cdot (\schur^d \otimes \schur^r)^\dagger \Big) \\
        & = \sum_{\lambda, S, T} K_{\lambda S T}^\dagger \cdot \Big( (\schur^d \otimes \schur^r) \cdot \ketbra{u}^{\otimes n} \cdot (\schur^d \otimes \schur^r)^\dagger \Big) \cdot K_{\lambda S T} \\
        & = \sum_{\lambda,S,T} K^\dagger_{\lambda S T} \cdot \Big( \sum_{\sigma, \tau} \ketbra{\sigma \sigma}{\tau \tau}_{\reg{YY'}} \otimes \ketbra{\epr_\sigma}{\epr_{\tau}}_{\reg{PP'}} \otimes \ketbra{u_{\sigma \sigma}}{{u_{\tau \tau}}}_{\reg{QQ'}} \Big) \cdot K_{\lambda S T} \\
        & = \sum_{\lambda, S,T} \frac{1}{\dim(V^r_\lambda)} \cdot \ketbra{\lambda}_{\reg{Y}} \otimes \ketbra{S}_{\reg{P}} \otimes \big(\bra{T}_{\reg{Q}'}\cdot \ketbra{u_{\lambda \lambda}}_{\reg{QQ'}} \cdot \ket{T}_{\reg{Q}'}\big) \tag{\Cref{eq:kraus_purifychan}} \\
        & = \sum_{\lambda} \frac{1}{\dim(V^r_\lambda)} \cdot \ketbra{\lambda} \otimes I_{\dim(\lambda)} \otimes \tr_{\reg{Q'}}(\ketbra{u_{\lambda \lambda}}_{\reg{QQ'}} )  \\
        & =  \sum_{\lambda} \frac{\dim(\lambda)}{\dim(V^r_\lambda)} \cdot \ketbra{\lambda} \otimes I_{\dim(\lambda)} \otimes \nu_\lambda^d(\sigma_u).
    \end{align*}
    This completes the proof. 
\end{proof}

\begin{proof}[Proof of \Cref{prop:base_implements_pgm}]
    Note that 
    \begin{align*}
        & \int_{\ket{u} \in Y} D[n] \cdot \purifychan^{d,r,\dagger}(\ketbra{u}^{\otimes n})  \cdot du \\
        & = (\schur^d)^\dagger \cdot \Big(\sum_{\lambda \vdash n, \ell(\lambda) \leq r} \frac{D[n] \cdot \dim(\lambda)}{\dim(V^r_\lambda)} \cdot \ketbra{\lambda} \otimes I_{\dim(\lambda)} \otimes \int_{\ket{u} \in Y} \nu^d_\lambda(\sigma_u) \cdot du \Big) \cdot \schur^d \tag{\Cref{lem:purifychan_adjoint_on_n_fold_u}}\\
        & = (\schur^d)^\dagger \cdot \Big(\sum_{\lambda \vdash n, \ell(\lambda) \leq r} \frac{D[n] \cdot \dim(\lambda)}{\dim(V^r_\lambda)} \cdot \ketbra{\lambda} \otimes I_{\dim(\lambda)} \otimes \int_{\sigma \in X} \nu^d_\lambda(\sigma) \cdot d\mu_{HS}^{d,r}(\sigma)  \Big) \cdot \schur^d \tag{\Cref{eq:pushfoward_change_of_variables}}\\
        & = \int_{\sigma \in X} (\schur^d)^\dagger \cdot \Big( \sum_{\lambda \vdash n, \ell(\lambda) \leq r} \frac{D[n] \cdot \dim(\lambda)}{\dim(V^r_\lambda)} \cdot \ketbra{\lambda} \otimes I_{\dim(\lambda)} \otimes \nu^d_\lambda(\sigma) \cdot d\mu_{\mathrm{HS}}^{d,r}(\sigma)\Big) \cdot \schur^d \\
        & = \int_{\sigma \in X} M^{d,r}_{\mathrm{HS}}(\sigma). \tag{\Cref{eq:M_HS_in_Schur}} 
    \end{align*}
    Thus,
    \begin{align*}
        \Pr_{\mixed(\gps)}[ \bsigma \in X| \psi ] & = \tr_{\reg{AB}} \Big( \purifychan^{d,r}( \psi ) \cdot \int_{\ket{u} \in Y} D[n] \cdot \ketbra{u}^{\otimes n}  \cdot du\Big) \\
        & = \tr_{\reg{A}} \Big( \psi \cdot \int_{\ket{u} \in Y} D[n] \cdot \purifychan^{d,r,\dagger}(\ketbra{u}^{\otimes n})  \cdot du\Big) \\
        & = \tr_{\reg{A}} \Big( \psi \cdot \int_{\sigma \in X} M^{d,r}_{\mathrm{HS}}(\bsigma) \Big) \\
        & = \Pr_{\mathrm{PGM}} [ \bsigma \in X | \psi]. \qedhere
    \end{align*}
\end{proof}

\paragraph{Viewing $\mixed^+(\gps)$ as a PGM.} Note that $\mixed^+(\gps)$, described in \Cref{fig:gps-reduction-improved}, can be regarded as the following two-step process: first, weak Schur sample $\rho^{\otimes n}$, to obtain $\blambda \vdash n$ and a state $\rho|_{\blambda}$;  second, apply $\mixed(\gps)$ to $\rho|_{\blambda}$ with $r$ set to $\ell(\blambda)$. As a consequence of \Cref{prop:base_implements_pgm}, we therefore have the following~result. 

\begin{corollary}
    $\mixed^+(\gps)$ implements weak Schur sampling, and, conditioned on the Young diagram $\blambda$ observed, followed by a pretty good measurement over the rank-$\ell(\blambda)$ Hilbert--Schmidt measure.
\end{corollary}

\end{document}